\documentclass[11pt]{article}

\usepackage{amsthm,amsmath,amssymb,bbm}
\usepackage{natbib}
\usepackage{multirow}
\usepackage[pdftex]{graphicx}
\usepackage{subfigure}
\usepackage{makecell}
\usepackage{booktabs}
\usepackage{array}
\usepackage{url}
\usepackage{algorithm}
\usepackage{algorithmic}
\usepackage{bm}
\usepackage{wrapfig}
\usepackage{lipsum}
\usepackage{mathrsfs}
\usepackage{dsfont}
\usepackage{titling}

\usepackage{multirow}

\usepackage{setspace}

\usepackage[usenames,dvipsnames,svgnames,table]{xcolor}
\usepackage[colorlinks,
linkcolor=red,
anchorcolor=blue,
citecolor=blue
]{hyperref}

\newcommand{\sgn}{\mathop{\mathrm{sgn}}}

\usepackage{smile}

\usepackage{xr}

\newcommand{\e}{\mathbb{E}}
\newcommand{\nn}{\nonumber}

\newcommand{\sn}{\sum_{i=1}^n}
\newcommand{\bdelta}{\bm{\delta}}

\newcommand{\FDP}{{\rm FDP}}
\newcommand{\AFDP}{{\rm AFDP}}

\def\##1\#{\begin{align}#1\end{align}}
\def\$#1\${\begin{align*}#1\end{align*}}

\newcommand {\cov}{\textnormal {cov}}
\newcommand {\var}{\textnormal {var}}

\def \T{\mathrm{\scriptstyle T}} 
\def\sn {\sum_{i=1}^n}
\def \wt {\widetilde}

\newcommand{\Rom}[1]{\text{\uppercase\expandafter{\romannumeral #1\relax}}}

\usepackage{geometry}
\usepackage{enumitem}
\usepackage{chngcntr}
\counterwithout{equation}{section}

\begin{document}


\title{FarmTest: Factor-adjusted robust multiple testing \\ with {approximate} false discovery control\thanks{Jianqing Fan is Honorary Professor, School of Data Science, Fudan University, Shanghai, China and Frederick L. Moore '18 Professor of Finance, Department of Operations Research and Financial Engineering, Princeton University, NJ 08544 (E-mail: jqfan@princeton.edu). Yuan Ke is Assistant Professor, Department of Statistics, University of Georgia, Athens, GA 30602 (E-mail: yuan.ke@uga.edu).  Qiang Sun is Assistant Professor, Department of Statistical Sciences, University of Toronto, Toronto, ON M5S 3G3, Canada (E-mail: qsun@utstat.toronto.edu).  Wen-Xin Zhou is Assistant Professor, Department of Mathematics, University of California, San Diego, La Jolla, CA 92093 (E-mail: wez243@ucsd.edu).  The bulk of the research were conducted while Yuan Ke, Qiang Sun and Wen-Xin Zhou were postdoctoral fellows at Department of Operations Research and Financial Engineering, Princeton University. This work is supported by NSERC Grant RGPIN-2018-06484, a Connaught Award, NSF Grants DMS-1662139, DMS-1712591, DMS-1811376, NIH Grant R01-GM072611, and NSFC Grant 11690014.}}

\author{Jianqing Fan, Yuan Ke, Qiang Sun, and Wen-Xin Zhou}


\date{}

\maketitle

\vspace{-0.5in}

\begin{abstract}
Large-scale multiple testing with correlated and heavy-tailed data arises in a wide range of research areas from genomics, medical imaging to finance. Conventional methods for estimating the false discovery proportion (FDP) often ignore the effect of heavy-tailedness and the dependence structure among test statistics, and thus may lead to inefficient or even inconsistent estimation.
Also, the commonly imposed  joint normality assumption is arguably too stringent for many applications.  To address these challenges, in this paper we propose a Factor-Adjusted Robust Multiple Testing ({\sl FarmTest}) procedure for large-scale simultaneous inference with control of the false discovery proportion. We demonstrate that robust factor adjustments are extremely important in both controlling the FDP and improving the power.  We identify  general conditions under which the proposed method produces consistent estimate of the FDP. As a byproduct that is of independent interest, we establish an exponential-type deviation inequality for a robust $U$-type covariance estimator under the spectral norm. Extensive numerical experiments demonstrate the advantage of the proposed method over several state-of-the-art methods especially when the data are generated from heavy-tailed distributions. The proposed procedures are implemented in the {\sf R}-package {\sf FarmTest}.
\end{abstract}


\section{Introduction}
\label{sec1}

Large-scale multiple testing problems with independent test statistics have been extensively explored and is now well understood  in both practice and theory \citep{BH1995, S2002, GW2004, LR2005}.
Yet, in practice, correlation effects often exist across many observed test statistics. For instance, in neuroscience studies, although the neuroimaging data may appear very high dimensional (with millions of voxels), the effect degrees of freedom are generally much smaller, due to spatial correlation and spatial continuity \citep{medland2014whole}.  In genomic studies, genes are usually correlated regulatorily or functionally: multiple genes may belong to the same regulatory pathway or there may exist gene-gene interactions. Ignoring these dependence structures will cause loss of statistical power or lead to inconsistent estimates.

To understand the effect of dependencies on multiple testing problems, validity of standard multiple testing procedures have been  studied under weak dependencies, see  \cite{BY2001}, \cite{S2003}, \cite{STS2004}, \cite{FZ2006},  \cite{C2007}, \cite{W2008}, \cite{CH2009}, \cite{BR2009} and \cite{LS2014}, among others. For example, it has been shown that, the Benjamini-Hochberg procedure or Storey's procedure, is still able to control the false discovery rate (FDR) or false discovery proportion, when only weak dependencies are present.
Nevertheless, multiple testing under general and strong dependence structures remains a challenge.
Directly applying standard FDR controlling procedures developed for independent test statistics in this case can lead to inaccurate false discovery control and spurious outcomes.
Therefore, correlations must be accounted for in the inference procedure; see, for example, \cite{O2005}, \cite{E2007, E2010}, \cite{LS2008}, \cite{SC2009}, \cite{FKC2009}, \cite{SL2011}, \cite{FHG2012}, \cite{DS2012}, \cite{WZHO2015} and \cite{FH2017} for an unavoidably incomplete overview.

In this paper, we  focus on the case where the dependence structure can be characterized by latent factors, that is, there exist a few unobserved variables that correlate with the outcome. A multi-factor model is an effective tool for modeling  dependence, with wide applications
in genomics \citep{KSZ2006}, neuroscience \citep{PW2007} and financial economics \citep{B2003}. It relies on the identification of a linear space of random vectors capturing the dependence structure of the data. In \cite{FKC2009} and \cite{DS2012}, the authors assumed a strict factor model with independent idiosyncratic errors, and used the EM algorithm to estimate the factor loadings as well as  the realized factors. The FDP is then estimated  by subtracting out the realized common factors. \cite{FHG2012} considered a general setting for estimating the FDP, where the test statistics follow a multivariate normal distribution with an arbitrary but known covariance structure. Later,  \cite{FH2017}  used the POET estimator \citep{FLM2013} to estimate the unknown covariance matrix, and then proposed a fully data-driven estimate of the FDP. Recently, \cite{WZHO2015} considered a more complex model with both observed primary variables and unobserved latent factors.

All the methods above  assume joint normality of factors and noise, and thus methods based on least squares regression, or likelihood generally, can be applied.
However, normality is really an idealization of the complex random world. For example, the distribution of the normalized gene expressions is often far from normal, regardless of the normalization methods used \citep{PH2005}.
Heavy-tailed data also frequently appear in many other scientific fields, such as financial engineering \citep{C2001} and biomedical imaging \citep{Eklund2016}.
In finance, the seminal papers by \cite{M1963} and \cite{F1963} discussed the power law behavior of asset returns, and \cite{C2001} provided extensive evidence of heavy-tailedness in financial returns. More recently, in functional MRI studies, it has been observed by \cite{Eklund2016} that the parametric statistical methods failed to produce valid  clusterwise inference, where the principal cause is  that the spatial autocorrelation functions do not follow the assumed Gaussian shape.  The heavy-tailedness issue may further be amplified by high dimensionality in large-scale inference.
 In the context of multiple testing, as the dimension gets larger, more outliers are likely to appear, and this may lead to significant false discoveries. It is therefore imperative to develop  inferential procedures that adjust dependence and are  robust to heavy-tailedness at the same time.

In this paper, we investigate the problem of large-scale multiple testing under dependence via an approximate factor model, where the outcome variables are correlated with each other through latent factors. To simultaneously incorporate the dependencies and tackle with heavy-tailed data, we propose a factor-adjusted robust multiple  testing (FarmTest) procedure. As we proceed, we gradually unveil the whole procedure in four steps.  First, we consider an oracle factor-adjusted procedure given the knowledge of the factors and loadings, which provides the key insights into the problem. Next, using the idea of adaptive Huber regression \citep{ZBFL2017, sun2016adaptive}, we consider estimating the realized factors when the loadings were known and provide a robust control of the FDP. In the third part, we propose two robust covariance matrix estimators, a $U$-statistic based estimator and another one  based on elementwise robustification. We then apply spectral decomposition to these estimators and use principal factors to recover the factor loadings. The final part, which is provided in Appendix~A, gives  a fully data-driven testing procedure based on sample splitting: use part of the data for loading construction and the other part for simultaneous inference.


 \begin{figure}[t!]
 \centering
 \includegraphics[width=5.2 in]{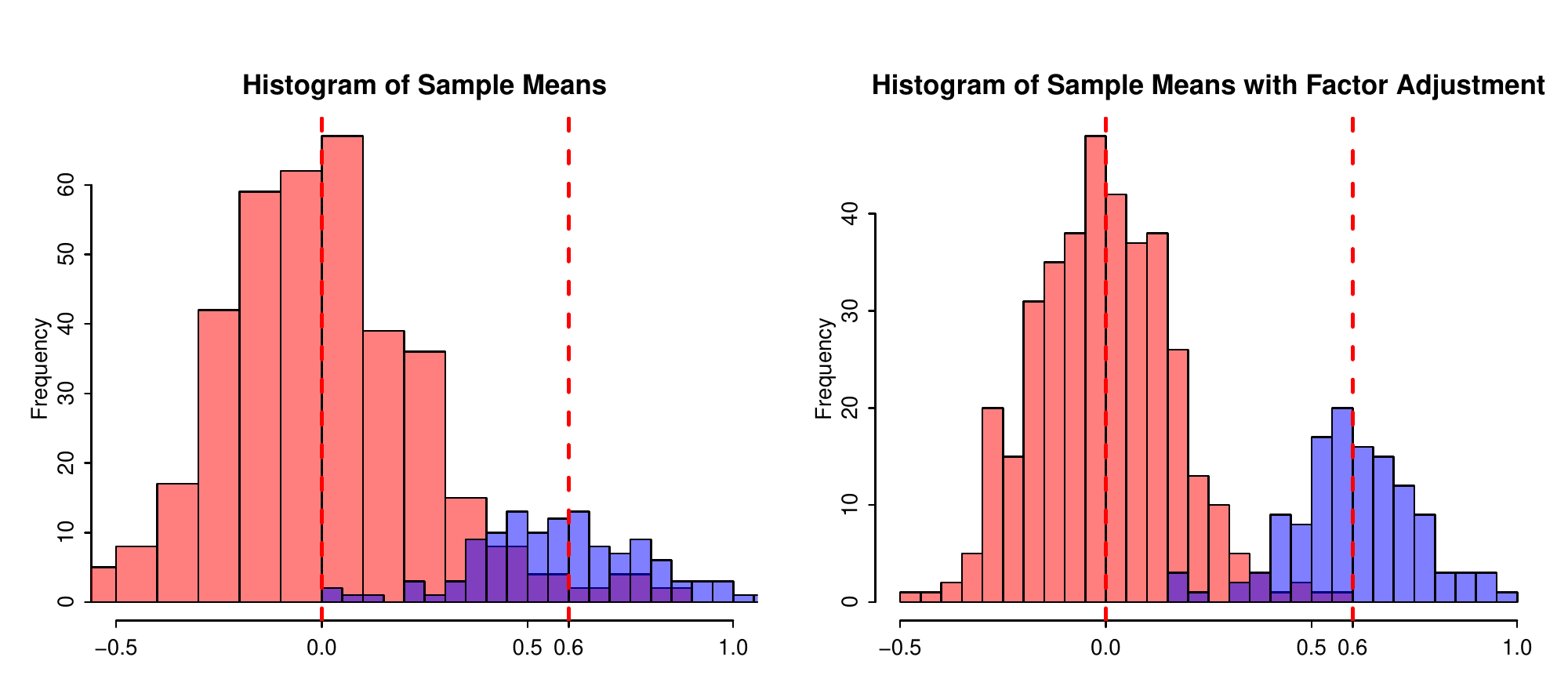}
 \includegraphics[width=5.2 in]{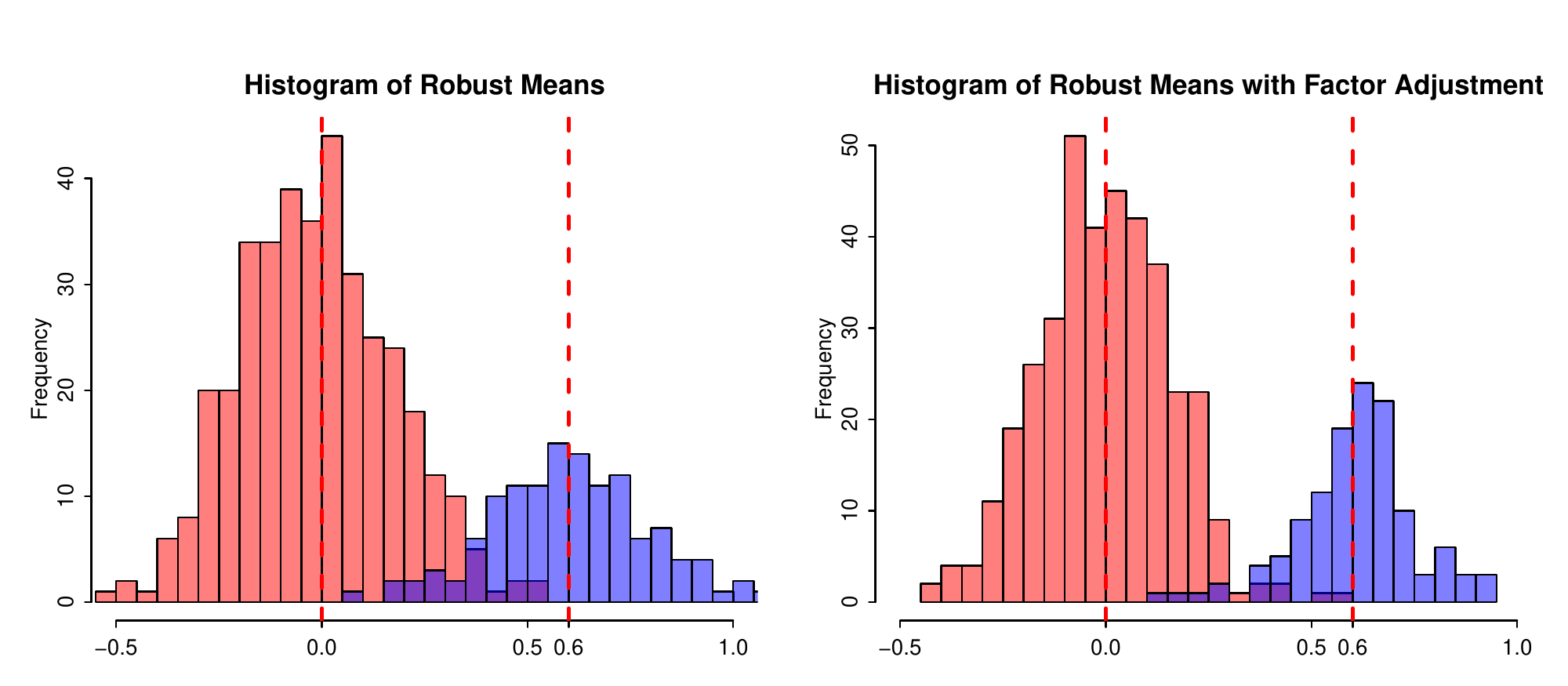}
  \caption{Histograms of four different mean estimators for simultaneous inference.}
 \label{Fig_mean}
\end{figure}

First we illustrate our methodology with a numerical example that consists of observations $\bX_i$'s generated from a three-factor model:
$$	\bX_i  = \bmu + \Bb \bbf_i + \bvarepsilon_i, ~~~ i=1, \ldots , n,  $$
where $\bbf_i \sim \cN( \textbf{0} , \Ib_3)$ and the entries of $\Bb$ are independent and identically distributed (IID) from a uniform distribution, $\mathcal{U}(-1, 1)$. The idiosyncratic errors, $\bvarepsilon_i$'s, are independently generated from the $t_3$-distribution with 3 degrees of freedom. The sample size $n$ and dimension $p$ are set to be $100$ and $500$, respectively. We take the true means to be $\mu_j=0.6$ for $1\leq j \leq 0.25\times p$ and $0$ otherwise. In Figure \ref{Fig_mean}, we plot the histograms of sample means, robust mean estimators, and their counterparts with factor-adjustment. Details of robust mean estimation and the related factor-adjusted procedure are specified in Sections~\ref{sec:method} and \ref{sec:3}. Due to the existence of latent factors and heavy-tailed errors, there is a large overlap between sample means from the null and alternative, which makes it difficult to distinguish them from each other. With the help of either robustification or factor-adjustment, the null and alternative are better separated as shown in the figure. Further, with both factor-adjustment and robustification, the resulting estimators are tightly concentrated around the true means so that the signals are evidently differentiated from the noise. This example demonstrates the effectiveness of the factor-adjusted robust multiple testing procedure.

The rest of the paper proceeds as follows. In Section~\ref{sec:method}, we describe a generic factor-adjusted robust multiple testing procedure under the approximate factor model. In Section \ref{sec:3}, we gradually unfold the proposed method, while we establish its theoretical properties along the way. Section \ref{sec:4} is devoted to simulated numerical studies. Section \ref{sec:realdata} analyzes an empirical  dataset. We conclude the paper in Section \ref{sec:discuss}. Proofs of the main theorems and technical lemmas are provided in the online supplement.

\noindent
{\sc Notation}. We  adopt the following notations throughout the paper. 
For any $d\times d$ matrix $\Ab= (A_{k\ell})_{1\leq k,\ell\leq d}$, we write $\| \Ab \|_{\max} = \max_{1\leq k,\ell\leq d} |A_{k \ell}|$, $ \| \Ab \|_{1} =  \max_{1\leq \ell \leq d} \sum_{k=1}^d |A_{k\ell}|$ and $\| \Ab \|_{\infty} =  \max_{1\leq k \leq d} \sum_{\ell=1}^d |A_{k\ell}|$. Moreover, we use $\| \Ab \|$ and ${\rm tr}(\Ab)= \sum_{k=1}^d A_{kk}$ to denote the spectral norm and the trace of $\Ab$.
When $\Ab$ is symmetric, we have $\| \Ab \|  =  \max_{1\leq k\leq d} | \lambda_k(\Ab)|$, where $\lambda_1(\Ab)\geq \lambda_2(\Ab)\geq \cdots \geq\lambda_d(\Ab)$ are the eigenvalues of $\Ab$, and it holds
$\| \Ab \|  \leq  \| \Ab \|_{1}^{1/2} \| \Ab \|_{\infty }^{1/2}  \leq  \max\{\| \Ab \|_{1} , \| \Ab \|_{\infty }\} \leq  d^{1/2} \| \Ab \|$.
We use $\lambda_{\max}(\Ab)$ and $\lambda_{\min}(\Ab)$ to denote the maximum and minimum eigenvalues of $\Ab$, respectively.
\section{FarmTest}
\label{sec:method}

In this section, we describe a generic factor-adjusted robust multiple testing procedure under the approximate factor model.

\subsection{Problem setup}
\label{sec:2.1}

Let $\bX= (X_1, \ldots, X_p)^\T$ be a $p$-dimensional random vector with mean $\bmu = (\mu_1, \ldots, \mu_p)^\T$ and covariance matrix $\bSigma = (\sigma_{jk})_{1\leq  j, k\leq p}$.  We assume the dependence structure in $\bX$ is captured by a few latent factors such that $\bX = \bmu + \Bb \bbf + \bvarepsilon$, where $\Bb = (\bb_1, \ldots, \bb_p)^\T \in \RR^{p\times K}$ is the  deterministic factor loading matrix, $\bbf=(f_{i1},\ldots, f_{iK})^\T  \in \RR^K$ is  the zero-mean latent random factor, and $\bvarepsilon = (\varepsilon_{1} , \ldots, \varepsilon_{ p})^\T \in \RR^p$  consists of  idiosyncratic errors that are uncorrelated with $\bbf$.  Suppose we observe $n$  random samples $\bX_{1}, \ldots, \bX_n$ from $\bX$, satisfying
\#\label{obs}
	\bX_i   = \bmu + \Bb \bbf_i + \bvarepsilon_i, ~~~ i=1, \ldots , n,
\#
where $\bbf_i$'s and $\bvarepsilon_i$'s are IID samples of $\bbf$ and $\bvarepsilon$, respectively. Assume that $\bbf$ and $\bvarepsilon$ have covariance matrices $\bSigma_f$ and $\bSigma_\varepsilon=(\sigma_{\varepsilon, jk})_{1\leq j, k\leq p}$.
{In addition, note that $\Bb$ and $\bbf_i$ are not separately identifiable as they both are unobserved. For an arbitrary $K \times K$ invertible matrix $\Hb$, one can choose $\Bb^*=\Bb \Hb$ and $\bbf_i^*=\Hb^{-1}\bbf_i$ such that $\Bb^* \bbf_i^*=\Bb \bbf_i$. Since $\Hb$ contains $K^2$ free parameters, we impose the following conditions to make $\Bb$ and $\bbf$ identifiable:
\#  \label{id.ass}
	\bSigma_f = \Ib_K ~~~ \textnormal{ and } ~~~  \Bb^\T \Bb  ~\textnormal{is diagonal},
\#
where the two conditions provide $K(K+1)/2$ and $K(K-1)/2$ restrictions, respectively. The choice of identification conditions is not unique. We refer to \cite{LM1971} and \cite{BL2012} for details of more identification strategies. Model \eqref{obs} with observable factors has no identification issue and is studied elsewhere \citep{ZBFL2017}.}

In this paper, we are interested in simultaneously testing the following hypotheses
\#
	H_{0j} : \mu_j =  0 \ \ \mbox{ versus } \ \  H_{1j} : \mu_j \neq 0 , \ \ \mbox{ for } 1\leq j\leq p, \label{mht}
\#
based on the observed data $\{ \bX_i \}_{i=1}^n$. Many existing works \citep[e.g.][]{FKC2009, FHG2012, FH2017} in the literature  assume multivariate normality of the idiosyncratic errors.     However,  
 the Gaussian assumption on the sampling distribution is often unrealistic in many practical applications. For each feature, the measurements across different subjects consist of samples from potentially different distributions with quite different scales, and thus can be highly skewed and heavy-tailed. In the big data regime, we are often dealing with thousands or tens of thousands of features simultaneously. Simply by chance, some variables exhibit heavy and/or asymmetric tails. As a consequence, with the number of variables grows, some outliers may turn out to be so dominant that they can be mistakenly regarded as discoveries. Therefore, it is imperative to develop robust alternatives that are insensitive to outliers and data contaminations. 

For each $1\leq j\leq p$, let $T_j$ be a generic test statistic for testing the individual hypothesis $H_{0j}$. For a prespecified thresholding level $z>0$, we reject the $j$-th null hypothesis whenever $|T_j| \geq z$. The number of total discoveries $R(z)$ and the number of false discoveries $V(z)$ can be written as
\#
	R(z) = \sum_{j=1}^p  I( |T_j | \geq z )  \ \ \mbox{ and } \ \ V(z) = \sum_{j\in \mathcal{H}_0}  I( |T_j| \geq z ), \label{R.V.def}
\#
respectively, where $\mathcal{H}_0 := \{ j : 1\leq  j\leq p, \mu_j = 0\}$ is the set of the true nulls with cardinality $p_0=| \mathcal{H}_0| =\sum_{j=1}^p I(\mu_j=0)$. We are mainly interested in controlling the false discovery proportion, $\FDP(z) = V(z)/ R(z)$ with the convention $0/0 = 0$. We remark here that $R(z)$ is observable given the data, while $V(z)$, which depends on the set of true nulls, is an unobserved random quantity that needs to be estimated. Comparing with FDR control, controlling FDP is arguably more relevant as it is directly related to the current experiment.

\subsection{A generic procedure}
\label{sec:2.2}

We now propose a \underline{F}actor-\underline{A}djusted \underline{R}obust \underline{M}ultiple \underline{Test}ing procedure, which we call FarmTest. As the name suggests, this procedure utilizes the dependence structure in $\bX$ and is robust against heavy tailedness of the error distributions.
Recent studies in \cite{FLW2017} and \cite{ZBFL2017} show that the Huber estimator \citep{Huber1964} with a properly diverging robustification parameter admits a sub-Gaussian-type deviation bound  for heavy-tailed data under mild moment conditions. This new perspective motivates new methods, as described below. To begin with, we revisit the Huber loss and  the robustification parameter. 
\begin{definition} \label{Huber.def} 
The Huber loss $\ell_\tau(\cdot)$ \citep{Huber1964} is defined as
$$
	\ell_\tau(u) =
	\begin{cases}
	 u^2 /2 ,    & \mbox{if } | u | \leq \tau ,  \\
	\tau | u | -   \tau^2 /2 ,   &  \mbox{if }  | u | > \tau ,
	\end{cases}
$$
where $\tau>0$ is refereed to as the {\it robustification parameter}  that trades bias for robustness.
\end{definition}

\begin{figure}        \centering
                \includegraphics[scale=0.4]{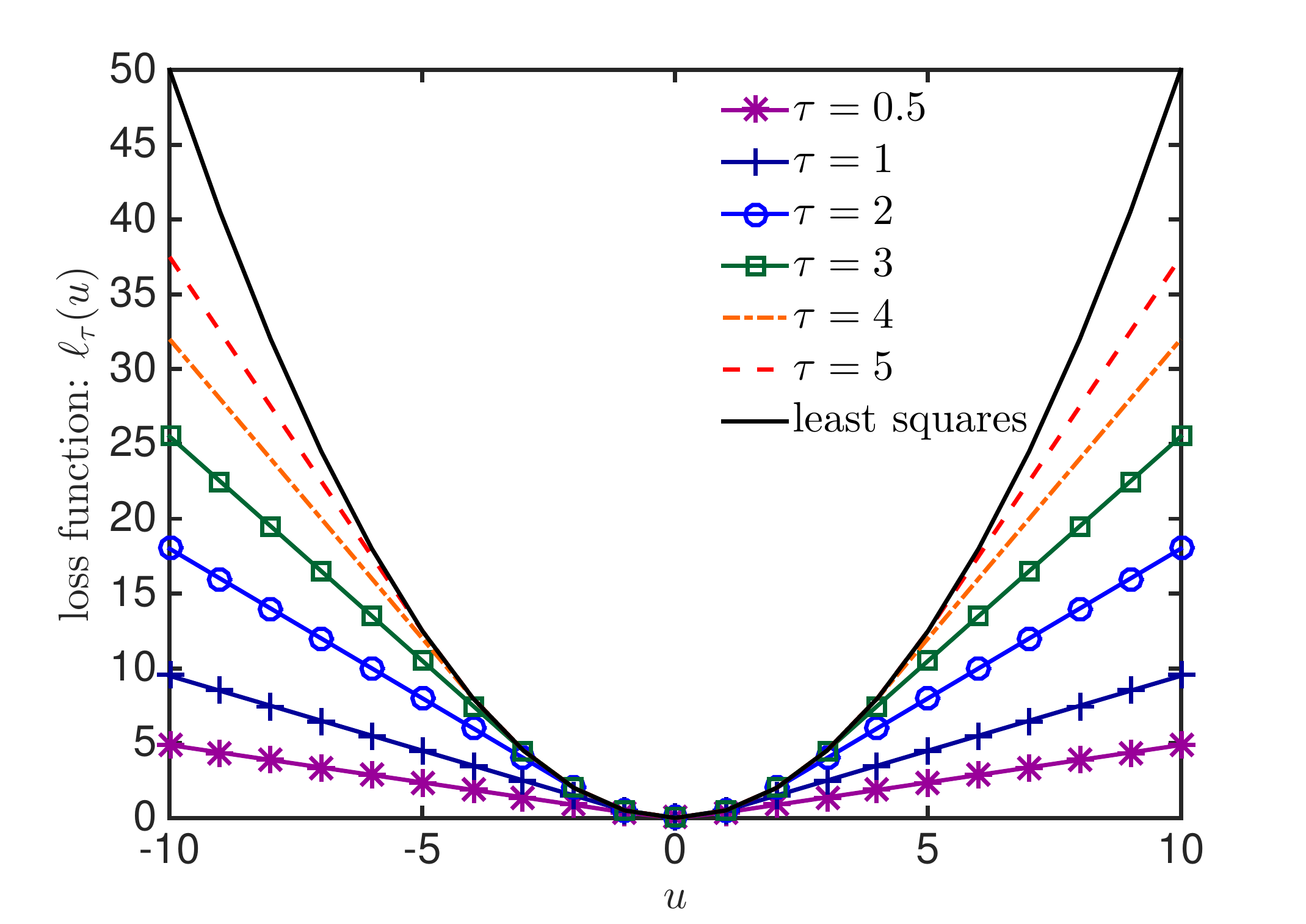}
                \vspace{-0.5cm}
                \caption{The Huber loss function $\ell_\tau(\cdot)$ with varying robustification parameters and the quadratic loss function.}
                \label{fig:histogram}
\end{figure}

We refer to the Huber loss in Definition \ref{Huber.def} above as the adaptive Huber loss to recognize the adaptivity of the robustification parameter $\tau$. For any $1\leq j\leq p$, with a  robustification parameter $\tau_j>0$, we consider the following $M$-estimator of $\mu_j$:
\#\label{rest}
\hat{\mu}_j  = \arg\min_{\theta \in \RR } \sn \ell_{\tau_j} (X_{ij} - \theta ),
\#
where we suppress the dependence of $\widehat\mu_j$ on $\tau_j$ for simplicity. As shown in our theoretical results, the parameter $\tau$ plays an important role in controlling the bias-robustness tradeoff. To guarantee the asymptotic normality of $\hat{\mu}_j$ uniformly over $j=1,\ldots p$, and to achieve optimal bias-robustness tradeoff, we choose $\tau=\tau(n,p)$ of the form $C\sqrt{{n}/{\log(n p)}}$, where the constant $C>0$ can be selected via cross-validation. We refer to Section \ref{sec:4.1} for details. Specifically, we show that  $\sqrt{n}\, (\widehat \mu_j-\bb_j^\T\bar\bbf)$ is asymptotically normal with mean $\mu_j$ and variance $\sigma_{\varepsilon, jj}$ (with details given in Appendix~B):
\begin{align}
	\sqrt{n} \, (\hat{\mu}_j - \mu_j - \bb^\T_j\bar\bbf   ) & = \cN ( 0 , \sigma_{\varepsilon, jj}  )+ o_\PP(1) ~ \mbox{ uniformly over } j =1,\ldots, p. \label{stat.dec1}
\end{align}
Here, $\hat{\mu}_j$'s can be regarded as robust versions of the sample averages $\bar{X}_j = \mu_j + \bb_j^\T \bar{\bbf} + \bar{\varepsilon}_j$, where $\bar{X}_j = n^{-1} \sn X_{ij}$ and $\bar{\varepsilon}_j = n^{-1} \sn \varepsilon_{ij}$.

Given a prespecified level $\alpha \in (0,1)$, our testing procedure consists of three steps: (i) robust estimation of the loading vectors and factors; (ii) construction of factor-adjusted marginal test statistics and their $P$-values; and (iii) computing the critical value or threshold level with the estimated FDP controlled at $\alpha$. The detailed procedure is stated below.

\begin{algorithm}
{\sc FarmTest Procedure}.

\noindent
{\bf Input}: Observed data $\bX_i = (X_{i1} ,\ldots, X_{ip})^\T \in \RR^p$ for $i=1,\ldots, n$, a prespecified level $\alpha \in (0,1)$ and an integer $K\geq 1$.

\noindent
{\bf Procedure}:

\noindent
{\sc Step~1}: Construct a robust covariance matrix estimator $\hat{\bSigma}
\in \RR^{p\times p}$  based on observed data. Let $ \hat{\lambda}_1 \geq  \hat{\lambda}_2  \geq \cdots \geq   \hat{\lambda}_K$ be the top $K$ eigenvalues of $\hat{\bSigma}$, and $ \hat{\bv}_{ 1},  \hat{\bv}_{ 2}, \ldots, \hat{\bv}_K$ be the corresponding eigenvectors. Define $\hat{\Bb} =   (  \wt{\lambda}_1^{1/2} \hat{\bv}_1 , \ldots, \wt{\lambda}_K^{1/2} \hat{\bv}_K  )  \in \RR^{p\times K}$, where $\wt{\lambda}_k = \max(\hat{\lambda}_k, 0)$. Let $\hat{\bb}_1, \ldots, \hat{\bb}_p \in \RR^K$ be the $p$ rows of $\hat{\Bb}$, and define
\#
	\hat{\bbf} \in \arg\min_{\bbf \in \RR^K } \sum_{j=1}^p \ell_\gamma( \bar{X}_j - \hat{\bb}_j^\T \bbf ), \label{factor.est}
\#
where $\gamma = \gamma(n,p)>0$ is a robustification parameter.

\noindent
{\sc Step~2}:  Construct factor-adjusted test statistics
\#
	 {T}_j = \sqrt{  \frac{n}{\hat{\sigma}_{\varepsilon, jj} } } \, (  \hat{\mu}_j - \hat{\bb}_j^\T \hat{\bbf} \, ), \ \ j=1,\ldots, p, \label{hat.Tj}
\#
where $\hat{\sigma}_{\varepsilon, jj} =  \hat{\theta}_j - \hat{\mu}_j^2 - \| \hat{\bb}_j \|_2^2$, $ \hat{\theta}_j = \arg\min_{\theta \geq \hat{\mu}_j^2 + \| \hat{\bb}_j \|_2^2 } \sn \ell_{ \tau_{jj} } (X_{ij}^2 - \theta )$, $\tau_{jj} $'s are robustification parameters and $\widehat\mu_j$'s are defined in \eqref{rest}.  Here, we use the fact that $\EE( X_j^2 ) = \mu_j^2 + \| \bb_j \|_2^2 + \var(\varepsilon_j)$, according to the identification condition.

\noindent
{\sc Step~3}: Calculate the critical value $z_\alpha$ as
\#
	z_\alpha = \inf \{ z\geq 0 : {\FDP}^{{\rm A} }(z  ) \leq \alpha \} , \label{def.zalpha}
\#
where ${\FDP}^{{\rm A} }(z  ) =    2 p \Phi(-z) /  {R}(z )$ denotes the approximate FDP and $R(z)$ is as in \eqref{R.V.def}. Finally, for $j=1,\ldots, p$, reject $H_{0j}$ whenever $|T_j | \geq z_\alpha$.

\end{algorithm}

We expect that the factor-adjusted test statistic $T_j$ given in \eqref{hat.Tj} is close in distribution to standard normal for all $j=1,\ldots, p$. Hence, according to the law of large numbers, the number of false discoveries $V(z) = \sum_{j\in \mathcal{H}_0} I(|T_j| \geq z)$ should be close to $2 p_0 \Phi(-z)$ for any $z\geq0$. The number of null hypotheses $p_0$ is typically unknown. In the high dimensional and sparse regime, where both $p$ and $p_0$ are large and $p_1 = p-p_0 = o(p)$ is relatively small,  $ \FDP^{{\rm A} }$ in \eqref{def.zalpha} serves as a slightly conservative surrogate for the asymptotic approximation $2 p_0 \Phi(-z) /  {R}(z )$. If the proportion $ \pi_0 = p_0 / p$ is bounded away from 1 as $p \to \infty$, $ \FDP^{{\rm A} }$ tends to overestimate the true FDP. The estimation of $\pi_0$ has long been known as an interesting problem. See, for example, \cite{S2002}, \cite{LL2005}, \cite{MR2006}, \cite{JC2007} and \cite{J2008}, among others. Therefore, a more adaptive method is to combine the above procedure with, for example Storey's approach, to calibrate the rejection region for individual hypotheses. Let $\{  P_j = 2\Phi(-| T_j | ) \}_{j=1}^p$ be the approximate $P$-values. For a predetermined $\eta \in [0,1)$, \cite{S2002} suggested to estimate $\pi_0$ by
\#
	\hat{\pi}_0(\eta) = \frac{1}{(1-\eta) p} \sum_{j=1}^p  I( P_j > \eta ). \label{ratio.est}
\#
The fundamental principle that underpins Storey's procedure is that most of the large $P$-values come from the true null hypotheses and thus are uniformly distributed.  For a sufficiently large $\eta$, about $(1-\eta) \pi_0$ of the $P$-values are expected to lie in $(\eta , 1]$. Therefore, the proportion of $P$-values that exceed $\eta$ should be close to $(1-\eta) \pi_0$. A value of $\eta=1/2$ is used in the SAM software \citep{ST2003}; while it was shown in \cite{BR2009} that the choice $\eta = \alpha$ may have better properties for dependent $P$-values.

Incorporating the above estimate of $\pi_0$, a modified estimate of $\FDP$  takes the form
\#
 \FDP^{ {\rm A} }(z ;  \eta ) = 2  p \,  \hat{\pi}_0(\eta)  \Phi(-z) /  R(z )   , \, z \geq 0 . \nn
\#
Finally, for any prespecified $\alpha \in (0,1)$, we reject $H_{0j}$ whenever $|  T_j | \geq  {z}_{\alpha, \eta}$, where
\begin{align}
	 {z}_{\alpha, \eta} = \inf \{ z \geq  0: \FDP^{ {\rm A} }(z ; \eta) \leq \alpha \}. \label{adaptive.zalpha}
\end{align}
By definition, it is easy to see that $ {z}_{\alpha, 0}$ coincides with $z_\alpha$ given in \eqref{def.zalpha}.

\section{Theoretical properties}
\label{sec:3}

To fully understand the impact of factor-adjustment as well as robust estimation, we successively investigate the theoretical properties of the FarmTest through several steps, starting with an oracle procedure that provides key insights into the problem.

\subsection{An oracle procedure}
\label{sec:3.1}

First we consider an oracle procedure that serves as a heuristic device. In this section, we assume the loading matrix $\Bb$ is known and the factors $\{ \bbf_i \}_{i=1}^n$ are observable. In this case, it is natural to use the factor-adjusted data:  $\bY_i = \bX_i  -  \Bb \bbf_i = \bmu +  \bvarepsilon_i$, which has smaller componentwise variances (which are  $\{ \sigma_{\varepsilon, jj} \}_{j=1}^p$ and assumed known for the moment) than those of $\bX_i$.  Thus, instead of using $\sqrt{n} \, \hat{\mu}_j$ given in \eqref{rest}, it is more efficient to construct robust mean estimates using factor-adjusted data.  This is essentially the same as using the test statistic
\begin{align}
T_j^\circ  = \sqrt{ \frac{n}{\sigma_{\varepsilon ,jj}}  }  ( \hat \mu_j - \bb_j^\T \bar{\bbf} \, ), \label{oracle.stat}
\end{align}
whose distribution is close to the standard normal distribution under the $j$-th null hypothesis.
Recall that $p_0 = |\mathcal{H}_0|$ is the number of true null hypotheses. Then, for any $z \geq  0$,
\$
	\frac{1}{p_0} V( z )  =  \frac{1}{p_0}\sum_{j \in \mathcal{H}_0}  I( |T_j^\circ | \geq z ).
\$
Intuitively,  the (conditional) law of large numbers suggests that $p_0^{-1} V(z)  =   2 \Phi(-z) + o_{\PP}(1)$.
Hence, the FDP based on oracle test statistics admits an asymptotic expression
\# \label{OFDP.def}
	{\rm AFDP}_{\rm orc}(z) =  2p_0\Phi(-z)  / R(z) , \, z \geq 0,
\#
where ``AFDP" stands for the asymptotic FDP and a subscript ``orc'' is added to highlight its role as an oracle.

\begin{remark}
 For testing the individual hypothesis $H_{0j}$, \cite{FH2017} considered the test statistic $\sqrt{n} \bar{X}_j$, where $\bar{X}_j = (1/n) \sn X_{ij}$. The empirical means, without factor adjustments, are inefficient as elucidated in Section~\ref{sec1}.  In addition, they are sensitive to the tails of error distributions \citep{C2012}. In fact, with many collected variables, by chance only, some test statistics $\sqrt{n} \bar{X}_j$ can be so large in magnitude empirically that they may be mistakenly regarded as discoveries.
\end{remark}

We will show that  $\AFDP_{\rm orc}(z)$ provides a valid asymptotic approximation of the (unknown) true $\FDP$ using oracle statistics $\{T_j^\circ\}$ in high dimensions. The latter will be denoted as $\FDP_{{\rm orc}}(z)$. Let $\Rb_\varepsilon = (r_{\varepsilon , j k})_{1\leq j,k\leq p}$ be the correlation matrix of $\bvarepsilon = ( \varepsilon_1,\ldots, \varepsilon_p)^\T$, that is, $\Rb_\varepsilon = \Db_\varepsilon^{-1} \bSigma_\varepsilon \Db_\varepsilon^{-1}$ where $\Db_\varepsilon^2 = {\rm diag}(\sigma_{\varepsilon, 11}, \ldots, \sigma_{\varepsilon, pp})$. Moreover, write
\#
	\omega_{n,p} = \sqrt{n/ \log(np)}. \label{wnp.def}
\#
We impose the following moment and regularity assumptions.

\begin{assumption} \label{cond.corr}
{\rm
(i) $p=p(n) \to \infty$ and $\log(p) = o( \sqrt{n} )$ as $n\to \infty$; (ii) $\bX \in \RR^p$ follows the approximate factor model \eqref{obs} with $\bbf$ and $\bvarepsilon$ being independent; (iii) $\EE( \bbf )= {\mathbf 0}$, $\cov(\bbf) = \Ib_K$ and $\| \bbf \|_{\psi_2} \leq A_f$ for some $A_f>0$, where $ \| \cdot \|_{\psi_2}$ denotes the vector sub-Gaussian norm \citep{V2018}; (iv) There exist constants $C_\varepsilon , c_\varepsilon >0$ such that $c_\varepsilon \leq \min_{1\leq j\leq p} \sigma_{\varepsilon , jj}^{1/2} \leq \max_{1\leq j\leq p} \upsilon_j \leq C_\varepsilon $, where $\upsilon_j :=  ( \EE   \varepsilon_j^4  )^{1/4}$; (v) There exist constants $\kappa_0 \in (0,1)$ and $\kappa_1>0$ such that $\max_{1\leq j,k\leq p} |r_{\varepsilon ,jk}| \leq \kappa_0$ and $p^{-2} \sum_{1\leq j, k \leq p} |r_{\varepsilon ,jk}|  =O(  p^{-\kappa_1})$ as $p\to \infty$. }
\end{assumption}

Part (iii) of Assumption~\ref{cond.corr} requires $\bbf \in \RR^K$ to be a sub-Gaussian random vector. Typical examples include: (1) Gaussian and
Bernoulli random vectors, (2) random vector that is uniformly distributed on the Euclidean sphere in $\RR^K$ with center at the origin and radius $\sqrt{K}$, (3) random vector that is uniformly distributed on the Euclidean ball centered at the origin
with radius $\sqrt{K}$, and (4) random vector that is uniformly distributed on
the unit cube $[-1,1]^K$. In all these cases, the constant $A_f$ is a dimension-free constant. See Section 3.4 in \cite{V2018} for detailed discussions of multivariate sub-Gaussian distributions.
Part (v) is a technical condition on the covariance structure that allows $\varepsilon_1,\ldots, \varepsilon_p$ to be weakly dependent. It relaxes the sparsity condition on the off-diagonal entries of $\bSigma_{\varepsilon}$.

\begin{theorem} \label{thm1}
Suppose that Assumption~\ref{cond.corr} holds and $p_0 \geq a p$ for some constant $a \in (0,1)$. Let $\tau_j = a_j \omega_{n,p}$ with $a_j \geq \sigma_{jj}^{1/2}$ for $j=1,\ldots, p$, where $\omega_{n,p}$ is given by \eqref{wnp.def}. Then we have
\#
p_0^{-1} V(z) &= 2 \Phi(-z) + o_{\PP}(1) \label{V.LLN}  \\
 p^{-1} R(z) &=   \frac{1}{p}  \sum_{j=1}^p  \bigg\{ \Phi\bigg( -z+  \frac{  \sqrt{n}  \mu_j}{ \sqrt{ \sigma_{\varepsilon, jj} } }  \bigg)  +     \Phi\bigg(-z - \frac{\sqrt{n} \mu_j}{\sqrt{ \sigma_{\varepsilon, jj} }}  \bigg) \bigg\}   +o_{\PP}(1)   \label{R.LLN}
\#
uniformly over $z\geq 0$ as $n, p \to \infty$. Consequently, for any $z\geq 0$,
\begin{align}
	 |  \FDP_{{\rm orc}}(z) -  \AFDP_{\rm orc}(z)   |  = o_{\PP}(1) ~\mbox{ as } n, p \to \infty. \nn
\end{align}
\end{theorem}

\subsection{Robust estimation of loading matrix}
\label{sec:3.2}

{To realize the oracle procedure in practice, we need to estimate the loading matrix $\Bb$ and the covariance matrix $\bSigma$, especially its diagonal entries. Before proceeding, we first investigate how these preliminary estimates affect FDP estimation. Assume at the moment that $\bar{\bbf}$ is given, let $\wt \bb_1, \ldots, \wt \bb_p$ and $\wt \sigma_{11} , \ldots, \wt \sigma_{pp}$ be generic estimates of $\bb_1,\ldots, \bb_p$ and $\sigma_{11}, \ldots, \sigma_{pp}$, respectively. In view of \eqref{id.ass}, $\sigma_{\varepsilon, jj}$ can be naturally estimated by $ \wt \sigma_{jj} - \| \wt \bb_j \|_2^2$. The corresponding FDP and its asymptotic approximation are given by
\$
	\wt \FDP(z) =  {\wt V(z)}/{ \wt R(z)}  ~\mbox{ and }~ \wt \AFDP(z)  = {2p_0 \Phi(-z)}/{ \wt R(z)} ,   \ \ z\geq 0,
\$
where $\wt V(z) = \sum_{j\in \cH_0} I(|\wt T_j| \geq z)$, $\wt R(z) = \sum_{j=1}^p I(|\wt T_j| \geq z)$ and $\wt T_j = (n/\wt \sigma_{\varepsilon, jj})^{1/2} (\hat{\mu}_j - \wt \bb_j^\T \bar{\bbf})$ for $j=1,\ldots,p$. The following proposition shows that to ensure consistent FDP approximation or furthermore estimation, it suffices to establish the uniform convergence results in \eqref{matrix.est.error} for the preliminary estimators of $\Bb$ and $\{ \sigma_{jj}\}_{j=1}^p$. Later in Section~\ref{sec:3.3.1} and \ref{sec:3.3.2}, we propose two types of robust estimators satisfying \eqref{matrix.est.error} when $p=p(n)$ is allowed to grow exponentially fast with $n$.
\begin{proposition} \label{prop.obs}
Assume the conditions of Theorem~\ref{thm1} hold and that the preliminary estimates $\{ \wt \bb_j, \wt \sigma_{jj} \}_{j=1}^p$ satisfy
\#
	\max_{1\leq j\leq p} \| \wt \bb_j - \bb_j \|_2 = o_{\PP} \{ (\log n)^{-1/2} \}, \quad  	\max_{1\leq j\leq p} | \wt \sigma_{jj} - \sigma_{jj} | = o_{\PP} \{ (\log n)^{-1/2} \}.  \label{matrix.est.error}
\#
Then, for any $z\geq 0$, $| \wt \FDP(z) - \wt \AFDP(z)| = o_{\PP}(1)$ as $n,p\to\infty$.
\end{proposition}}

Next we focus on estimating $\Bb$ under identification condition \eqref{id.ass}. Write $\Bb=(\bar\bb_1, \ldots, \bar \bb_K)$ and assume without loss of generality that $\bar\bb_1, \ldots, \bar \bb_K \in \RR^p$ are ordered such that $\{ \| \bar \bb_\ell \|_2 \}_{\ell =1}^K$ is in a non-increasing order. In this notation, we have $\bSigma =  \sum_{\ell=1}^K \bar{\bb}_\ell \bar{\bb}_\ell^\T + \bSigma_\varepsilon$, and $\bar{\bb}_{\ell_1}^\T \bar{\bb}_{\ell_2} =0$ for $1\leq \ell_1 \neq \ell_2 \leq K$. Let $\lambda_1 , \ldots, \lambda_p$ be the eigenvalues of $\bSigma$ in a descending order, with associated eigenvectors denoted by $\bv_1,\ldots, \bv_p \in \RR^p$. By Weyl's theorem,
\begin{align}
	 | \lambda_j - \| \bar{\bb}_j  \|_2^2 | \leq \| \bSigma_\varepsilon \| ~\mbox{ for } 1\leq j\leq K ~\mbox{ and }~ | \lambda_j | \leq \| \bSigma_\varepsilon \| ~\mbox{ for } j >K. \nn
\end{align}
Moreover, under the pervasiveness condition (see Assumption~\ref{cond.pervasive} below), the eigenvectors $\bv_j$ and $\bar{\bb}_j/ \| \bar{\bb}_j \|_2$ of $\bSigma$  and $\Bb \Bb^\T$, respectively, are close to each other for $1\leq j\leq K$. The estimation of $\Bb$ thus depends heavily on estimating $\bSigma$ along with its eigenstructure.

In Sections~\ref{sec:3.3.1} and \ref{sec:3.3.2}, we propose two different robust covariance matrix estimators that are also of independent interest. The construction of $\hat{\Bb}$ then follows from Step 1 of the FarmTest procedure described in Section~\ref{sec:2.2}.

\subsubsection{$U$-type covariance estimation}
\label{sec:3.3.1}

First, we propose a $U$-type covariance matrix estimator, which leads to estimates of  the unobserved factors under condition \eqref{id.ass}. 
Let $\psi_\tau(\cdot)$ be the derivative of $\ell_\tau(\cdot)$ given by
\$
\psi_\tau(u) = \min(|u|, \tau ) \sign (u), \ \ u\in \RR.
\$
Given $n$ real-valued random variables $X_1,\ldots, X_n$ from $X$ with mean $\mu$, a fast and robust estimator of $\mu$ is given by $\widehat\mu_\tau = (1/n)\sum_{i=1}^n \psi_\tau(X_i)$.  \cite{M2016} extended this univariate estimation scheme to matrix settings based on the following definition on matrix functionals.

\begin{definition} \label{def.matrix.functional}
Given a real-valued function $f$ defined on $\RR$ and a symmetric $\Ab\in \RR^{d\times d}$ with eigenvalue decomposition $\Ab=\Ub\bLambda \Ub^\T$ such that $\bLambda=\textnormal{diag}(\lambda_1,\ldots, \lambda_d)$, $f(\Ab)$ is defined as $f(\Ab)=\Ub f(\bLambda)\Ub^\T$, where $f(\bLambda)=\textnormal{diag}  (f(\lambda_1), \ldots,f(\lambda_d)  )$.
\end{definition}

Suppose we observe $n$ random samples $\bX_1,\ldots, \bX_n$ from $\bX$ with mean $\bmu$ and covariance matrix $\bSigma=\EE \{ (\bX-\bmu)(\bX-\bmu)^\T \}$. If $\bmu$ were known, a robust estimator of $\bSigma$ can be simply constructed by $(1/n)\sum_{i=1}^n\psi_\tau \{(\bX_i-\bmu)(\bX_i-\bmu)^\T \}$. In practice, the assumption of a known $\bmu$ is often unrealistic.  Instead, we suggest to estimate $\bSigma$ using the following $U$-statistic based estimator:
\$
  \widehat\bSigma_{U}(\tau) =\frac{ 1}{\binom{n}{2}}\sum_{ 1\leq i <  j\leq n}\psi_\tau\bigg\{\frac{1}{2}(\bX_i -\bX_j)(\bX_i -\bX_j)^\T\bigg\}.
\$
Observe that $(\bX_i-\bX_j)(\bX_i- \bX_j)^\T$ is a rank one matrix with eigenvalue $\|  \bX_i-\bX_j \|_2^2$ and eigenvector $(\bX_i-\bX_j)/\| \bX_i-\bX_j \|_2$. Therefore, by Definition~\ref{def.matrix.functional}, $ \widehat\bSigma_{U}(\tau)$ can be equivalently written as
\#\label{u-type_cov}
\frac{ 1}{ \binom{n}{2}}\sum_{1\leq i < j\leq n}  \psi_\tau\bigg(\frac{1}{2} \| \bX_i -\bX_j \|_2^2 \bigg) \frac{(\bX_i-\bX_j)(\bX_i-\bX_j)^\T}{ \|  \bX_i-\bX_j  \|_2^2}.
\#
This alternative expression makes it much easier to compute. The following theorem provides an exponential-type deviation inequality for $\widehat\bSigma_U(\tau) $, representing a useful complement to the results in \cite{M2016}. See, for example, Remark~8 therein.

\begin{theorem}\label{thm:u-type}
Let
\#
v^2 = \frac{1}{2}\Big\| \EE \{(\bX-\bmu)(\bX-\bmu)^\T \}^2+\tr (\bSigma)\bSigma+2\bSigma^2 \Big\|.  \label{v2.def}
\#
For any $t>0$, the estimator $\widehat\bSigma_U = \widehat\bSigma_U(\tau)$ with $\tau \geq  (v/2)(n/t)^{1/2}$ satisfies
\$
 \PP \{ \|\widehat\bSigma_U -\bSigma\| \geq  4 v (t/n)^{1/2} \} \leq 2p \exp(-t).
\$
\end{theorem}

Given $\widehat\bSigma_{U}$, we can construct an estimator of $\Bb$ following Step~1 of the FarmTest procedure. Recall that $\hat{\bb}_1, \ldots, \hat{\bb}_p$ are the $p$ rows of $\hat{\Bb}$. To investigate the consistency of $\hat{\bb}_j$'s, let $\overline{\lambda}_1,\ldots, \overline{\lambda}_K$ be the top $K$ (nonzero) eigenvalues of $\Bb \Bb^\T$ in a descending order and $\overline{\bv}_1, \ldots, \overline{\bv}_K$ be the corresponding eigenvectors. Under identification condition \eqref{id.ass}, we have $\overline{\lambda}_\ell = \|\overline{\bb}_\ell \|_2^2$ and $\overline{\bv}_\ell = \overline{\bb}_\ell/ \|\overline{\bb}_\ell \|_2$ for $\ell =1,\ldots, K$.

\begin{assumption}[\sf Pervasiveness] \label{cond.pervasive}
{\rm
There exist positive constants $c_1$, $c_2$ and $c_3$ such that $c_1 p \leq \overline{\lambda}_{\ell} - \overline{\lambda}_{\ell+1} \leq c_2 p$ for $\ell =1,\ldots, K$ with $\overline{\lambda}_{K+1} := 0$, and $\| \bSigma_\varepsilon \| \leq c_3 <  \overline{\lambda}_K$.}
\end{assumption}

\begin{remark}
The pervasiveness condition is required for high dimensional spiked covariance model with the first
several eigenvalues well separated and significantly larger than the rest.  In particular, Assumption \ref{cond.pervasive} requires the top $K$ eigenvalues grow linearly with the dimension $p$. The corresponding eigenvectors can therefore be consistently estimated as long as sample size diverges \citep{FLM2013}. This condition is widely assumed in the literature \citep{{SW02},Bai_Ng_02}. The following proposition provides convergence rates of the robust estimators  $\{ \hat{\lambda}_\ell , \hat{\bv}_\ell \}_{\ell =1}^K$ under Assumption \ref{cond.pervasive}. The proof, which is given in in Appendix~D, is based on Weyl's inequality and a useful variant of the Davis-Kahan theorem \citep{YWS15}. We notice that some preceding works \citep{Ona2012, Shen_16,WF2017} have provided similar results under a weaker pervasiveness assumption which allows $p/n \rightarrow \infty$ in any manner and the spiked eigenvalues $\{ \overline{\lambda}_\ell \}_{\ell=1}^K$ are allowed to grow slower than $p$ so long as $c_\ell = p/(n \overline{\lambda}_\ell )$ is bounded.
\end{remark}

\begin{proposition} \label{prop1}
Under Assumption~\ref{cond.pervasive}, we have
\begin{gather}
	\max_{1\leq \ell \leq K }   |  \hat{\lambda}_\ell    -  \overline{\lambda}_\ell   |   \leq  \| \hat{\bSigma}_U - \bSigma  \|  +  \| \bSigma_\varepsilon \|  ~~\textnormal{and} \label{eigenvalue.bound} \\
	 \max_{1\leq \ell \leq K }    \| \hat{\bv}_\ell - \overline{\bv}_\ell  \|_2   \leq C  p^{-1}  (  \| \hat{\bSigma}_U - \bSigma  \|  +  \| \bSigma_\varepsilon  \|   ) , \label{eigenvector.bound}
\end{gather}
where $C >0$ is a constant independent of $(n,p)$.
\end{proposition}

We now show the properties of estimated loading vectors and estimated residual variances $\{ \hat \sigma_{\varepsilon , jj }\}_{j=1}^p$ that are defined below \eqref{hat.Tj}.

\begin{theorem} \label{thm3}
Suppose  Assumption~\ref{cond.corr}(iv) and Assumption~\ref{cond.pervasive} hold.  Let  $\tau = v_0 \omega_{n,p}$ with $v_0 \geq v/2$ for $v$ given in \eqref{v2.def}. Then, with probability at least $1-2n^{-1}$,
\#
 \max_{1\leq j\leq p}\| \hat{\bb}_j -  {\bb}_j  \|_2 \leq C_1 \{  v \sqrt{\log(np)} \, (np)^{-1/2} + p^{-1/2}   \}   \label{variance.concentration.1}
\#
as long as $n \geq  v^2 p^{-1} \log(np)$. In addition, if $n\geq C_2 \log(np)$, $\tau_j  = a_j \omega_{n,p}, \tau_{jj} = a_{jj} \omega_{n,p}$ with $a_j \geq \sigma_{jj}^{1/2}, a_{jj} \geq \var(X_j^2)^{1/2}$,  we have
\#
	\max_{1\leq j\leq p}  | \hat{\sigma}_{ \varepsilon ,jj} - \sigma_{\varepsilon , jj}  | \leq C_3 (  v  p^{-1/2} w_{n,p}^{-1} + p^{-1/2}   )  \label{variance.concentration.2}
\#
with probability greater than $1- C_4 n^{-1}$. Here, $C_1$--$C_4$ are positive constants that are independent of $(n,p)$.
\end{theorem}

\begin{remark}
According to Theorem \ref{thm3}, the robustification parameters can be set as $\tau_j  = a_j \omega_{n,p}$ and  $\tau_{jj} = a_{jj} \omega_{n,p}$, where $w_{n,p} $ is given in \eqref{wnp.def}. In practice, the constants $a_j$ and $a_{jj}$ can be chosen by cross-validation.
\end{remark}

\subsubsection{Adaptive Huber covariance estimation}
\label{sec:3.3.2}

In this section, we adopt an estimator that was first considered in \cite{FLW2017}. For every $1 \leq j \neq  k \leq p$, we define the robust estimate $\hat{\sigma}_{jk}$ of $\sigma_{jk}=\e(X_j X_k) - \mu_j \mu_k$ to be
\begin{align} \label{off-diagonal}
	\hat{\sigma}_{jk} = \hat{\theta}_{jk} -  \hat{\mu}_j \hat{\mu}_k ~\mbox{ with }~ \hat{\theta}_{jk} = \arg\min_{\theta \in \RR} \sn \ell_{\tau_{jk} } (X_{ij}X_{ik} - \theta ) ,
\end{align}
where $\tau_{jk}>0$ is a robustification parameter and $\hat{\mu}_j$ is defined in \eqref{rest}. 
This yields the adaptive Huber covariance estimator $\hat{\bSigma}_{{\rm H}} = (\hat{\sigma}_{jk})_{1\leq  j, k \leq p}$. The dependence of $\hat{\bSigma}_{{\rm H}}$ on $\{ \tau_{jk}: 1\leq j \leq k \leq   p\}$ and $\{ \tau_j\}_{j=1}^p$ is assumed without displaying.

\begin{theorem} \label{thm4}
Suppose Assumption~\ref{cond.corr}(iv) and Assumption~\ref{cond.pervasive} hold. Let $\tau_j  = a_j \omega_{n,p}, \tau_{jk} = a_{jk} \omega_{n,p^2}$ with $a_j \geq \sigma_{jj}^{1/2}, a_{jk} \geq \var(X_j^2)^{1/2}$ for $1\leq j, k\leq p$. Then, there exist positive constants $C_1$--$C_3$ independent of $(n,p)$ such that as long as $n\geq C_1 \log(np)$,
\begin{align}
\max_{1\leq j\leq p}\| \hat{\bb}_j -  {\bb}_j  \|_2 \leq C_2  ( \omega_{n,p}^{-1} + p^{-1/2} )   \nn \\
 \mbox{ and }~ \max_{1\leq j\leq p}  | \hat{\sigma}_{ \varepsilon ,jj} - \sigma_{\varepsilon , jj}  | \leq C_3 ( \omega_{n,p}^{-1} + p^{-1/2}  )  \nn
\end{align}
with probability greater than $1-  4n^{-1}$, where $w_{n,p}$ is given in \eqref{wnp.def}.
\end{theorem}

\subsection{Estimating realized factors}
\label{sec:3.3}

To make the oracle statistics $T_j^\circ$'s given in \eqref{oracle.stat} usable, it remains to estimate $\bar{\bbf}$. Since the loadings can be estimated in two different ways, let us first assume $\Bb$ is given and treat it as an input variable.

Averaging the approximate factor model \eqref{obs}, we have $\bar{\bX} = \bmu + \Bb \bar{\bbf} + \bar{\bvarepsilon}$, where $\bar{\bX} = (\bar{X}_1, \ldots, \bar{X}_p)^\T = (1/n)\sn \bX_i$ and $\bar{\bvarepsilon} := (\bar{\varepsilon}_1, \ldots, \bar{\varepsilon}_p)^\T =  (1/n) \sn \bvarepsilon_i$. This leads to
\begin{align}
	\bar{X}_j = \bb_j^\T \bar{\bbf} + \mu_j + \bar{\varepsilon}_j,\ \  j=1,\ldots, p. \label{new.model}
\end{align}
Among all $\{ \mu_j + \bar{\varepsilon}_j\}_{j=1}^p$, we may regard $\mu_j + \bar{\varepsilon}_j$ with $\mu_j \neq 0$ as outliers. Therefore, to achieve robustness, we estimate $\bar{\bbf}$ by solving the following optimization problem:
\begin{align}
	\widehat{\bbf}(\Bb)  \in  \arg\min_{  \bbf \in \RR^K } \sum_{j=1}^p \ell_{\gamma} ( \bar{X}_j  - \bb_j^\T \bbf )  ,  \label{oracle.factor.est}
\end{align}
where $\gamma = \gamma(n,p) > 0$ is a robustification parameter. Next, we define robust variance estimators $\hat{\sigma}_{\varepsilon, jj}$'s by
\begin{align}  
	\hat{\sigma}_{\varepsilon, jj}(\Bb)   = \hat{\theta}_j - \hat{\mu}_j^2 - \| \bb_j \|_2^2   ~\mbox{ with }~ \hat{\theta}_j = \arg\min_{\theta \, \geq\, \hat{\mu}_j^2 + \| \bb_j \|_2^2 } \sn \ell_{ \tau_{jj} } (X_{ij}^2 - \theta ), \nn
\end{align}
where $\tau_{jj}$'s are robustification parameters and $\hat{\mu}_j$'s are as in \eqref{rest}. Plugging $\{\widehat\sigma_{\varepsilon, jj}\}_{j=1}^p$ and $\widehat{\bbf}$ into \eqref{oracle.stat}, we obtain the following factor-adjusted test statistics
\begin{align}
	 {T}_j(\Bb) = \bigg\{  \frac{n}{\hat \sigma_{\varepsilon , jj}(\Bb) }  \bigg\}^{1/2}    \{ \hat{\mu}_j - {\bb}_j^\T \hat{\bbf}(\Bb) \}, \ \  j=1,\ldots, p. \label{dd.test}
\end{align}
For a given threshold $z \geq 0$, the corresponding FDP is defined as
\#
	 {\FDP}(z ; \Bb ) =  {V}(z ; \Bb ) /  {R}(z ; \Bb ), \nn
\#
where
$
{V}(z; \Bb ) =  \sum_{j \in \mathcal{H}_0}  I \{ |  {T}_j(\Bb)|\geq z \}$ \textnormal{and} ${R}(z; \Bb ) = \sum_{1\leq  j \leq p}  I \{ |  {T}_j (\Bb)| \geq z \}.
$
Similarly to \eqref{OFDP.def}, we approximate $ {\FDP}(z ; \Bb )$ by
\begin{align}  
	  \AFDP(z  ; \Bb ) =  2 p_0 \Phi(-z) / R(z ; \Bb )  .\nn
\end{align}
For any $z\geq 0$, the approximate FDP $ \AFDP(z;\Bb)$ is computable except $p_0$, which can be either estimated \citep{S2002} or upper bounded by $p$. Albeit being slightly conservative, the latter proposal is accurate enough in the sparse setting.

Regarding the accuracy of $  \AFDP(z  ; \Bb)$ as an asymptotic approximation of $ {\FDP}(z ; \Bb )$, we need to account for the statistical errors of $\{\widehat\sigma_{\varepsilon, jj} (\Bb)\}_{j=1}^p$ and $\hat{\bbf}(\Bb)$. To this end, we make the following structural assumptions on $\bmu$ and $\Bb$.

\begin{assumption} \label{cond.mineigen}
The  idiosyncratic errors $\varepsilon_1, \ldots, \varepsilon_p$ are mutually independent, and there exist constants $c_l , c_u>0$ such that $
	\lambda_{\min} (p^{-1}\Bb^\T \Bb ) \geq c_l$ and $\| \Bb \|_{\max} \leq c_u$.
\end{assumption}

\begin{assumption}[\sf Sparsity] \label{cond.sparsity}
There exist constants $C_\mu>0$ and $c_\mu \in (0, 1/2)$ such that $\| \bmu \|_\infty = \max_{1\leq j\leq p} | \mu_j | \leq C_\mu$ and $\| \bmu \|_0 = \sum_{j=1}^p I(\mu_j \neq 0) \leq p^{1/2 - c_\mu}$. Moreover, $(n,p)$ satisfies that $n \log(n) = o(p)$ as $n, p \to \infty$.
\end{assumption}

The following proposition, which is of independent interest, reveals an exponential-type deviation inequality for $\hat{\bbf}(\Bb)$ with a properly chosen $\gamma>0$.

\begin{proposition} \label{prop2}
Suppose that Assumption~\ref{cond.mineigen} holds. For any $t>0$, the estimator $\hat{\bbf}(\Bb)$ given in \eqref{oracle.factor.est} with $\gamma = \gamma_0 (p/t)^{1/2}$ for $\gamma_0 \geq \overline{\sigma}_\varepsilon  :=  (  p^{-1} \sum_{j=1}^p \sigma_{\varepsilon, jj}  )^{1/2}$ satisfies that with probability greater than $1- (2eK + 1) e^{-t}$,
\begin{align}  \label{factor.est.ubd}
	 \| \hat{\bbf}(\Bb) - \bar{\bbf} \|_2 \leq  C_1 \gamma_0 (K t)^{1/2} p^{-1/2}
\end{align}
as long as $p \geq \max \{   \| \bmu \|_2^2 / \overline{\sigma}_\varepsilon^2  , ( \| \bmu \|_1 / \overline{\sigma}_\varepsilon )^2 t , C_2 K^2 t \}$, where $C_1, C_2 >0$ are constants depending only on $c_l, c_u$ in Assumption~\ref{cond.mineigen}.
\end{proposition}

The convergence in probability of ${\FDP}( z ; \Bb )$ to $ \AFDP( z ; \Bb)$ for any $z \geq 0$ is investigated in the following theorem.

\begin{theorem} \label{thm2}
Suppose that Assumptions~\ref{cond.corr}\,(i)--(iv), Assumptions \ref{cond.mineigen} and \ref{cond.sparsity} hold.
Let $\tau_j = a_j \omega_{n,p}$, $\tau_{jj} = a_{jj} \omega_{n,p}$ with $a_j \geq \sigma_{jj}^{1/2}$, $a_{jj} \geq \var(X_j^2)^{1/2}$ for $j=1,\ldots, p$, and $\gamma = \gamma_0\{ p/  \log(n) \}^{1/2}$ with $\gamma_0 \geq \overline{\sigma}_\varepsilon$. Then, for any $z \geq  0$, $|{\FDP}( z ; \Bb )- \AFDP(z  ; \Bb)  | = o_{\PP} (1)$ as $n, p \to \infty$.
\end{theorem}

\section{Simulation studies}\label{sec:4}

\subsection{Selecting robustification parameters}
\label{sec:4.1}

The robustification parameter involved in the Huber loss plays an important role in the proposed procedures both theoretically and empirically.
In this section, we describe the use of cross-validation to calibrate robustification parameter in practice.
To highlight the main idea, we restrict our attention to the mean estimation problem.

Suppose we observe $n$ samples $X_1, \ldots, X_n$ from $X$ with mean $\mu$. For any given $\tau>0$, the Huber estimator is defined as $\hat{\mu}_\tau  = \arg\min_{\theta \in \RR } \sn \ell_{\tau} ( X_i- \theta )$, or equivalently, the unique solution of the equation $ \sn \psi_{\tau} ( X_i- \theta )=0$.
Our theoretical analysis suggests that the theoretically optimal $\tau$ is of the form $C_\sigma \omega_n$, where $\omega_n$ is a specified function of $n$ and $C_\sigma>0$ is a constant that scales with $\sigma$, the standard deviation of $X$.
This allows us to narrow down the search range by selecting $C_\sigma$ instead via the $K$-fold ($K=5$ or 10) cross-validation as follows.
First, we randomly divide the sample into $K$ subsets, ${\cal I}_1,\ldots, \cI_K$, with roughly equal sizes. The \textit{cross-validation} criterion for a given $C>0$ can be defined as
\#\label{cv_criterion}
 {\rm CV}(C )=\frac{1}{n}\sum\limits_{k=1}^K \sum\limits_{i \in {\cal I}_k }\{ X_i -\hat{\mu}^{(-k)}_{\tau_C} \}^2,
\#
where $\hat{\mu}^{(-k)}_{\tau_C}$ is  the Huber estimator using data not in the $k$-th fold, namely
$$
	\hat{\mu}^{(-k)}_{\tau_C} = \arg\min_{\theta \in \RR}\sum_{ \ell =1 ,\ell\neq k}^K
\sum_{i \in {\cal I}_\ell } \ell_{ \tau_C } ( X_i - \theta ) ,
$$
and  $\tau_C = C \omega_n$. In practice, let $\cal C$ be a set of grid points for $C$. We choose $C_\sigma$ and therefore $\tau$ by
$\hat{C}_\sigma =\arg\min_{C \in {\cal C}  } {\rm CV}(C )$ and $\hat{\tau} =\widehat{C}_\sigma \omega_n$.

The robustification parameters involved in the FarmTest procedure can be selected in a similar fashion by modifying the loss function and the {cross-validation} criterion (\ref{cv_criterion}) accordingly. The theoretical order $\omega_n$ can be chosen as the rate that guarantees optimal bias-robustness tradeoff. Based on the theoretical results in Section~\ref{sec:3}, we summarize the optimal rates for various robustification parameters in Table \ref{Sim_tab_cv}. Robust estimation of $\mu_j$'s and the adaptive Huber covariance estimator involve multiple robustification parameters. If $X_1, \ldots, X_p$ are homoscedastic, it is reasonable to assume $\tau_j=\tau_{\mu}$ in  (\ref{rest}) for all $j=1, \ldots , p$. {Then we can choose $\tau_{\mu}$ by applying the cross-validation over a small subset of the covariates $X_1, \ldots, X_p$.  Similarly, we can set $\tau_{jk}=\tau_{\Sigma}$ in (\ref{off-diagonal}) for all $j,k$ and calibrate $\tau_{\Sigma}$ by applying the cross-validation over a subset of the entries.}

\begin{table}[htbp]
\setlength{\tabcolsep}{1em}
\begin{center}
\caption{Optimal rates for robustification parameters}
\medskip
\label{Sim_tab_cv}
{\renewcommand{\arraystretch}{1.5}
\begin{tabular}{ccc}
\hline\hline
\text{Estimator} & \text{Parameter} & \text{Optimal Rate}
\\\hline
Robust estimator of $\mu_j$ &  $\tau_j$  in (\ref{rest}) & $\sqrt{n/\log(np)}$
\\
$U$-type covariance estimator  & $\tau$ in (\ref{u-type_cov}) & $p \sqrt{n /\log (p)}$
\\
Adaptive Huber covariance estimator & $\tau_{jk}$  in (\ref{off-diagonal})& $\sqrt{n/\log(np^2)}$
\\
Robust estimator of $\bar{\bbf}$ &  $\gamma$ in (\ref{oracle.factor.est}) &  $\sqrt{p/\log(n)}$
\\\hline
\end{tabular}
}
\end{center}

\end{table}

\subsection{Settings}\label{sec:4.2}

In the simulation studies, we take $(p_1, p)= (25,500)$ so that $\pi_1=p_1/p=0.05$, $n\in \{100, 150, 200 \}$ and use $t=0.01$ as the threshold value for $P$-values. Moreover, we set the mean vector $\bmu= (\mu_1,\ldots, \mu_p)^\T$ to be $\mu_j=0.5$ for $1 \leq j \leq 25$ and $\mu_j =0$ otherwise. We repeat 1000 replications in each of the scenarios below. The robustifications parameters are selected by five-fold cross-validation under the guidance of their theoretically optimal orders. The data-generating processes are as follows.

\medskip
\noindent
{\bf Model 1: Normal factor model}. Consider a three-factor model $\bX_i  = \bmu + \Bb \bbf_i + \bvarepsilon_i$, $i=1, \ldots , n$, where $\bbf_i \sim \mathcal{N}( \mathbf{0} , \Ib_3)$, $\Bb=(b_{j \ell})_{1\leq j\leq p, 1\leq \ell \leq 3}$ has IID entries $b_{j\ell}$'s generated from the uniform distribution $\mathcal{U}(-2, 2)$.


\medskip
\noindent
{\bf Model 2: Synthetic factor model}. Consider a similar three-factor model as in Model~1, except that $\bbf_i$'s and $ \bb_j$'s are generated independently from $\mathcal{N}( \mathbf{0}, \bSigma_f)$ and $\mathcal{N}(\bmu_B, \bSigma_B)$, respectively, where $\bSigma_f$, $\bmu_B$ and $\bSigma_B$ are calibrated from the daily returns of S\&P 500's top 100 constituents (ranked by the market cap) between July 1st, 2008 and June 29th, 2012.

\medskip
\noindent
{\bf Model 3: Serial dependent factor model}. Consider a similar three-factor model as in Model 1, except that $\bbf_i$'s are generated from a stationary VAR$(1)$ model $\bbf_i={\bf \Pi} \bbf_{i-1} + {\bxi}_i$ for $i=1, \ldots , n$, with $\bbf_0={\bf 0}$ and $\bxi_i$'s IID drawn from $\mathcal{N}({\bf 0}, \Ib_3)$.
The $(j,  k)$-th entry of ${\bf \Pi}$ is set to be 0.5 when $j=k$ and $0.4^{|j-k|}$ otherwise.

The idiosyncratic errors in these three models are generated from one of the following four distributions. Let $\bSigma_{\bvarepsilon}$ be a sparse matrix whose diagonal entries are 3 and off-diagonal entries are drawn from IID $0.3 \times {\rm Bernoulli}(0.05)$;
\begin{itemize}
\item[(1)] Multivariate normal distribution $\mathcal{N}( \mathbf{0} , \bSigma_{\bvarepsilon})$;
\item[(2)] Multivariate $t$-distribution $t_{3}(\mathbf{0}  ,\bSigma_{\bvarepsilon})$ with 3 degrees of freedom;
\item[(3)] IID Gamma distribution with shape parameter 3 and scale parameter 1;
\item[(4)] IID re-scaled log-normal distribution $a \{ \exp(1+1.2 Z ) - b\}$, where $Z \sim \mathcal{N}(0, 1)$ and $a, b >0$ are chosen such that it has mean zero and variance $3$.
\end{itemize}

\subsection{FDP estimation}

In our robust testing procedure, the covariance matrix is either estimated by the entry-wise adaptive Huber method or by the $U$-type robust covariance estimator. The corresponding tests are labeled as {\sf FARM-H} and {\sf FARM-$U$}, respectively.

In this section, we compare {\sf FARM-H} and {\sf FARM-$U$} with three existing non-robust tests.
The first one is a factor-adjusted procedure using the sample mean and sample covariance matrix, denoted by {\sf FAM}. The second one is the {\sf PFA} method, short for principal factor approximation, proposed by \cite{FH2017}. In contrast to {\sf FAM}, {\sf PFA} directly uses the unadjusted test statistics and only accounts for the effect of latent factors in FDP estimation.  The third non-robust procedure is the {\sf Naive} method, which completely  ignores the factor dependence.


We first examine the accuracy of FDP estimation, which is assessed by the median of the relative  absolute error (RAE) between the estimated FDP and $\FDP_{{\rm orc}}(t):=\frac{ \sum_{j\in \cH_0} I( P_j \leq t) } { \max\{ 1, \sum_{j=1}^p I(P_j\leq t) \}}$, where $P_j = 2\Phi(-|T_j^{{\rm o}}|)$ and $T_j^{{\rm o}}$ are the oracle test statistics given in (\ref{oracle.stat}). For a given threshold value $t$, RAE for $k$th simulation is defined as
$$
	\mathrm{RAE}(k)= |\widehat{ \FDP}(t, k) - \FDP_{{\rm orc}}(t, k)|/\FDP_{{\rm orc}}(t, k),  \quad k=1,  \ldots ,1000,
$$
where  $\widehat{ \FDP}(t, k)$ is the estimated FDP in the $k$th simulation using one of the five competing methods and $\FDP_{{\rm orc}}(t,k)$ is the oracle FDP in the $k$th experiment. The median of RAEs are presented in Table \ref{Sim_tab_FDP_relative}. We see that, although the {\sf PFA} and {\sf FAM} methods achieve the smallest estimation errors in the normal case, {\sf FARM-H} and {\sf FARM-$U$} perform comparably well. In other words, a high level of efficiency is achieved if the underlying distribution is normal.
The {\sf Naive} method performs worst as it ignores the impact of the latent factors. In heavy-tailed cases, both {\sf FARM-H} and {\sf FARM-$U$} outperform the non-robust competitors by a wide margin, still with the {\sf Naive} method being the least favorable. In summary, the proposed methods achieve high degree of robustness against heavy-tailed errors, while losing little or no efficiency under normality.

\begin{table}[htbp]
\small
\begin{center}
\caption{Median relative  absolute error between estimated and oracle FDP}
\medskip
\label{Sim_tab_FDP_relative}
\begin{tabular}{c|c|c|ccccc}
\hline\hline
& &  &  \multicolumn{3}{c}{$p=500$}
\\
&   $\bvarepsilon_i$ &  $n$
  & {\sf FARM-H} & {\sf FARM-$U$} & {\sf FAM} & {\sf PFA} & {\sf Naive}
\\\hline
\multirow{12}{*}{Model 1} &
\multirow{3}{*}{Normal}       & 100  & 0.8042	& 0.8063	& 0.7716	& 0.7487	& 1.789
\\
&                             & 150  & 0.7902	& 0.7925	& 0.7467	& 0.7790	& 1.599
\\
&                             & 200  & 0.7665	& 0.7743	& 0.7437	& 0.7363	& 1.538
\\\cline{2-8}
& \multirow{3}{*}{$t_3$}      & 100  & 0.7047	& 0.7539	& 1.3894	& 1.4676	& 2.061
\\
&                             & 150  & 0.6817	& 0.6002	& 1.1542	& 1.2490	& 1.801
\\
&                             & 200  & 0.6780	& 0.5244	& 0.9954	& 1.1306	& 1.579
\\\cline{2-8}
& \multirow{3}{*}{Gamma}      & 100  & 0.7034	& 0.7419	& 1.4986	& 1.7028	& 3.299
\\
&                             & 150  & 0.6844	& 0.6869	& 1.4396	& 1.5263	& 2.844
\\
&                             & 200  & 0.6393	& 0.6446	& 1.3911	& 1.4563	& 2.041
\\\cline{2-8}
& \multirow{3}{*}{LN}         & 100  & 0.6943	& 0.7104	& 1.5629	& 1.7255	& 3.292
\\
&                             & 150  & 0.6487	& 0.6712	& 1.6128	& 1.7742	& 3.092
\\
&                             & 200  & 0.6137	& 0.6469	& 1.4476	& 1.4927	& 2.510
\\\hline
\multirow{12}{*}{Model 2} &
\multirow{3}{*}{Normal}       & 100  & 0.6804	& 0.7079	& 0.6195	& 0.6318	& 1.676
\\
&                             & 150  & 0.6928	& 0.6873	& 0.6302	& 0.6136	& 1.573
\\
 &                            & 200  & 0.6847	& 0.6798	& 0.6037	& 0.6225	& 1.558
\\\cline{2-8}
& \multirow{3}{*}{$t_3$}      & 100  & 0.6438	& 0.6641	& 1.3939	& 1.4837	& 2.206
\\
&                             & 150  & 0.6258	& 0.6466	& 1.2324	& 1.2902	& 1.839
\\
&                             & 200  & 0.6002	& 0.6245	& 1.0368	& 1.0811	& 1.481
\\\cline{2-8}
& \multirow{3}{*}{Gamma}      & 100  & 0.6404	& 0.6493	& 1.6743	& 1.7517	& 3.129
\\
&                             & 150  & 0.5979	& 0.5991	& 1.3618	& 1.4405	& 2.657
\\
&                             & 200  & 0.5688	& 0.5746	& 1.0803	& 1.1595	& 2.035
\\\cline{2-8}
& \multirow{3}{*}{LN}         & 100  & 0.7369	& 0.7793	& 2.0022	& 2.0427	& 3.664
\\
&                             & 150  & 0.6021	& 0.6122	& 1.7935	& 1.8796	& 3.056
\\
&                             & 200  & 0.5557	& 0.5588	& 1.6304	& 1.8059	& 2.504
\\\hline
\multirow{12}{*}{Model 3} &
\multirow{3}{*}{Normal}       & 100  & 0.7937	& 0.8038	& 0.7338	& 0.7651	& 1.991
\\
&                             & 150  & 0.7617	& 0.7750	& 0.7415	& 0.7565	& 1.888
\\
 &                            & 200  & 0.7544	& 0.7581	& 0.7428	& 0.7440	& 1.858
\\\cline{2-8}
& \multirow{3}{*}{$t_3$}      & 100  & 0.7589	& 0.7397	& 1.4302	& 1.6053	& 2.105
\\
&                             & 150  & 0.6981	& 0.7010	& 1.2980	& 1.3397	& 1.956
\\
&                             & 200  & 0.6596	& 0.6846	& 1.1812	& 1.1701	& 1.847
\\\cline{2-8}
& \multirow{3}{*}{Gamma}      & 100  & 0.7134	& 0.7391	& 1.7585	& 1.9981	& 3.945
\\
&                             & 150  & 0.6609	& 0.6744	& 1.5449	& 1.7437	& 3.039
\\
&                             & 200  & 0.6613	& 0.6625	& 1.4650	& 1.4869	& 2.295
\\\cline{2-8}
& \multirow{3}{*}{LN}         & 100  & 0.7505	& 0.7330	& 1.8019	& 1.9121	& 3.830
\\
&                             & 150  & 0.6658	& 0.7015	& 1.7063	& 1.7669	& 3.278
\\
&                             & 200  & 0.6297	& 0.6343	& 1.5944	& 1.6304	& 2.937
\\\hline

\end{tabular}
\end{center}

\end{table}

\subsection{Power performance}

In this section, we compare the powers of the five methods under consideration. The empirical power is defined as the average ratio between the number of correct rejections and $p_1$. The results are displayed in Table \ref{Sim_tab_Power}. In the normal case, {\sf FAM} has a higher power than {\sf PFA}. This is because {\sf FAM} adjusts the effect of latent factors for each individual hypothesis so that the signal-to-noise ratio is higher.
Again, both {\sf FARM-H} and {\sf FARM-$U$} tests only pay a negligible price in power under normality. In heavy-tailed cases, however, these two robust methods achieve much higher empirical powers than their non-robust counterparts. Moreover, to illustrate the relationship between the empirical power and signal strength, Figure~\ref{Sim_fig_Power} displays the empirical power versus signal strength ranging from 0.1 to 0.8 for Model~1 with $(n,p)=(200,500)$ and $t_3$-distributed errors.

\begin{table}[htbp]
\small
\begin{center}
\caption{Empirical powers}
\medskip
\label{Sim_tab_Power}
\begin{tabular}{c|c|c|ccccc}
\hline\hline
& &  &  \multicolumn{3}{c}{$p=500$}
\\
&   $\bvarepsilon_i$ &  $n$
  & {\sf FARM-H} & {\sf FARM-$U$} & {\sf FAM} & {\sf PFA} & {\sf Naive}
\\\hline
\multirow{12}{*}{Model 1} &
\multirow{3}{*}{Normal}       & 100  & 0.853 & 0.849  & 0.872  & 0.863  & 0.585
\\
&                             & 150  & 0.877 & 0.870  & 0.890  & 0.882  & 0.624
\\
 &                            & 200  & 0.909 & 0.907  & 0.924  & 0.915  & 0.671
\\\cline{2-8}
& \multirow{3}{*}{$t_3$}      & 100  & 0.816 & 0.815  & 0.630  & 0.610  & 0.442
\\
&                             & 150  & 0.828 & 0.826  & 0.668  & 0.657  & 0.464
\\
&                             & 200  & 0.894 & 0.870  & 0.702  & 0.691  & 0.502
\\\cline{2-8}
& \multirow{3}{*}{Gamma}      & 100  & 0.816 & 0.813  & 0.658  & 0.639  & 0.281
\\
&                             & 150  & 0.830 & 0.825  & 0.684  & 0.663  & 0.391
\\
 &                            & 200  & 0.889 & 0.873  & 0.712  & 0.707  & 0.433
\\\cline{2-8}
&    \multirow{3}{*}{LN}      & 100  & 0.798 & 0.786  & 0.566  & 0.534  & 0.242
\\
&                             & 150  & 0.817 & 0.805  & 0.587  & 0.673  & 0.292
\\
&                             & 200  & 0.844 & 0.835  & 0.613  & 0.605  & 0.369
\\\hline
\multirow{12}{*}{Model 2} &
\multirow{3}{*}{Normal}       & 100  & 0.801 & 0.799  & 0.864  & 0.855  & 0.584
\\
&                             & 150  & 0.856 & 0.846  & 0.880  & 0.870  & 0.621
\\
 &                            & 200  & 0.904 & 0.900  & 0.911  & 0.904  & 0.659
\\\cline{2-8}
& \multirow{3}{*}{$t_3$}      & 100  & 0.810 & 0.802  & 0.612  & 0.601  & 0.402
\\
&                             & 150  & 0.825 & 0.814  & 0.638  & 0.632  & 0.457
\\
&                             & 200  & 0.873 & 0.859  & 0.695  & 0.683  & 0.484
\\\cline{2-8}
& \multirow{3}{*}{Gamma}      & 100  & 0.804 & 0.798  & 0.527  & 0.509  & 0.216
\\
&                             & 150  & 0.821 & 0.819  & 0.594  & 0.557  & 0.289
\\
 &                            & 200  & 0.885 & 0.875  & 0.638  & 0.606  & 0.379
\\\cline{2-8}
&    \multirow{3}{*}{LN}      & 100  & 0.763 & 0.757  & 0.463  & 0.434  & 0.206
\\
&                             & 150  & 0.799 & 0.795  & 0.495  & 0.479  & 0.228
\\
&                             & 200  & 0.826 & 0.819  & 0.529  & 0.511  & 0.312
\\\hline
\multirow{12}{*}{Model 3} &
\multirow{3}{*}{Normal}       & 100  & 0.837 & 0.832  & 0.848  & 0.833  & 0.535
\\
&                             & 150  & 0.856 & 0.848  & 0.864  & 0.857  & 0.594
\\
 &                            & 200  & 0.875 & 0.871  & 0.902  & 0.896  & 0.628
\\\cline{2-8}
& \multirow{3}{*}{$t_3$}      & 100  & 0.801 & 0.796  & 0.606  & 0.591  & 0.403
\\
&                             & 150  & 0.818 & 0.816  & 0.640  & 0.612  & 0.426
\\
&                             & 200  & 0.881 & 0.872  & 0.675  & 0.643  & 0.501
\\\cline{2-8}
& \multirow{3}{*}{Gamma}      & 100  & 0.792 & 0.785  & 0.385  & 0.329  & 0.205
\\
&                             & 150  & 0.818 & 0.809  & 0.472  & 0.435  & 0.281
\\
 &                            & 200  & 0.874 & 0.867  & 0.581  & 0.565  & 0.367
\\\cline{2-8}
&    \multirow{3}{*}{LN}      & 100  & 0.783 & 0.776  & 0.355  & 0.336  & 0.187
\\
&                             & 150  & 0.804 & 0.795  & 0.442  & 0.406  & 0.231
\\
&                             & 200  & 0.859 & 0.849  & 0.514  & 0.487  & 0.326
\\\hline
\end{tabular}
\end{center}

\end{table}

\begin{figure}[htbp]
 \centering
 \includegraphics[scale=0.5]{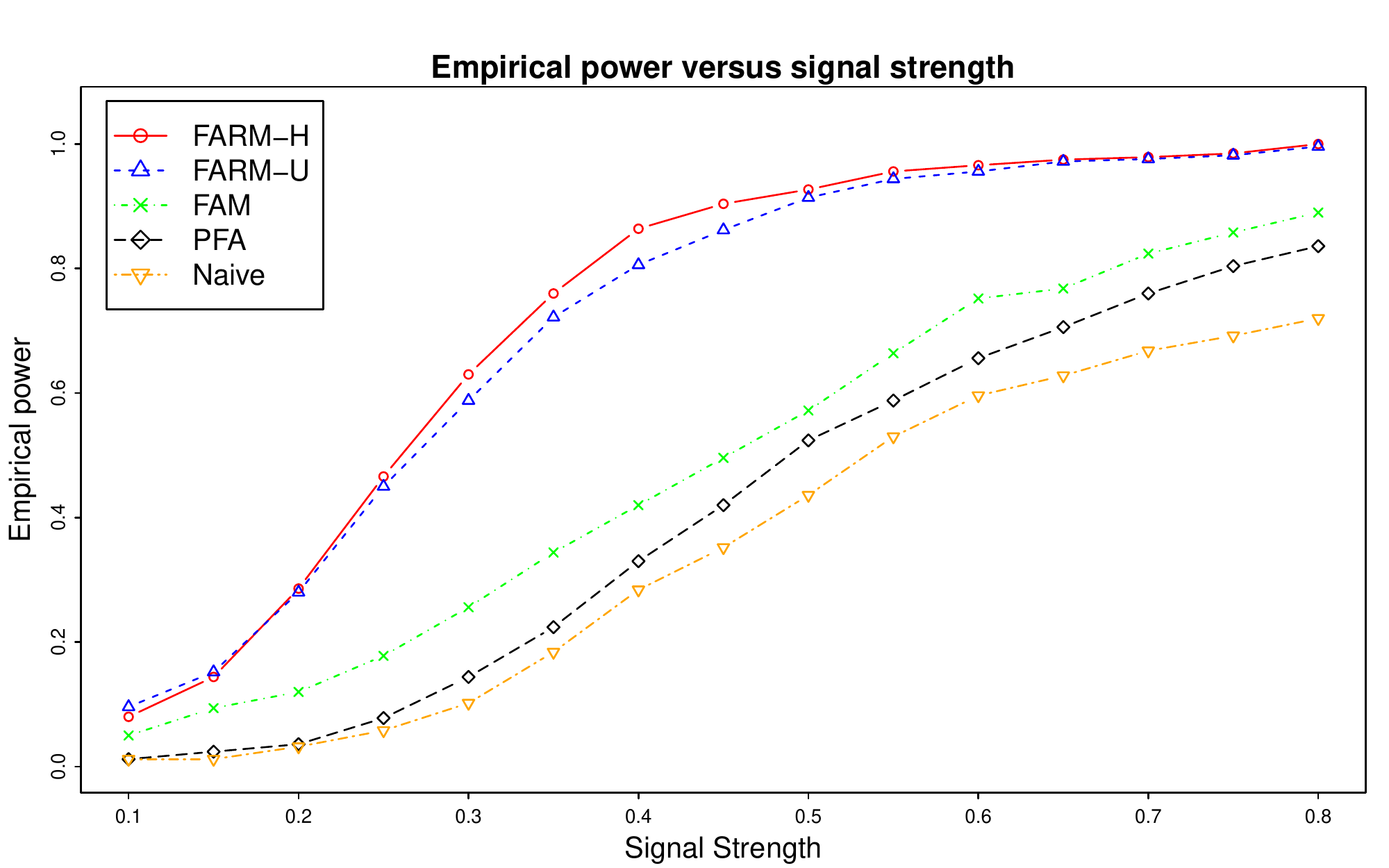}
  \begin{singlespace}
  \caption{Empirical power versus signal strength. The data are generated from Model 1 with $(n,p)=(200,500)$ and $t_3$-distributed noise. } \label{Sim_fig_Power}
  \end{singlespace}
\end{figure}

\subsection{FDP/FDR control}\label{sec:4.5}

In this section, we compare the numerical performance of the five tests in respect of FDP/FDR control.
We take $p=500$ and let $n$ gradually increase from 100 to 200.
The empirical FDP is defined as the average false discovery proportion based on 200 simulations.
At the prespecified level $\alpha=0.05$, Figure \ref{Sim_fig_FDP} displays the empirical FDP versus the sample size under Model~1.
In the normal case, all the four factor-adjusted  tests, {\sf FARM-H}, {\sf FARM-$U$}, {\sf FAM} and {\sf PFA}, have empirical FDPs controlled around or under $\alpha$.
For heavy-tailed data, {\sf FARM-H} and {\sf FARM-$U$} manage to control the empirical FDP under $\alpha$ for varying sample sizes; while {\sf FAM} and {\sf PFA} lead to much higher empirical FDPs, indicating more false discoveries. This phenomenon is in accord with our intuition that outliers can sometimes be mistakenly regarded as discoveries.
The {\sf Naive} method performs worst throughout all models and settings. Due to limitations of space, numerical results for Models~2 and 3 are given in Appendix E of the online supplement.

\begin{figure}[htbp]
 \centering
 \includegraphics[scale=0.4]{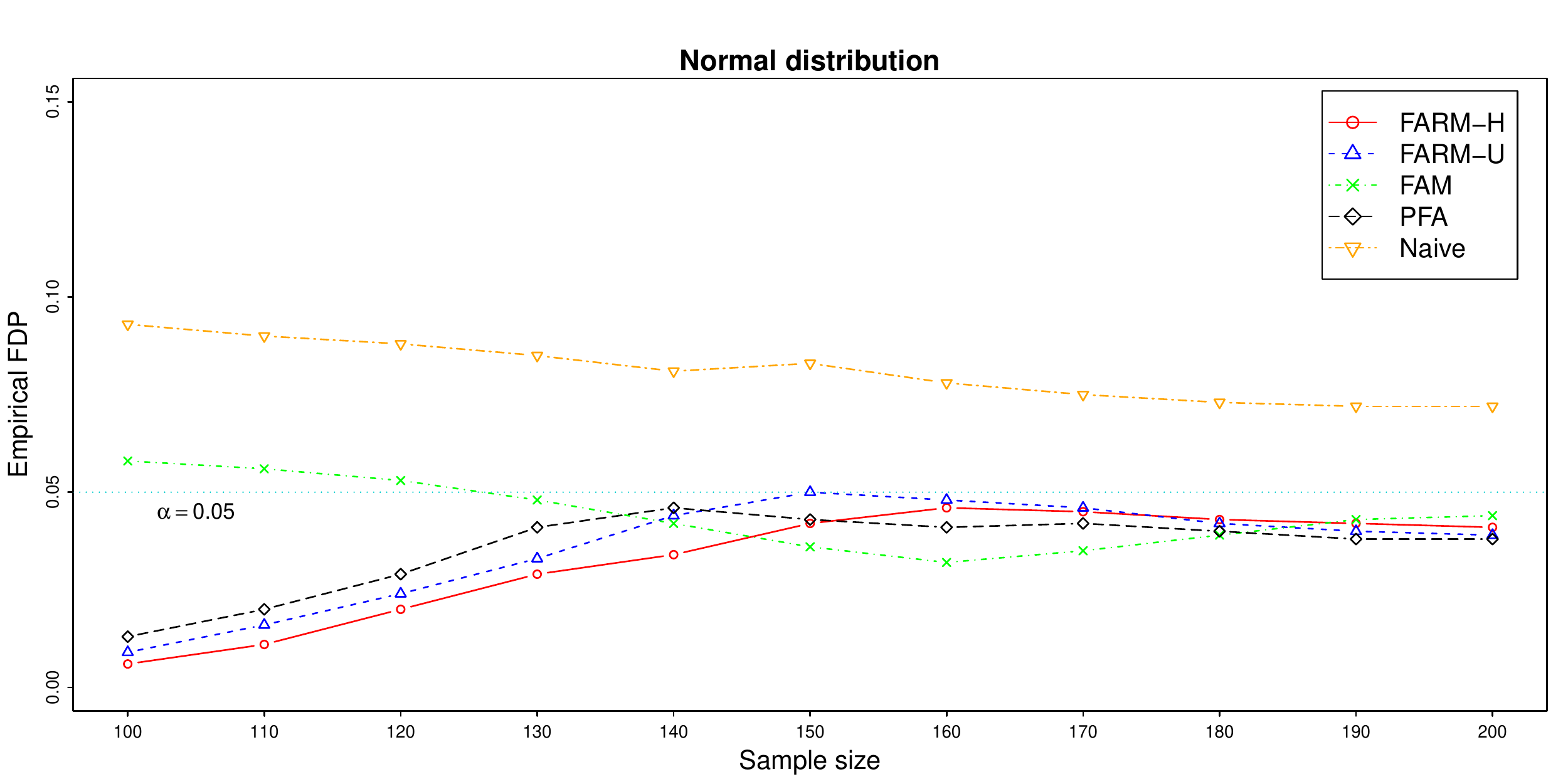}
 \includegraphics[scale=0.4]{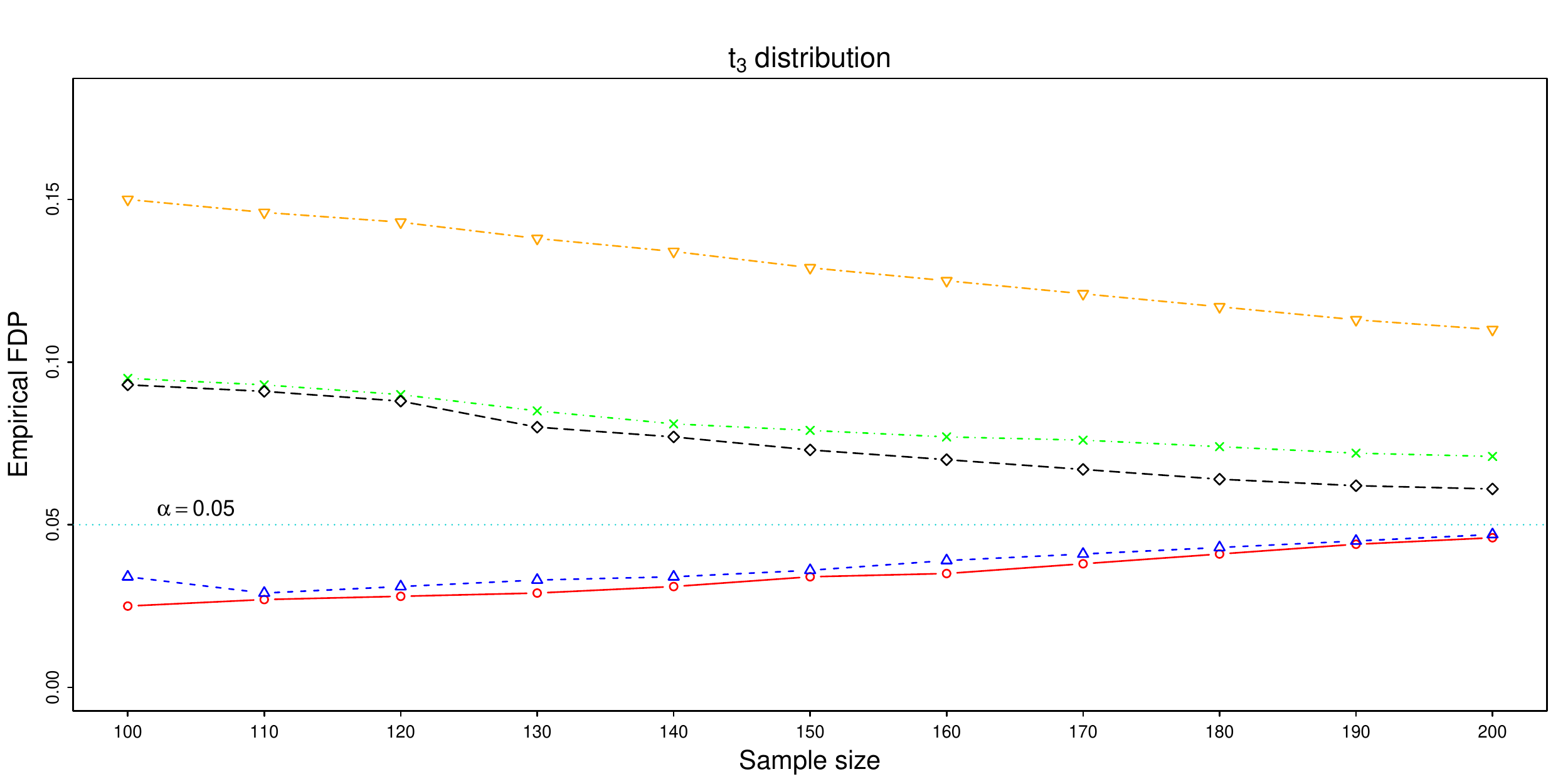}
 \includegraphics[scale=0.4]{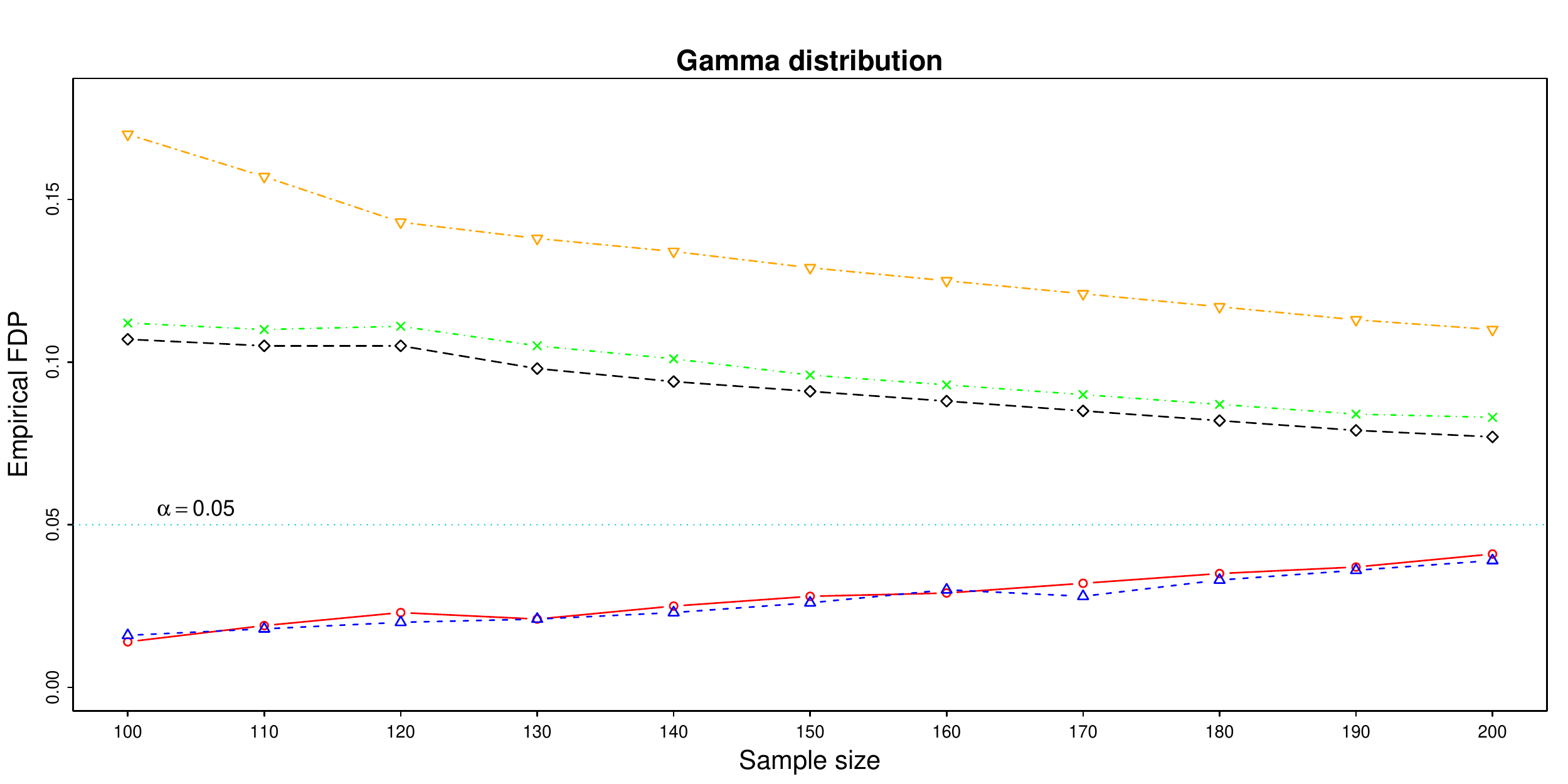}
 \includegraphics[scale=0.4]{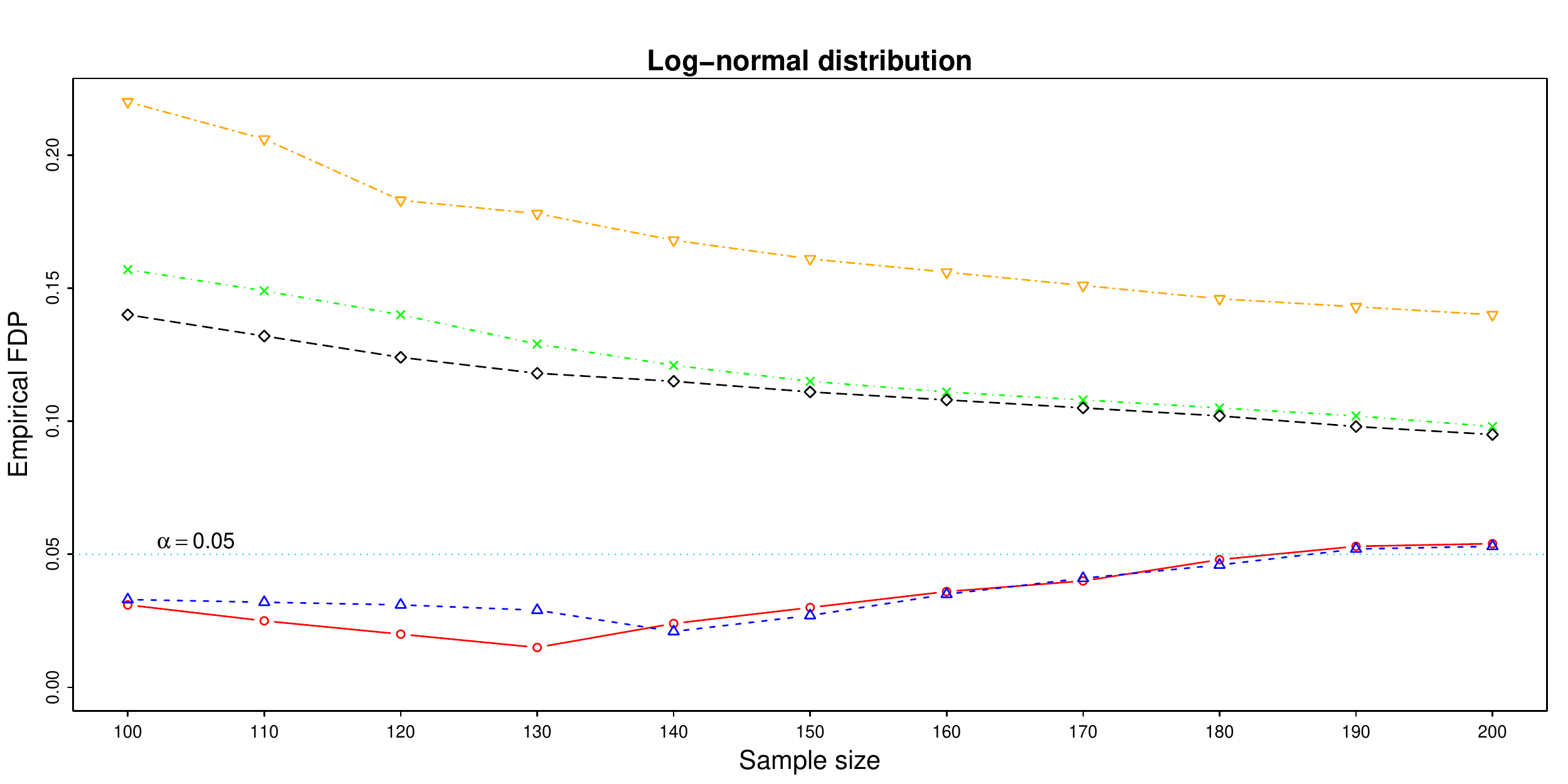}
  \begin{singlespace}
  \caption{Empirical FDP versus sample size for the five tests at level $\alpha=0.05$. The data are generated from Model 1 with $p=500$ and sample size $n$ ranging from 100 to 200 with a step size of 10. The panels from top to bottom correspond to the four error distributions in Section \ref{sec:4.2}. } \label{Sim_fig_FDP}
  \end{singlespace}
\end{figure}

\section{Real data analysis}
\label{sec:realdata}

\cite{Oberthuer_06} analyzed the German Neuroblastoma Trials NB90-NB2004 (diagnosed between 1989 and 2004) and developed a gene expression based classifier.
For 246 neuroblastoma patients, gene expressions over 10,707 probe sites were measured. The binary response variable is the 3-year event-free survival information of the patients (56 positive and 190 negative).
We refer to  \cite{Oberthuer_06} for a detailed description of the dataset.
In this study, we divide the data into two groups, one with positive responses and the other with negative responses, and test the equality of gene expression levels at all the 10,707 probe sites simultaneously.  To that end, we generalize the proposed FarmTest to the two-sample case by defining the following two-sample $t$-type statistics
$$
T_j=\frac{ (\hat{\mu}_{1j}- \hat{{\bb}}_{1j}^\T \hat{\bbf}_{1} )-
(\hat{\mu}_{2j}- \hat{{\bb}}_{2j}^\T \hat{\bbf}_{2} )}
{(\hat \sigma_{1\varepsilon , jj}/56 + \hat \sigma_{2\varepsilon , jj}/190)^{1/2}},
~~ j=1, \ldots , 10707,
$$
where the subscripts 1 and 2 correspond to the positive and negative groups, respectively. Specifically, $\hat{\mu}_{1j}$ and  $\hat{\mu}_{2j}$ are the robust mean estimators obtained from minimizing the empirical Huber risk (\ref{rest}), and $\hat{{\bb}}_{1j}$, $\hat{{\bb}}_{2j}$, $\hat{\bbf}_{1}$ and $\hat{\bbf}_{2}$ are robust estimators of the factors and loadings based on the $U$-type covariance estimator. In addition, $\hat\sigma_{1\varepsilon , jj}$ and $\hat\sigma_{2\varepsilon , jj}$ are the variance estimators defined in (\ref{dd.test}). As before, the robustification parameters are selected via five-fold cross-validation with their theoretically optimal orders taking into account.

\begin{figure}[htbp]
 \centering
 \includegraphics[width=5.3 in]{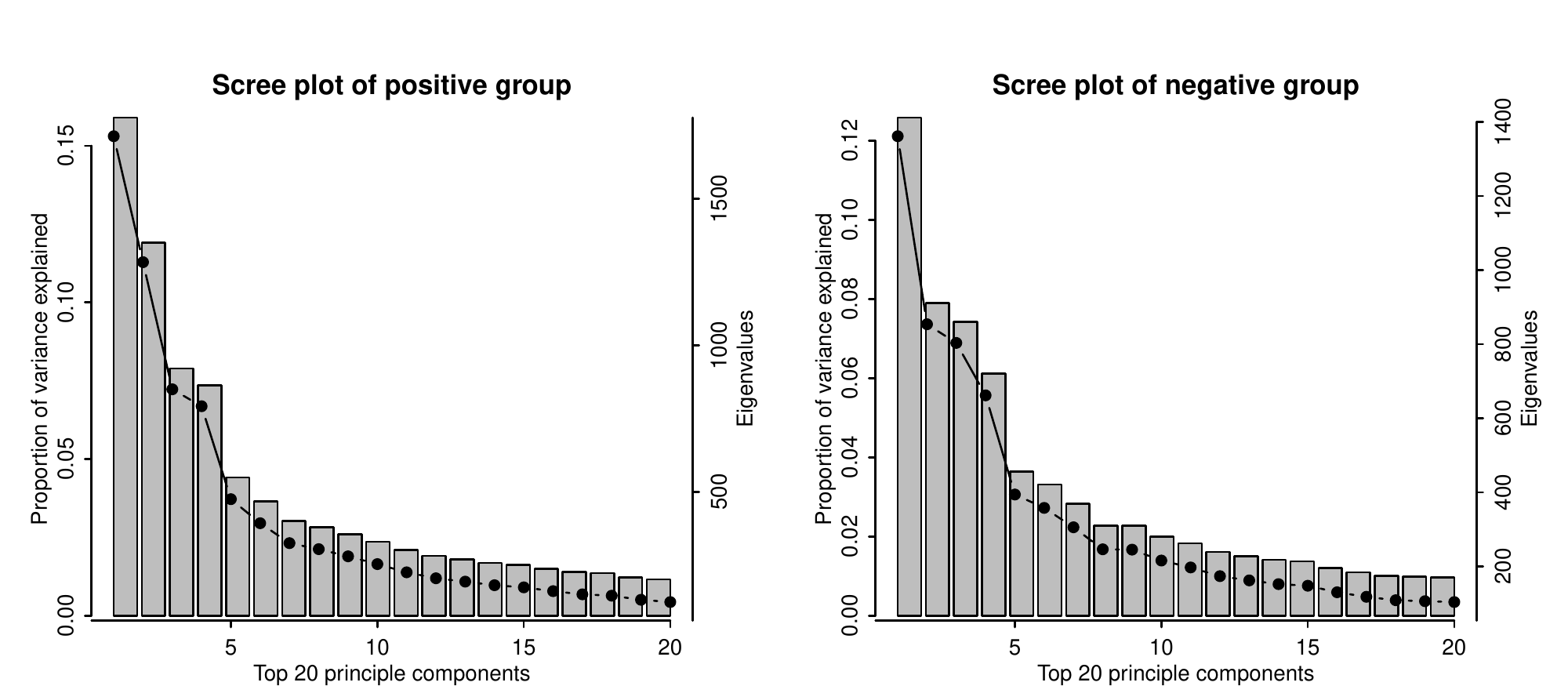}
 \begin{singlespace}
  \caption{{\bf Scree plots for positive and negative groups}.
  The bars represent the proportion of variance explained by the top 20 principal components. The dots represent the corresponding eigenvalues in descending order.}
 \label{Fig_Scree}
 \end{singlespace}
\end{figure}

\begin{figure}[htbp]
 \centering
 \includegraphics[width=5.3 in]{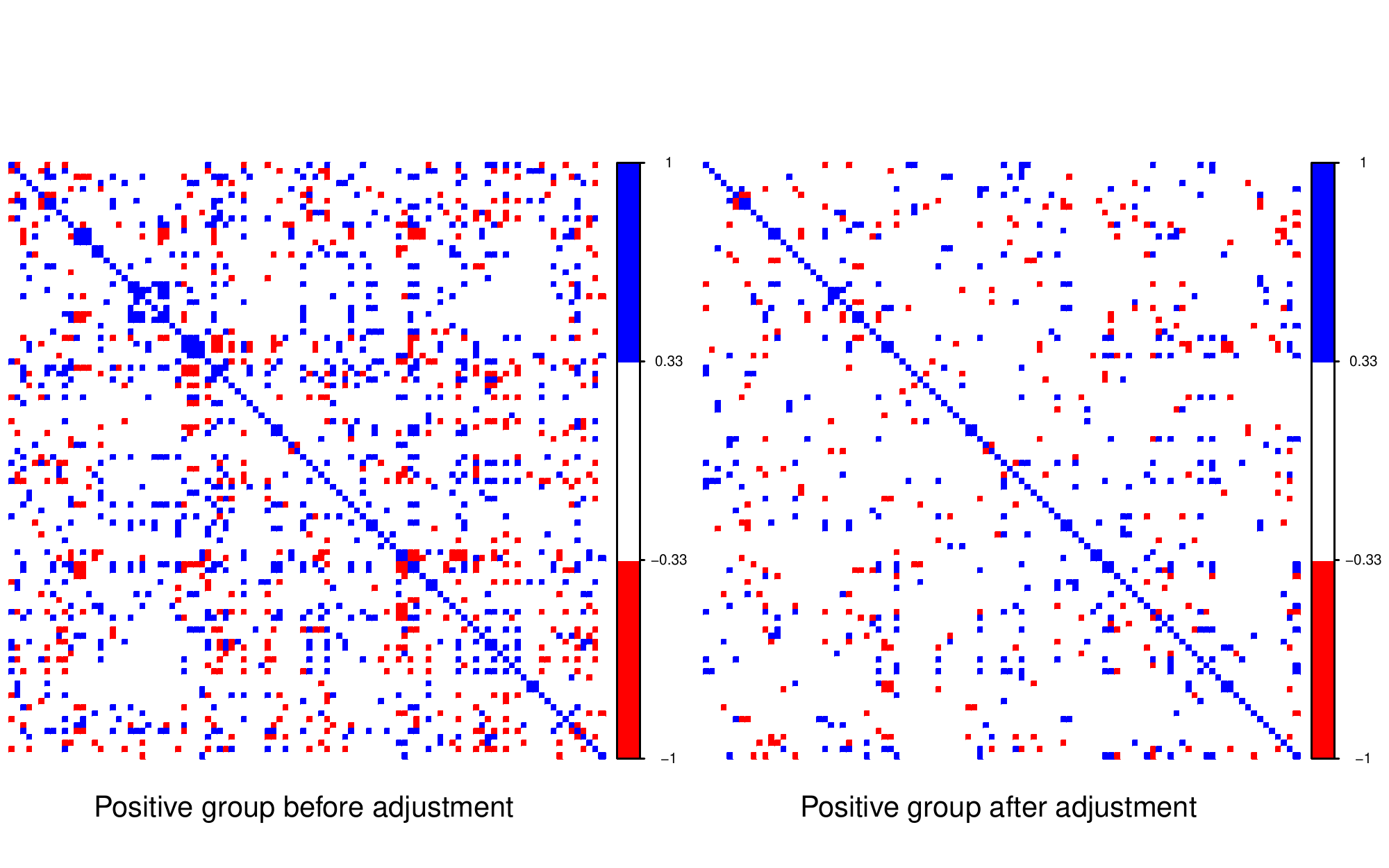}

 \vspace*{-0.15in}

 \includegraphics[width=5.3 in]{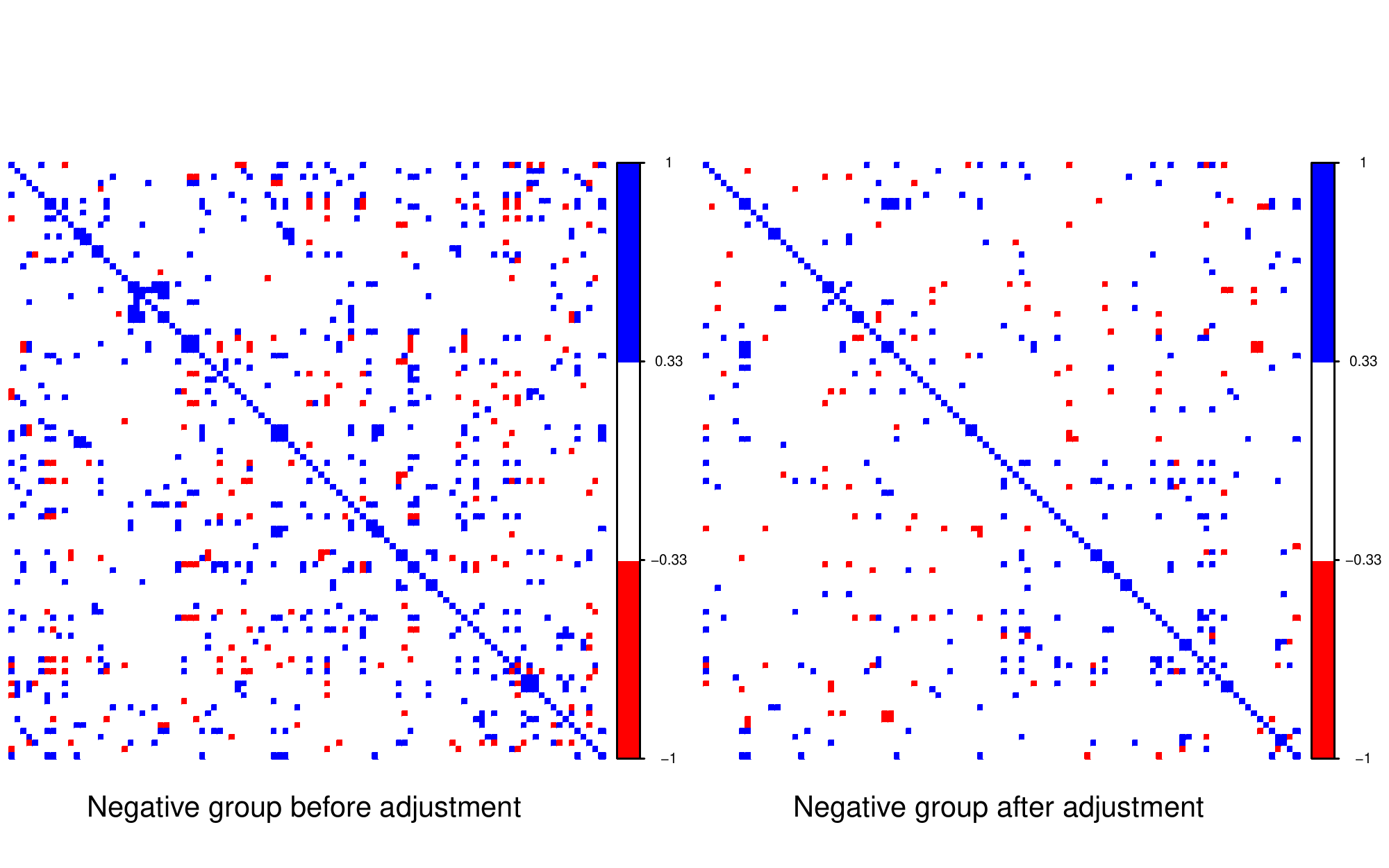}
  \begin{singlespace}
  \caption{{\bf Correlations among the first 100 genes before and after factor-adjustment}. The pixel plots are the correlation matrices of the first 100 gene expressions.
  In the plots, the blue pixels represent the entries with correlation greater than 1/3 and the red pixels represent the entries with correlation smaller than -1/3.}
 \label{Fig_corr}
 \end{singlespace}
\end{figure}

\begin{figure}[htbp]
 \centering
 \includegraphics[width=5.3 in]{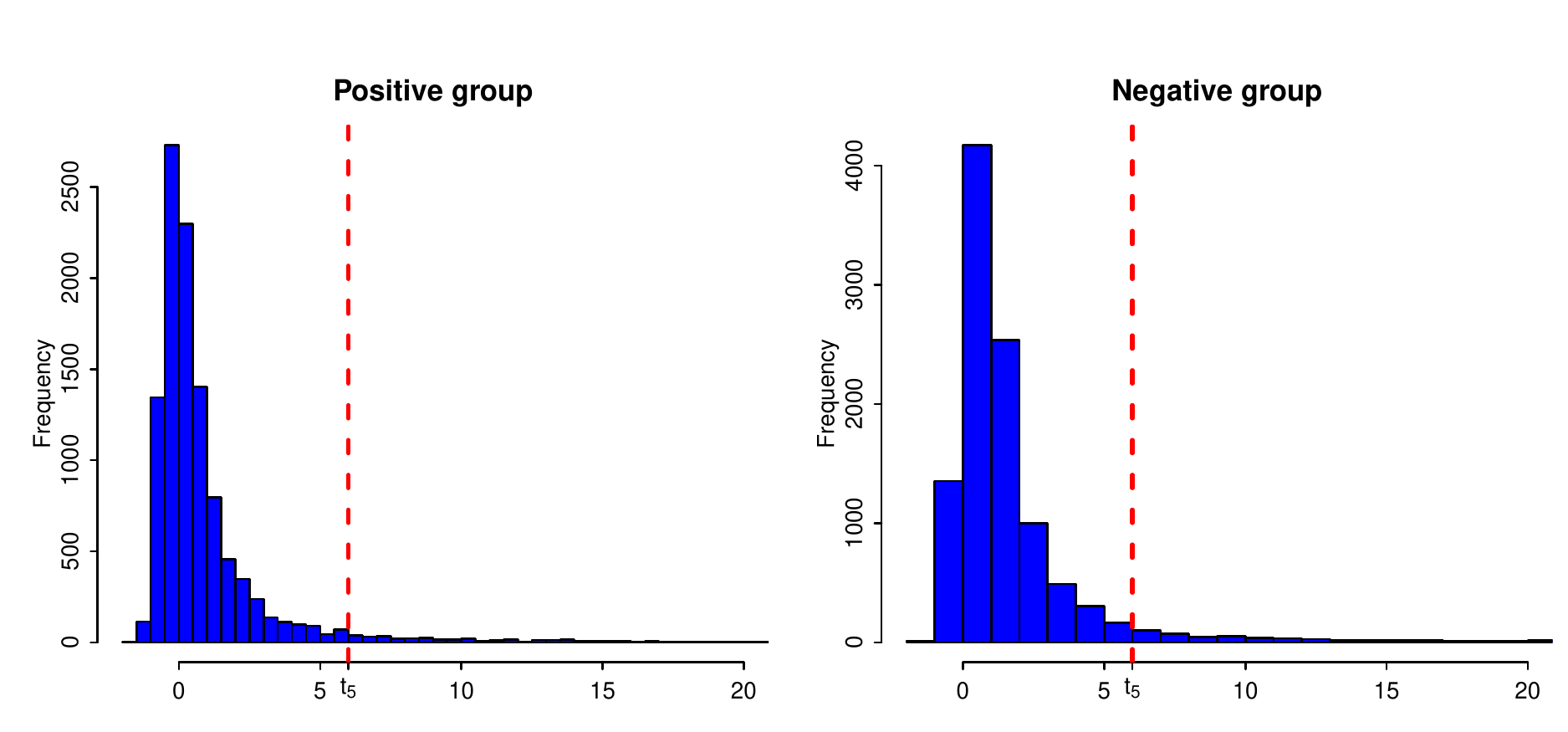}
   \begin{singlespace}
  \caption{{\bf Histogram of excess kurtosises for the gene expressions in positive and negative groups}. The dashed line at 6 is the excess kurtosis of $t_5$-distribution.}
 \label{Fig_EK}
 \end{singlespace}
\end{figure}

We use the eigenvalue ratio method \citep{LY2012, Ahn_Horen_2013} to determine the number of factors.  Let  $\lambda_k(\hat{\bSigma})$ be the
$k$-th largest eigenvalue of $\hat{\bSigma}$ and $K_{\max}$ a prespecified upper bound. The number of factors can then be estimated by
$$
\hat{K}=\arg\max\limits_{ 1\leq k \leq K_{\max}} \ {\lambda_k(\hat{\bSigma})}/{\lambda_{k+1}(\hat{\bSigma})}.
$$
The eigenvalue ratio method suggests $K=4$ for both positive and negative groups. Figure~\ref{Fig_Scree} depicts scree plots of the top 20 eigenvalues for each group. The gene expressions in both groups are highly correlated. As an evidence, the top 4 principal components (PCs) explain 42.6\% and 33.3\% of the total variance for the positive and negative groups, respectively.

To demonstrate the importance of the factor-adjustment procedure, for each group, we plot the correlation matrices of the first 100 gene expressions before and after adjusting the top 4 PCs; see Figure~\ref{Fig_corr}. The blue and red pixels in Figure~\ref{Fig_corr} represent the pairs of gene expressions whose absolute correlations are greater than 1/3. Therefore, after adjusting the top 4 PCs,  the number of off-diagonal entries with strong correlations is significantly reduced in both groups. To be more specific, the number drops from 1452 to 666 for the positive group and from 848 to 414 for the negative group.

Another stylized feature of the data is that distributions of many gene expressions are heavy-tailed. To see this, we plot histograms of  the excess kurtosis of the gene expressions in Figure~\ref{Fig_EK}. The left panel of the Figure \ref{Fig_EK} shows that 6518 gene expressions have positive excess kurtosis with 420 of them greater than 6. In other words, more than 60\% of the gene expressions in the positive group have tails heavier than the normal distribution and about 4\% are severely heavy tailed as their tails are fatter than the $t$-distribution with 5 degrees of freedom. Similarly, in the negative group, 9341 gene expressions exhibit positive excess kurtosis with 671 of them greater than 6. Such a heavy-tailed feature indicates the necessity of using robust methods to estimate the mean and covariance of the data.

{We apply four tests, the two-sample {\sf FARM-H} and {\sf FARM-$U$}, the {\sf FAM} test and the naive method, to this dataset.
At level $\alpha=0.01$, the two-sample {\sf FARM-H} and {\sf FARM-$U$} methods identify, respectively, 3912 and 3855 probes with different gene expressions, among which 3762 probes are identical. This shows an approximately 97\% similarity in the two methods.
The {\sf FAM} and naive methods discover  3509 and 3236 probes, respectively. For this dataset, accounting for latent factor dependence indeed leads to different statistical conclusions. This visible discrepancy between the two robust methods and {\sf FAM} highlights the importance of robustness and reflects the difference in power of detecting differently expressed probes. The effectiveness of factor adjustment is also highlighted in the discovery of significant genes}.

\section{Discussion and extensions}\label{sec:discuss}

In this paper, we have developed a factor-adjusted multiple testing procedure (FarmTest) for large-scale simultaneous inference with dependent and heavy-tailed data, the key of which lies in a robust estimate of the false discovery proportion. The procedure has two attractive features: First, it incorporates dependence information to construct marginal test statistics. Intuitively, subtracting common factors out leads to higher signal-to-noise ratios, and therefore makes the resulting FDP control procedure more efficient and powerful. Second, to achieve robustness against heavy-tailed errors that may also be asymmetric, we used the adaptive Huber regression method \citep{FLW2017, ZBFL2017} to estimate the realized factors, factor loadings and variances. We believe that these two properties will have further applications to higher criticism for detecting sparse signals with dependent and non-Gaussian data; see \cite{DHJ2011} for the independent case.

In other situations, it may be more instructive to consider the mixed effects regression modeling of the data \citep{FKC2009, WZHO2015}, that is, $X_j  = \mu_j  + \bbeta_j^\T \bZ + \bb_j^\T \bbf + \varepsilon_j$ for $j =1,\ldots, p$, where $\bZ \in \RR^q$ is a vector of explanatory variables (e.g., treatment-control, phenotype, health trait), $\bbeta_j$'s are $q\times 1$ vectors of unknown slope coefficients, and $\bbf$, $\bb_j$'s and $\varepsilon_j$'s have the same meanings as in \eqref{obs}. Suppose we observe independent samples $ (\bX_1, \bZ_1),\ldots, (\bX_n, \bZ_n)$ from $(\bX,\bZ)$ satisfying
\#
	\bX_i  =  \bmu + \bTheta \bZ_i + \Bb \bbf_i + \bvarepsilon_i , \ \ i=1,\ldots, n,  \nn
\#
where $\bTheta = (\bbeta_1, \ldots, \bbeta_p)^\T \in \RR^{p\times q}$. In this case, we have $\EE(\bX_i | \bZ_i) = \bmu + \bTheta \bZ_i$ and $\cov(\bX_i | \bZ_i ) = \Bb \bSigma_f \Bb^\T + \bSigma_\varepsilon$. The main issue in extending our methodology to such a mixed effects model is the estimation of $\bTheta$. For this, we construct robust estimators $(\hat{\mu}_j , \hat \bbeta_j)$ of $(\mu_j, \bbeta_j)$, defined as
$$
	(\hat{\mu}_j , \hat \bbeta_j) \in \arg\min_{\mu\in \RR, \, \bbeta_j \in \RR^q} \sn \ell_{\tau_j} (X_{ij} - \mu - \bbeta_j^\T \bZ_i ),  \ \ 1\leq  j\leq p,
$$
where $\tau_j$'s are robustification parameters. 
Taking $\hat{\bTheta}= (\hat{\bbeta}_1,\ldots, \hat{\bbeta}_p)^\T$, the FarmTest procedure in Section~\ref{sec:2.2} can be directly applied with $\{ \bX_i \}_{i=1}^n$ replaced by $\{ \bX_i - \hat{\bTheta} \bZ_i \}_{i=1}^n$. {However, because $\hat{\bTheta}$ depends on $\{ (\bX_i, \bZ_i) \}_{i=1}^n$, the adjusted data $\bX_1 - \hat{\bTheta} \bZ_1, \ldots, \bX_n - \hat{\bTheta} \bZ_n$ are no longer independent, which causes the main difficulty of extending the established theory in Section~\ref{sec:3} to the current setting.  One way to bypass this issue and to facilitate the theoretical analysis is the use of sample splitting as discussed in Appendix~A of the online supplement.
The FarmTest procedure for mixed effects models was also implemented in the {\sf R}-package {\sf FarmTest} (\href{https://cran.r-project.org/web/packages/FarmTest}{https://cran.r-project.org/web/packages/FarmTest}).}

\section*{Appendix}
\appendix

\section{Sample splitting}
\label{sec:sample.split}

Our procedure described in Sections~3.2 and 3.3 consists of two parts, the calibration of a factor model (i.e. estimating $\Bb$ in equation (1)) and multiple inference. The construction of the test statistics, or equivalently, the $P$-values, relies on a ``fine" estimate of $\bar{\bbf}$ based on the linear model in (25). In practice, $\bb_j$'s are replaced by the fitted loadings $\hat{\bb}_j$'s using the methods in Section~3.2.

To avoid mathematical challenges caused by the reuse of the sample, we resort to the simple idea of sample splitting \citep{H1969,C1975}: half the data are used for calibrating a factor model and the other half are used for multiple inference. We refer to \cite{R2016} for a modern look at inference based on sample splitting. Specifically, the steps are summarized below.

\begin{itemize}
\item[(1)] Split the data $\mathcal{X}= \{ \bX_1, \ldots, \bX_n \}$ into two halves $\mathcal{X}_1$ and $\mathcal{X}_2$. For simplicity, we assume that the data are divided into two groups of equal size $m=n/2$.

\item[(2)] We use $\mathcal{X}_1$ to estimate $\bb_1, \ldots, \bb_p$ using either the $U$-type method (Section~3.2.1) or the adaptive Huber method (Section~3.2.2). For simplicity, we focus on the latter and denote the estimators by $\hat{\bb}_1(\mathcal{X}_1),\ldots, \hat{\bb}_p(\mathcal{X}_1)$.

\item[(3)] Proceed with the remain steps in the FarmTest procedure using the data in $\mathcal{X}_2$. Denote the resulting test statistics by ${T}_1, \ldots, {T}_p$. For a given threshold $z \geq 0$, the
corresponding FDP and its asymptotic expression are defined as
\#
	 {\FDP}_{\rm sp}(z ) = {V}(z ) /  {R}(z ) ~\mbox{ and }~  \AFDP_{\rm sp}(z ) =  2 p \Phi(-z) /    R(z), \nn
\#
respectively, where $ {V}(z ) =  \sum_{j \in \mathcal{H}_0}  I ( |  {T}_j |\geq z )$,  ${R}(z ) = \sum_{1\leq j \leq p}  I (  |  {T}_j  | \geq z  )$ \textnormal{and} the subscript ``sp'' stands for sample splitting.
\end{itemize}

The purpose of sample splitting employed in the above procedure is to facilitate the theoretical analysis. The following result shows that the asymptotic FDP $\AFDP_{\rm sp}(z)$, constructed via sample splitting, provides a consistent estimate of ${\FDP}(z)$.

\begin{theorem} \label{thm.all}
Suppose that Assumptions~1\,(i)--(iv), Assumptions 2--4 hold. Let $\tau_j = a_j \omega_{n,p}$, $\tau_{jj} = a_{jj} \omega_{n,p}$ with $a_j \geq \sigma_{jj}^{1/2}$, $a_{jj} \geq \var(X_j^2)^{1/2}$ for $j=1,\ldots, p$, and let $\gamma = \gamma_0\{ p/  \log(np) \}^{1/2}$ with $\gamma_0 \geq \overline{\sigma}_\varepsilon$. Then, for any $z\geq 0$, $|  \AFDP_{\rm sp}(z ) - {\FDP}_{\rm sp}(z) | = o_{\PP} (1)$ as $n, p \to \infty$.
\end{theorem}

\section{Derivation of (6)}\label{appendix:A}

For any $t$  and $a_j \geq  \sigma_{jj}^{1/2} $,  Lemma~\ref{lemBR} in Section~\ref{sec.pre} shows that, conditionally on $\bbf_i$'s, the rescaled robust estimator $\sqrt{n} \,\hat{\mu}_j$ with $\tau_j  = a_j (n/t)^{1/2}$ satisfies
\#
	\sqrt{n} \, (\hat{\mu}_j-\mu_j-\bb^\T_j \bar\bbf ) & =\bigg\{ {  \frac{1}{\sqrt{n}} \sum_{i=1}^n \ell'_{\tau_j}(\bb_j^\T \bbf_i + \varepsilon_{ij})-  \sqrt{n} \, \bb^\T_j\bar\bbf} \bigg\} + R_{1j} , \label{stat.dec1.app}
\#
where the remainder $R_{1j}$ satisfies $\PP ( |R_{1j}| \lesssim a_j  n^{-1/2} t  ) \geq  1-  3 e^{-t}$.
The stochastic term $n^{-1/2} \sn \ell'_{\tau_j}(\bb_j^\T \bbf_i + \varepsilon_{ij})-\sqrt{n}\,\bb^\T_j\bar\bbf $ in (6) can be decomposed as
\begin{align} \label{stat.dec2}
	 \frac{1}{\sqrt{n}} \sum_{i=1}^n \ell'_{\tau_j}(\bb_j^\T \bbf_i + \varepsilon_{ij})   -\sqrt{n}\, \bb_j^\T \bar{\bbf} & = \underbrace{  \frac{1}{\sqrt{n}}\sn   \{  \ell'_{\tau_j}(\bb_j^\T \bbf_i + \varepsilon_{ij}) - \EE_{\bbf_i}  \ell'_{\tau_j}(\bb_j^\T \bbf_i + \varepsilon_{ij}) \}  }_{S_j} \nn \\
 & \quad  +  \underbrace{ \frac{1}{\sqrt{n}} \sum_{i=1}^n \{\EE_{\bbf_i}  \ell'_{\tau_j}(\bb_j^\T \bbf_i + \varepsilon_{ij}) - \bb_j^\T \bbf_i  \} }_{R_{2j}} ,
\end{align}
where $\bar{\bbf} = n^{-1} \sum_{i=1}^n \bbf_i$ and $\EE_{\bbf_i}(\cdot) = \EE (\cdot | \bbf_i )$ denotes the conditional expectation given $\bbf_i$. Under the finite fourth moment condition $\upsilon_{j}  := ( \EE  \varepsilon_j^4 )^{1/4} <\infty$, it follows from Lemma~\ref{app.meanvar} that as long as $n \geq  4 a_j^{-2}  \max_{1\leq i\leq n} (\bb_j^\T \bbf_i )^2 t$,
\#  \label{Ej.ubd}
	| R_{2j} | \leq  8  a_j^{-3} \upsilon_{j}^4 \,  n^{-1} t^{3/2}  .
\#
Given $\{ \bbf_i \}_{i=1}^n$, $S_j$ in \eqref{stat.dec2} is a sum of (conditionally) independent random variables with (conditional) mean zero. In addition, we note from \eqref{approxi.var} in Lemma~\ref{app.meanvar} that the (conditional) variance of $ \ell'_{\tau_j}(\bb_j^{\rm T} \bbf_i + \varepsilon_{ij})$ given $\bbf_i$ satisfies $|  \var_{\bbf_i} \{   \ell'_{\tau_j}(\bb_j^{\rm T} \bbf_i + \varepsilon_{ij}) \}  - \sigma_{\varepsilon ,jj} | \lesssim n^{-1} t$. Therefore, by the central limit theorem, the conditional distribution of $S_j$ given $\{ \bbf_i \}_{i=1}^n$ is asymptotically normal with mean zero and variance $\sigma_{\varepsilon,jj}$ as long as $t=t(n,p)=o(n)$. This, together with \eqref{Ej.ubd} implies that, conditioning on $\{ \bbf_i \}_{i=1}^n$, the distribution of $\sqrt{n} \,\hat{\mu}_j$ is close to a normal distribution with mean $\sqrt{n} \, ( \mu_j +   \bb_j^{{\rm T}} \bar{\bbf}  )$ and variance $\sigma_{\varepsilon,jj}$. Under the identifiability condition (2), $ \sigma_{\varepsilon,jj} = \sigma_{jj} - \| \bb_j \|_2^2$ for $j=1,\ldots, p$. We complete the proof.

\section{Proofs of main results}
\label{sec.proof}
In this section, we present the proofs for Theorems~1--5 and Theorem~\ref{thm.all}, starting with some preliminary results whose proofs can be found in Section~\ref{app:B}. Recall that
$$
	w_{n,p} = \sqrt{\frac{n}{\log(np)}},
$$
which will be frequently used in the sequel. Also, we use $c_1, c_2, \ldots$ and $C_1, C_2, \ldots$ to denote constants that are independent of $(n,p)$, which may take different values at each occurrence.

\subsection{Preliminaries}
\label{sec.pre}

For each $1\leq j\leq p$, define the zero-mean error variable $\xi_j = X_j - \mu_j$ and let $\mu_{j,\tau} =   \argmin_{\theta \in \RR }  \e  \ell_\tau(X_j - \theta ) $ be the approximate mean, where $\ell_\tau(\cdot)$ is the Huber loss given in (5). Throughout,  we use $\psi_\tau$ to denote the derivative of $\ell_\tau$, that is,
\$
	\psi_\tau(u) = \ell_\tau'(u) =  \min(|u|, \tau) \sgn(u) , \ \  u  \in \RR.
\$
Lemma~\ref{lem2} provides an upper bound on the approximation bias $|\mu_j - \mu_{j,\tau}|$, whose proof is given in Section \ref{secC3}.

\begin{lemma}  \label{lem2}
Let $1\leq j\leq p$ and assume that $\upsilon_{\kappa,j} = \e  (  | \xi_j |^{\kappa} ) <\infty$ for some $\kappa \geq   2$. Then, as long as $\tau > \sigma_{jj}^{1/2}$, we have
\begin{align}  \label{lem2.1}
 | \mu_{j, \tau} - \mu_j  |    \leq  ( 1- \sigma_{jj}  \tau^{-2}  )^{-1} \upsilon_{\kappa, j}  \tau^{1-\kappa}   .
\end{align}
\end{lemma}

The following concentration inequality is from Theorem~5 in \cite{FLW2017}, showing that  $\hat{\mu}_j$ with a properly chosen robustification parameter $\tau$ exhibits sub-Gaussian tails when the underlying distribution has heavy tails with only finite second moment.

\begin{lemma} \label{lem3}
For every $1\leq j\leq p$ and $t>0$, the estimator $\hat{\mu}_j$ in (5) with $\tau = a (n/t)^{1/2}$ for $a \geq  \sigma_{jj}^{1/2}$ satisfies $\PP \{  | \hat{\mu}_j  - \mu_j  |  \geq   4 a (t/n)^{1/2} \} \leq 2e^{-t}$ as long as $n\geq  8t$.
\end{lemma}

The next result provides a nonasymptotic Bahadur representation for $\hat{\mu}_j$. In particular, we show that when the second moment is finite, the remainder of the Bahadur representation for $\hat{\mu}_j$ exhibits sub-exponential tails.  The proof of Lemmas~\ref{lemBR}--\ref{infinity.perturbation} can be found respectively in Sections \ref{secC4}--\ref{secC7}.

\begin{lemma} \label{lemBR}
For every $1\leq j \leq p$ and for any $t\geq  1$, the estimator $\hat{\mu}_j$ in (5) with $\tau = a (n/t)^{1/2}$ and $a \geq  \sigma_{jj}^{1/2}$ satisfies that as long as $n \geq  8t$,
\begin{align}  \label{Bahadur.representation}
	   \bigg|  \sqrt{n} \, (\hat{\mu}_j - \mu_j  ) - \frac{1}{\sqrt{n}}  \sn \psi_\tau(  \xi_{ij}  )   \bigg|     \leq C  \frac{a t}{\sqrt{n}}
\end{align}
with probability greater than $1- 3 e^{-t}$, where $\xi_{ij} = X_{ij} - \mu_j$ and $C>0$ is an absolute constant.
\end{lemma}

Under factor model (1), note that $\xi_j = \bb_j^\T \bbf + \varepsilon_j $ for every $j$. The following conclusion reveals that the differences between the first two (conditional) moments of $\xi_{j} $ and $\psi_\tau(\xi_j)$ given $\bbf$ vanish faster if higher moments of $\varepsilon_j$ exist.

\begin{lemma} \label{app.meanvar}
Assume that $ \e (   | \varepsilon_j  |^{\kappa}  ) <\infty$ for some $1 \leq j\leq p$ and $\kappa \geq  2$.
\begin{itemize}
\item[(1)] On the event $\mathcal{G}_j :=  \{ |\bb_j^\T \bbf| < \tau \}$,
\# \label{approxi.mean}
	 | \e_{\bbf}  \psi_\tau( \xi_j)  -  \bb_j^\T \bbf  | \leq \min\bigg\{    \frac{\sigma_{\varepsilon , jj}}{ \tau - |\bb_j^{{\rm T}} \bbf|}  ,   \frac{ \e | \varepsilon_j |^\kappa }{ (\tau - |\bb_j^{{\rm T}} \bbf|)^{\kappa - 1} }  \bigg\}
\#
almost surely. In addition, if $\kappa>2$, we have
\# \label{approxi.var}
   \sigma_{\varepsilon , jj}  -    \frac{\e ( | \varepsilon_j |^\kappa ) }{(\tau - |\bb_j^\T \bbf|)^{\kappa-2}}  \bigg\{ \frac{2}{\kappa-2}   +     \frac{ \e ( | \varepsilon_j |^\kappa ) }{ (\tau - |\bb_j^\T \bbf|)^{ \kappa  }} \bigg\}   \leq  {\rm var}_{\bbf} \{ \psi_{\tau}(\xi_j)  \}    \leq \sigma_{\varepsilon ,jj}
\#
almost surely on $\mathcal{G}_j $.

\item[(2)] Assume that $\upsilon_{jk } := \e  ( |\varepsilon_j|^\kappa )  \vee \e  ( | \varepsilon_k|^\kappa  ) <\infty$ for some $1\leq j\neq k \leq p $ and $\kappa>2$. Then
\begin{align} \label{approxi.cov}
	 |  \cov_{\bbf}  (\psi_\tau(\xi_j), \psi_\tau(\xi_k)  )   - \cov( \varepsilon_j , \varepsilon_k) | \leq C  \max (  \upsilon_{jk} \tau^{2-\kappa} , \upsilon_{jk}^2 \tau^{2-2\kappa}  )
\end{align}
almost surely on $\mathcal{G}_{jk} := \{  | \bb_j^\T \bbf |  \vee  | \bb_k^\T \bbf | \leq \tau/2  \}$, where $C>0$ is an absolute constant.
\end{itemize}
\end{lemma}

\begin{lemma} \label{factor.concentration}
Suppose that Assumption~1 holds. Then, for any $t>0$,
\begin{align}
		\PP \{  \| \sqrt{n} \bar{\bbf} \|_2 > C_1  A_f (K+t)^{1/2} \} \leq e^{-t} , \label{factor.concentration.ineq} \\
		\PP\bigg\{ \max_{1\leq i\leq n} \| \bbf_i \|_2 > C_1 A_f (K+\log n + t)^{1/2} \bigg\} \leq e^{-t} , \label{factor.l2.ineq} \\
		\mbox{ and }~\PP [  \|  \hat{\bSigma}_f - \Ib_K \|_2 > C_2 \max \{   A_f^2 n^{-1/2} ( K +t)^{1/2} , A_f^4 n^{-1} (K+t)  \}  ] \leq 2e^{-t} , \label{factor.cov.concentration}
\end{align}	
where $\bar{\bbf} = n^{-1}\sn \bbf_i$, $\hat{\bSigma}_f = n^{-1} \sn \bbf_i \bbf_i^\T$ and $C_1, C_2>0$ are absolute constants.
\end{lemma}

The following lemma provides an $\ell_\infty$-error bound for estimating the eigenvectors $\overline{\bv}_\ell$'s of $\Bb^\T \Bb$. The proof is based on an $\ell_\infty$ eigenvector perturbation bound developed in \cite{FWZ2016} and is given in Appendix~\ref{app:B}.

\begin{lemma} \label{infinity.perturbation}
Suppose Assumption~2 holds. Then we have
\begin{align}
	\max_{1\leq \ell \leq K }  |  \wt{\lambda}_\ell - \overline{\lambda}_\ell | \leq  p \| \hat{\bSigma}_{{\rm H}}  -\bSigma \|_{\max} + \| \bSigma_\varepsilon \| \label{eigenvalue.perturb} \\
	\mbox{ and }~ \max_{1\leq \ell \leq K }   \| \hat{\bv}_\ell - \overline{\bv}_\ell \|_\infty \leq  C (  p^{-1/2}  \| \hat{\bSigma}_{{\rm H}} - \bSigma\|_{\max}  +  p^{-1} \| \bSigma_\varepsilon \|  ) ,  \label{eigenvector.perturb}
\end{align}
where $C>0$ is a constant independent of $(n,p)$.
\end{lemma}

\subsection{Proof of Theorem~1}

To prove (15) and (16), we will derive the following stronger results that
\#
		p_0^{-1}V(z)&  = 2\Phi(-z) + O_{\PP}\{p^{-\kappa_1/2} + w_{n,p}^{-1/2} + n^{-1/2} \log(np)\}  \label{V.LLN+} \\
\mbox{ and }~ p^{-1} R(z) & =  \frac{1}{p}  \sum_{j=1}^p  \bigg\{ \Phi\bigg( -z+  \frac{  \sqrt{n}  \mu_j}{ \sqrt{ \sigma_{\varepsilon, jj} } }  \bigg) +     \Phi\bigg(-z - \frac{\sqrt{n} \mu_j}{\sqrt{ \sigma_{\varepsilon, jj} }}  \bigg) \bigg\}  \nn \\
&  \quad~ + O_{\PP}\{p^{-\kappa_1/2} + w_{n,p}^{-1/2} + n^{-1/2} \log(np)\}   \label{R.LLN+}
\#
uniformly over $z\geq 0$ as $n, p \to \infty$, where $w_{n,p} = \sqrt{n/\log(np)}$.

For $1\leq j\leq p$ and $t\geq  1$, set $\tau_j = a_j \sqrt{n/t}$ with $a_j \geq \sigma^{1/2}_{jj}$. By Lemma~\ref{lemBR}, for every $j \in \cH_0$ so that $\mu_j=0$,
\begin{align}
 	 |  T_j^\circ  - \sigma_{\varepsilon ,jj}^{-1/2} ( S_j + R_{2j} )  |  \leq c_1    \frac{ a_jt}{\sqrt{\sigma_{\varepsilon ,jj} n}}  \label{pf1.1}
\end{align}
with probability greater than $1- 3 e^{-t}$ as long as $n \geq  8t$, where
\begin{align}	
	 S_j = \frac{1}{\sqrt{n}} \sn S_{ij}    ~\mbox{ with }~ S_{ij} : =  \psi_{\tau_j}(\bb_j^\T \bbf_i + \varepsilon_{ij}) - \EE_{\bbf_i}   \psi_{\tau_j}(\bb_j^\T \bbf_i + \varepsilon_{ij})    , \label{T0j.decom}
\end{align}
$R_{2j} = n^{-1/2} \sn   \{\EE_{\bbf_i}   \psi_{\tau_j}(\bb_j^\T \bbf_i +\varepsilon_{ij})  - \bb_j^\T \bbf_i   \}$. For $j=1,\ldots, p$, denote by $\mathcal{E}_{1j}(t)$ the event that \eqref{pf1.1} holds. Define $\mathcal{E}_1(t) = \bigcap_{j=1}^p \mathcal{E}_{1j}(t)$, on which it holds
\begin{align}
	\sum_{j \in \mathcal{H}_0}  I\bigg( |T_{0j} | \geq  z +  \frac{c_1 a_j t}{\sqrt{\sigma_{\varepsilon , jj} n}}  \bigg) \leq 	V( z )  \leq  \sum_{j \in \mathcal{H}_0} I\bigg( |T_{0j} | \geq  z -    \frac{ c_1 a_j t}{\sqrt{\sigma_{\varepsilon , jj} n }} \bigg), \label{pf1.2}
\end{align}
where $T_{0j} := \sigma_{\varepsilon , jj}^{-1/2} ( S_j + R_{2j} )$. Next, let $\mathcal{E}_2(t)$ denote the event on which the following hold:
\$
	\| \sqrt{n} \bar{\bbf} \|_2 \leq  C_1  A_f (K+t)^{1/2} , \ \ \max_{1\leq i\leq n} \| \bbf_i \|_2 \leq C_1 A_f (K+\log n + t)^{1/2}, \nn \\
	\mbox{and }~  \|  \hat{\bSigma}_f - \Ib_K \|_2 \leq  C_2 \max \{   A_f^2 n^{-1/2} ( K +t)^{1/2} , A_f^4 n^{-1} (K+t)  \}   . \nn
\$
From Lemmas~\ref{lemBR}, \ref{factor.concentration} and the union bound, it follows that
$$
	\PP \{  \mathcal{E}_1(t)^{ {\rm c} } \} \leq p e^{-t} ~~\mbox{ and }~~ \PP\{ \mathcal{E}_2(t)^{ {\rm c} } \} \leq 4 e^{-t}.
$$

With the above preparations, we are ready to prove \eqref{V.LLN+}. The proof of \eqref{R.LLN+} follows the same argument and therefore is omitted. Note that, on the event $\cE_2(t)$,
\$
	\max_{1\leq i\leq n} | \bb_j^\T \bbf_i | \leq  C_1  A_f \| \bb_j \|_2  (K+\log n + t)^{1/2}  ~\mbox{ for all }  1\leq j \leq p.
\$
By the definition of $\tau_j$'s,
\#
	\max_{1\leq i\leq n} | \bb_j^\T \bbf_i |  \leq \tau_j /2 ~\mbox{ for all } j =1,\ldots, p, \label{pf1.3}
\#
as long as $n\geq 4 (C_1 A_f)^2 (K + \log n + t ) t$. This, together with Lemma~\ref{factor.concentration}, implies $|R_{2j}| \leq 8 a_j^{-3}  \upsilon_j^4\, n^{-1} t^{3/2} $ almost surely on $\mathcal{E}_2(t)$ for all sufficiently large $n$. Moreover, taking \eqref{pf1.2} into account we obtain that, almost surely on the event $\mathcal{E}_1(t)  \cap \mathcal{E}_2(t)$,
\begin{align}
	\sum_{j \in \mathcal{H}_0} I( | \sigma_{\varepsilon , jj}^{-1/2} S_j  | \geq  z + c_2   n^{-1/2} t ) \leq 	V( z )  \leq  \sum_{j \in \mathcal{H}_0} I( | \sigma_{\varepsilon , jj}^{-1/2} S_j   | \geq  z - c_2n^{-1/2} t )  \label{pf1.4}
\end{align}
as long as $n \gtrsim (K+ t)t$. For $x\in \RR$, define
\begin{align}
\widetilde{V}_+(x) = \sum_{j \in \mathcal{H}_0} I(  \sigma_{\varepsilon , jj}^{-1/2} S_j   \geq  x   ) ~\mbox{ and }~ \widetilde{V}_-(x) = \sum_{j \in \mathcal{H}_0}  I( \sigma_{\varepsilon , jj}^{-1/2} S_j \leq -x   ) ,  \nn
\end{align}
so that  \eqref{pf1.4} can be written as
\begin{align}
	& 	p_0^{-1}  \{ \widetilde V_+ ( z + c_2 n^{-1/2} t ) +  \widetilde V_- ( z + c_2 n^{-1/2} t )  \} \nn \\
	& \leq  p_0^{-1} V(z ) \leq  p_0^{-1} \{ \widetilde V_+ ( z - c_2 n^{-1/2} t ) +  \widetilde V_- ( z - c_2 n^{-1/2} t ) \}   . \label{pf1.5}
\end{align}
Therefore, to prove \eqref{V.LLN+} it suffices to focus on $\widetilde{V}_+$ and $\widetilde{V}_-$.

Observe that, conditional on $\mathcal{F}_n := \{ \bbf_1, \ldots, \bbf_n \}$, $I(  \sigma_{\varepsilon ,11}^{-1/2} S_1 \geq  z  ) , \ldots,  I( \sigma_{\varepsilon ,pp}^{-1/2} S_p \geq  z  )$ are weakly correlated random variables. Define
$$
	Y_j = I ( \sigma_{\varepsilon , jj}^{-1/2} S_j \geq  z ) ~\mbox{ and } ~ P_j = \PP ( \sigma_{\varepsilon , jj}^{-1/2} S_j \geq  z   |\mathcal{F}_n  )
$$
for $j=1,\ldots, p$, and note that
\begin{align}
	\var\bigg( \frac{1}{p_0} \sum_{j \in \mathcal{H}_0 } Y_j  \bigg| \mathcal{F}_n \bigg) & = \frac{1}{p_0^2}  \sum_{j \in \mathcal{H}_0 } \var(Y_j  | \mathcal{F}_n) + \frac{1}{p_0^2} \sum_{j,k\in \mathcal{H}_0: j\neq k} \cov(Y_j, Y_k |\mathcal{F}_n)  \nn \\
	& \leq \frac{1}{4 p_0} +  \frac{1}{p_0^2}  \sum_{j,k\in \mathcal{H}_0: j\neq k}  \{  \e(Y_j Y_k | \mathcal{F}_n)  - P_j P_k \}   \label{pf1.6}
\end{align}
almost surely. In the following, we will study $P_j$ and $\EE(Y_j Y_k | \cF_n)$ separately, starting with the former. Conditional on $\mathcal{F}_n$, $S_j$ is a sum of independent zero-mean random variables with conditional variance $s^2_{j} := \var( S_j | \mathcal{F}_n ) = n^{-1} \sn s_{ij}^2$ where $s_{ij}^2 := \var( S_{ij}  | \mathcal{F}_n  )$. Let $G \sim \cN(0,1)$ be a standard normal random variable independent of the data. By the Berry-Esseen inequality,
\begin{align}
	 & \sup_{x\in \RR}  |   \PP (   \sigma_{\varepsilon , jj}^{-1/2} S_j \leq x | \mathcal{F}_n  )   - \PP ( s_j  \sigma_{\varepsilon , jj}^{-1/2}    G\leq x | \mathcal{F}_n ) | \nn \\
	& \lesssim   \frac{1}{  (  n s_{j} )^{3/2}}\sn  \e_{\bbf_i} | \psi_{\tau_j}(\bb_j^\T \bbf_i + \varepsilon_{ij})  |^3  \lesssim \frac{1}{     ( n s_j )^{3/2}}\sn  (   | \bb_j^\T \bbf_i  |^3+  \e | \varepsilon_{ij}|^3 )   \label{pf1.7}
\end{align}
almost surely, where conditional on $\mathcal{F}_n$, $s_j  \sigma_{\varepsilon , jj}^{-1/2}   G   \sim \cN (0 , s_j^2  \sigma_{\varepsilon , jj}^{-1}  )$. Since $\max_{1\leq i\leq n}| \bb_j^\T \bbf_i |\leq \tau_j/2$ for all $1\leq j\leq p$ on $\mathcal{E}_2(t)$, applying Lemma~\ref{app.meanvar} with $\kappa=4$ yields
\begin{align}
	   \sigma_{\varepsilon , jj} -    4 a_j^{-2}\upsilon_j^4  ( 1 +   16 a_j^{-4} \upsilon_j^4  \, n^{-2}t^2  ) n^{-1} t \leq    s^2_{j}      \leq  \sigma_{\varepsilon , jj}   \label{var.bound}
\end{align}
almost surely on the event $\mathcal{E}_2(t)$. Using \eqref{var.bound} and Lemma~A.7 in the supplement of \cite{SZ2015}, we get
\begin{align}
	 \sup_{x\in \RR}  | \PP ( s_j  \sigma_{\varepsilon , jj}^{-1/2}    G\leq x  | \mathcal{F}_n  )  - \Phi(x) | \lesssim   a_j^{-2}\upsilon_j^4  \, n^{-1} t  \label{Gaussian.comparison}
\end{align}
almost surely on $\mathcal{E}_2(t)$ as long as $n \gtrsim  (K+t)t$. Putting \eqref{pf1.7} and \eqref{Gaussian.comparison} together we conclude that, almost surely on $\cE_2(t)$,
\begin{align}
	 \max_{1\leq j\leq p }| P_j -  \Phi (-z)  | \lesssim     n^{-1/2}  (K+ \log n + t)^{1/2}	  \label{pf1.8}
\end{align}
uniformly for all $z\geq 0$ as long as $ n \gtrsim (K+t) t$.

Next we consider the joint probability $\e(Y_j Y_k | \mathcal{F}_n) = \PP ( \sigma_{\varepsilon , jj}^{-1/2} S_j  \geq  z , \sigma_{\varepsilon , kk}^{-1/2} S_k \geq  z  | \mathcal{F}_n  )$ for a fixed pair $(j,k)$ satisfying $1\leq j\neq k\leq p$. Define bivariate random vectors $\bxi_i = ( s_{j}^{-1} S_{ij}, s_{k}^{-1} S_{ik} )^\T$ for $i=1,\ldots, n$. Observe that $\bxi_1,\ldots,\bxi_n$ are conditionally independent random vectors given $\mathcal{F}_n$. Denote by $\Ab =(a_{uv})_{1\leq u, v \leq 2}$ the covariance matrix of $n^{-1/2} \sn \bxi_i = (s_{j}^{-1}S_j, s_{k}^{-1}S_k)^\T$ given $\mathcal{F}_n$ such that
$$
	a_{11} = a_{22}=1 ~~\mbox{ and }~~	a_{12} = a_{21} = \frac{1}{ n s_{j}s_{k} } \sn \cov_{\bbf_i} ( S_{ij}, S_{ik} ).
$$
By Lemma~\ref{app.meanvar} and \eqref{var.bound}, we have $ | a_{12} - r_{\varepsilon , jk} |  \lesssim  n^{-1} t$ almost surely on $\mathcal{E}_2(t)$. Therefore, the matrix $\Ab$ is positive definite almost surely on $ \mathcal{E}_2(t)$ whenever $ n \gtrsim t$. Let $\bG = (G_1, G_2)^\T$ be a Gaussian random vector with $\EE (\bG) = {\mathbf 0}$ and $\cov(\bG) = \Ab$. Then, applying Theorem~1.1 in \cite{B2005}, a multivariate Berry-Esseen bound, to $n^{-1/2} \sn \bxi_i$ gives
\begin{align}
	& \sup_{x, y \in \RR} | \PP ( s_{j}^{-1} S_j \geq   x,  s_{k}^{-1} S_k \geq  y | \mathcal{F}_n  ) -  \PP ( G_1 \geq  x, G_2 \geq  y  )  |  \nn \\
	& \lesssim \frac{1}{n^{3/2}} \sn \e  (  \| \Ab^{-1/2} \bxi_i \|_2^3 )  \lesssim \frac{1}{\sqrt{n}} +   \frac{1}{n^{3/2}}\sn ( | \bb_j^\T \bbf_i  |^3 +  | \bb_k^\T \bbf_i  |^3 ) \nn
\end{align}
almost surely on $\mathcal{E}_2(t)$. Taking $x=s_j^{-1} \sigma_{\varepsilon , jj}^{1/2}\, z$ and $y=s_k^{-1} \sigma_{\varepsilon ,kk}^{1/2} \,z$ implies
\begin{align}
	 | \e(Y_j Y_k | \mathcal{F}_n)   -  \PP ( G_1 \geq   s_j^{-1} \sigma_{\varepsilon , jj}^{1/2} \, z , \, & G_2 \geq  s_k^{-1} \sigma_{\varepsilon ,kk}^{1/2} \,z  | \mathcal{F}_n )  | \nn \\
	 & \lesssim \frac{1}{\sqrt{n}} +   \frac{1}{n^{3/2}}\sn ( | \bb_j^\T \bbf_i  |^3 +  | \bb_k^\T \bbf_i  |^3 ) . \label{joint.approxi.1}
\end{align}
For the bivariate Gaussian random vector $ (G_1, G_2)^\T$ with $a_{12} = \cov(G_1, G_2)$, it follows from Corollary~2.1 in \cite{LS2002} that, for any $x, y \in \RR$,
\begin{align}
	  |  \PP ( G_1 \geq  x, G_2 \geq  y   ) - \{ 1- \Phi(x) \} \{ 1- \Phi(y) \}  | \leq \frac{|a_{12} |}{4}  \exp\bigg\{ - \frac{x^2 + y^2}{2(1+ |a_{12}|)} \bigg\} \leq \frac{|a_{12} |}{4}  . \nn
\end{align}
This, together with the Gaussian comparison inequality \eqref{Gaussian.comparison} gives
\begin{align}
	 | \PP ( G_1 \geq  s_j^{-1} \sigma_{\varepsilon , jj}^{1/2} \, z ,  G_2 \geq  s_k^{-1} \sigma_{\varepsilon , kk}^{1/2} \,z  | \mathcal{F}_n ) - \Phi(-z)^2 | \lesssim  | r_{\varepsilon , jk} | + n^{-1} t  \label{joint.approxi.2}
\end{align}
almost surely on $\mathcal{E}_2(t)$ as long as $n\gtrsim (K+t)t$.

Consequently, it follows from \eqref{pf1.6}, \eqref{pf1.8}, \eqref{joint.approxi.1}, \eqref{joint.approxi.2} and Assumption~1 that
\begin{align}
	 \e[  \{  p_0^{-1} \widetilde V_+(z)  & -   \Phi(-z)  \}^2  | \mathcal{F}_n ]   \lesssim    p^{ -\kappa_1}     +  n^{-1/2}   (K+ \log n + t)^{1/2}   \label{pf1.11}
\end{align}
almost surely on $\mathcal{E}_2(t)$ as long as $n\gtrsim (K+t)t$. A similar bound can be obtained for $\e [  \{  p_0^{-1} \widetilde V_-(z)    -   \Phi(-z)   \}^2  | \mathcal{F}_n ]$. Recall that $\PP\{ \mathcal{E}_1(t)\cap \mathcal{E}_2(t) \} \geq  1- (p +  4)e^{-t}$ whenever $n \geq 8 t$. Finally, taking $t= \log(np )$ in \eqref{pf1.5} and \eqref{pf1.11} proves \eqref{V.LLN+}.   \qed

\subsection{Proof of Proposition~1}

To begin with, observe that
\#
	\bigg|  \wt T_j  -  \sqrt{\frac{n}{\wt{\sigma}_{\varepsilon, jj}}} ( \hat{\mu}_j -  \bb_j^\T \bar{\bbf} \,)  \bigg|  =  \sqrt{\frac{n}{\wt{\sigma}_{\varepsilon, jj}}} | (\wt{\bb}_j - \bb_j )^\T \bar{\bbf} |   \leq  \sqrt{\frac{n}{\wt{\sigma}_{\varepsilon, jj}}} \| \bar{\bbf} \|_2 \|  \wt{\bb}_j - \bb_j \|_2  , \nn \\
\bigg|  \sqrt{\frac{n}{\wt{\sigma}_{\varepsilon, jj}}} ( \hat{\mu}_j - \bb_j^\T \bar{\bbf}  \,)  - T^\circ _j  \bigg|
  \leq \bigg| \frac{ 1}{ \sqrt{ \wt{\sigma}_{\varepsilon, jj} }  } - \frac{1}{\sqrt{ \sigma_{\varepsilon, jj} } } \bigg| (    |\sqrt{n} \,\hat{\mu}_j|  +  \| \bb_j \|_2 \| \sqrt{n} \bar \bbf \|_2 ) . \nn
\#
By Lemma~\ref{factor.concentration}, $\| \sqrt{n} \bar \bbf \|_2 \lesssim (K + \log n)^{1/2}$ with probability greater than $1-n^{-1}$.
Moreover, it follows from Lemma~\ref{lem3}  that $\max_{1\leq j\leq p}  | \hat{\mu}_j - \mu_j | \lesssim  \{\log (np)\}^{1/2} n^{-1/2}$ with probability at least $1- 2 n^{-1}$.
Putting the above calculations together, we conclude that
\#
	\max_{j \in \mathcal{H}_0 } |  \wt T_j - T_j^\circ  |  \lesssim \frac{\log(np)}{\sqrt{n}}  +  (K + \log n)^{1/2}  \max_{1\leq j\leq p} ( \| \wt{\bb}_j - \bb_j \|_2  + |\wt{\sigma}_{jj} - \sigma_{jj} | )  \nn
\#
with probability at least $1- 3 n^{-1}$. Combining this with the proof of Theorem~1 and condition (17) implies $p_0^{-1}   \wt V(z  ) = 2 \Phi(-z) + o_{\PP}(1)$. Similarly, it can be proved that \eqref{R.LLN+} holds with $R(z)$ replaced by $\wt R(z)$. The conclusion follows immediately. \qed

\subsection{Proof of Theorem~2}
We first note that the $\widehat\bSigma = \widehat\bSigma_U$ defined is a $U$-statistic of order two. For simplicity, let $\cC$ denote the set of $\binom{n}{2}$ distinct pairs $(i_1, i_2)$ satisfying $1\leq i_1<i_2\leq n$. Let $h(\bX_i, \bX_j)=2^{-1}(\bX_i-\bX_j)(\bX_i-\bX_j)^\T$ and $Y_{ij}=\psi_\tau(h(\bX_i, \bX_j))=\tau\psi_{1}(\tau^{-1}h(\bX_i, \bX_j))$, such that
\$
\widehat\bSigma= \frac{1}{ \binom{n}{2}}\sum_{(i,j) \in \cC}Y_{ij}.
\$

We now rewrite the $U$-statistic $\widehat\bSigma$ as an average of dependent averages of identically and independently  distributed random matrices. Define $k=[n/2]$, the greatest integer $\leq n/2$ and define  
\$
W_{(1,\ldots, n)}=   k^{-1} ( Y_{12}+ Y_{23}+\ldots+ Y_{2k-1, 2k} ).
\$
Let $\cP$ denote the class of all $n!$ permutations  of $(1,\ldots, n)$ and  $\pi=(i_1, \ldots, i_n): \{ 1,\ldots, n\} \mapsto \{ 1,\ldots, n\}$ be a permutation, i.e. $\pi(k) =i_k$ for $k=1,\ldots, n$. Then it can be shown that
\$
\widehat\bSigma =\frac{1}{n!}\sum_{\pi\in\cP} W_\pi.
\$
Using the convexity of maximum eigenvalue function $\lambda_{\max}(\cdot)$ along with  the convexity of the exponential function,  we obtain
\$
\exp \{ \lambda_{\max} (\widehat\bSigma -\bSigma  ) / \tau \} \leq \frac{1}{n!} \sum_{\pi\in\cP}\exp \{  \lambda_{\max}  (W_\pi-\bSigma  ) / \tau \}.
\$
Combining this with Chebyshev's inequality delivers
\$
	& \PP\{ \lambda_{\max} (\widehat\bSigma-\bSigma )\geq  t/\sqrt{n}\} \\ 
&=\PP\Big[ \exp\{ \lambda_{\max}(k\widehat\bSigma -k\bSigma  ) /\tau \} \geq  \exp \{kt/(\tau\sqrt{n} \,  ) \} \Big] \\
&\leq e^{-  {kt}/(\tau \sqrt{n}) } \frac{1}{n!} \sum_{\pi\in \cP}\EE\exp  \{ \lambda_{\max}(kW_\pi-k\bSigma  )/\tau \} \\
&\leq e^{-  {kt}/(\tau \sqrt{n}) } \frac{1}{n!} \sum_{\pi\in \cP}\EE \tr\exp\{  (kW_\pi-k\bSigma  ) / \tau \},
\$
where we use the property $e^{\lambda_{\max}(\Ab)}\leq \tr e^\Ab$ in the last inequality. For a given permutation $\pi=( i_1, \ldots, i_n)\in \cP$, we write  $Y_{\pi j}= Y_{i_{2j-1}i_{2j}}$ and $H_{\pi j}=h(\bX_{i_{2j-1}}, \bX_{i_{2j}})$ with $\EE H_{\pi j}=\bSigma$.   We then rewrite $W_\pi$ as
$W_\pi= k^{-1} ( Y_{\pi1}+\ldots+ Y_{\pi k} )$, where $Y_{\pi j}$'s are mutually independent. Before proceeding, we introduce the following lemma whose proof is based on elementary calculations.

\begin{lemma}\label{lemma:log}
For any $\tau>0$ and $x \in \RR$, we have $\psi_\tau(x)=\tau\psi_1(x/\tau)$ and
\$
  -\log (1-x+x^2)\leq \psi_1(x)\leq \log (1+x+x^2) ~\mbox{ for all } x\in \RR.
\$
\end{lemma}

From Lemma \ref{lemma:log} we see that the matrix $Y_{\pi j}$ can be bounded as
\$
-\log ( \Ib_p -H_{\pi j}/\tau+H^2_{\pi j}/\tau^2 ) \leq  Y_{\pi j}  /\tau \leq \log ( \Ib_p + H_{\pi j}/\tau+H_{\pi j}^2/\tau^2 ).
\$
Using this property we can bound $\EE\exp\{  \tr (kW_\pi-k\bSigma  ) / \tau \}$ by
\#\label{thm:utype:eq:1}
&\EE_{[k-1]}\EE_k\tr\exp \Bigg\{ \sum_{j=1}^{k-1}{Y_{\pi j}}- (k/\tau) \bSigma +{Y_{\pi k}}\Bigg\} \nn\\
&\hspace{1cm}\leq \EE_{[k-1]}\EE_k\tr\exp \Bigg\{ \sum_{j=1}^{k-1}{ Y_{\pi j}}- (k/\tau) \bSigma +\log (\Ib_p +H_{\pi j}/\tau+H_{\pi j}^2/\tau^2 )\Bigg\}
\#
To further bound the right-hand side of \eqref{thm:utype:eq:1}, we follow a similar argument as in \cite{M2016}. The following lemma, which is taken from \cite{lieb1973convex}, is commonly referred to as the Lieb's concavity theorem.
\begin{lemma}\label{lemma:lieb}
For any symmetric matrix $H \in \RR^{d\times d}$, the function
\$
f(A)=\tr\exp (H+\log A ) ,  \ \  A   \in \RR^{d \times d}
\$
is concave over the set of all positive definite matrices.
\end{lemma}

Applying Lemma \ref{lemma:lieb} repeatedly along with Jensen's inequality, we obtain
\$
\EE\{ \tr\exp (kW_\pi-k\bSigma  ) / \tau \}
&\leq \EE\tr\exp \Bigg\{\sum_{j=1}^{k-1}{ Y_{\pi j}}-(k/\tau) \bSigma +\log (\Ib_p+\EE H_{\pi k}/\tau+\EE H_{\pi k}^2/\tau^2 )\Bigg\}\\
&\leq \tr \exp\Bigg\{\sum_{j=1}^k\log (\Ib_p +\EE H_{\pi j}/\tau +\EE H_{\pi j}^2/\tau^2 )-( k/\tau ) \bSigma \Bigg\}\\
&\leq \tr \exp\Bigg(\sum_{j=1}^k\EE H_{\pi k}^2/\tau^2\Bigg),
\$
where we use the inequality $\log (1+x)\leq x$ for $x>-1$ in the last step.  The following lemma gives an explicit form for $v^2 :=\|\EE H_{\pi k}^2\|_2$.

\begin{lemma}\label{lemma:v2}
We have
\$
 \| \EE h^2(\bX_1, \bX_2) \|_2 =\frac{1}{2} \Big\|\EE \{(\bX-\bmu)(\bX-\bmu)^\T \}^2+\tr (\bSigma)\bSigma+2\bSigma^2\Big\|.
\$
\end{lemma}
\begin{proof}[Proof of Lemma \ref{lemma:v2}]
Write $\bX= \bX_1$ and $\bY=\bX_2$. Without loss of generality, assume that $\EE(\bX)\!=\!\EE(\bY)\!=\!{\bf 0}$.  Let $H_1=\bX\bX^\T,~H_2=\bY\bY^\T,~H_{12}=\bX\bY^\T$ and $H_{21}=\bY\bX^\T$.  Then
\$
\{(\bX-\bY)(\bX-\bY)^\T \}^2
&=(H_1+H_2-H_{12}-H_{21})^2\\
&=H_1^2+H_2^2+H_{12}^2+H_{21}^2 +H_1H_2+H_2H_1+H_{12}H_{21}+H_{21}H_{12}\\
&~~~~-H_{1}H_{12}-H_{12}H_1-H_1H_{21}-H_{21}H_1-H_2H_{12}-H_{12}H_2\\
&~~~~~~~~-H_2H_{21}-H_{21}H_2,
\$
which, by symmetry, implies that
\$
\EE \{(\bX-\bY)(\bX-\bY)^\T \}^2=2\EE H_1^2+2\EE H_{12}^2 +2\EE H_1H_2+2\EE H_{12}H_{21}.
\$
In the following we calculate the four expectations on the right-hand side of the above equality separately.  For the first term, note that
\$
\EE H_1^2=\EE (\bX\bX^\T\bX\bX^\T ).
\$
Let $A = (A_{jk})=H_{12}^2$ and we have
\$
\EE A_{jk}=\EE\bigg(\sum_{\ell =1}^p X_\ell Y_\ell X_j Y_k\bigg)=\EE\bigg(Y_k \sum_{\ell =1}^p X_j X_\ell Y_\ell \bigg) =\sum_{\ell =1}^p \sigma_{j \ell }\sigma_{\ell k},
\$
where $\sigma_{jk}$ is the $(j,k)$-th entry of $\bSigma$. Therefore, we have $\EE H_{12}^2=\bSigma^2$.  For $\EE H_1H_2$, using independence, we can show that $\EE H_1 H_2=\bSigma^2$. For $\EE H_{12}H_{21}$, we have
\$
\EE H_{12}H_{21}=\EE (\bX\bY^\T \bY\bX^\T )=\EE \{\EE (\bX\bY^\T\bY\bX^\T|\bY )\}=\tr(\bSigma)\bSigma.
\$
Putting the above calculations together completes the proof.
\end{proof}

For any $u>0$, putting the above calculations together and letting $\tau \geq  2v^2\sqrt{n}/u$ yield
\$
& \PP   \{ \lambda_{\max}(\widehat\bSigma-\bSigma  )\geq  u/\sqrt{n}\} \\
& \leq e^{-k u/(\sqrt{n}\tau)} \tr\exp\bigg(\sum_{j=1}^k\EE H_{\pi k}^2/\tau^2  \bigg)
\leq p \exp\bigg(-\frac{k u}{\sqrt{n}\tau}+\frac{k v^2}{\tau^2}\bigg)\\
&\leq p \exp\bigg(-\frac{k u^2}{4n v^2}\bigg)\leq p \exp\bigg(-\frac{u^2}{16 v^2}\bigg),
\$
where we use the fact that $k:= [n/2]\geq  n/4$ for $n\geq  2$ in the last inequality. On the other hand, it can be similarly shown that
\$
\PP \{ \lambda_{\min}(\widehat\bSigma-\bSigma )\leq -u /\sqrt{n}\} \leq p \exp\bigg(-\frac{u^2}{16 v^2}\bigg)
\$

Combining the above two inequalities and putting $u = 4v\sqrt{t}$ complete the proof. \qed

\subsection{Proof of Theorem~3}

First we bound $\max_{1\leq j\leq p}\|  \hat{\bb}_j -  \bb_j \|_2$. For any $t>0$, it follows from Theorem~2 that with probability greater than $1- 2p e^{-t}$, $\| \hat{\bSigma}_U - \bSigma \|  \leq 4v (t/n)^{1/2}$, where $v$ is as in (19).  Define $\widetilde \bb_j  = (\overline{\lambda}_1^{1/2} \hat{v}_{1 j}, \ldots, \overline{\lambda}_K^{1/2} \hat{v}_{K j} )^\T \in \RR^K$, such that $ \| \hat{\bb}_j - \bb_j \|_2 \leq   \| \hat{\bb}_j -  \widetilde \bb_j \|_2 +  \| \widetilde {\bb}_j -  \bb_j \|_2$. By Assumption~2, (20) and (21), we have
\begin{align}
 |  \hat{\lambda}_\ell^{1/2} - \overline{\lambda}_\ell^{1/2} | = |\hat{\lambda}_\ell  - \overline{\lambda}_\ell  | / (\hat{\lambda}_\ell^{1/2}+ \overline{\lambda}_\ell^{1/2})  \lesssim p^{-1/2} ( \| \hat{\bSigma}_U - \bSigma \| + \| \bSigma_\varepsilon \|  )  , \nn \\
 \| \overline{\bv}_\ell \|_{\infty}= \| \overline{\bb}_\ell \|_{\infty} / \| \overline{\bb}_\ell \|_2  \leq \| \Bb \|_{\max} / \| \overline{\bb}_\ell \|_2 \lesssim p^{-1/2}  \nn \\
  \mbox{ and }~  \| \hat{\bv}_\ell \|_\infty  \leq \| \hat{\bv}_\ell -  \overline \bv_\ell \|_2 +\| \overline \bv_\ell \|_\infty \lesssim  p^{-1}   \| \hat{\bSigma}_U - \bSigma \|  + p^{-1/2}  . \nn
\end{align}
On the event $ \{ \| \hat{\bSigma}_U - \bSigma \| \leq 4v (t/n)^{1/2} \}$, it follows that
\#
 |  \hat{\lambda}_\ell^{1/2} - \overline{\lambda}_\ell^{1/2} |  \lesssim    v  \sqrt{t} \,(np)^{-1/2} + p^{-1/2} ~\mbox{ and }~ \| \hat{\bv}_\ell \|_\infty \lesssim p^{-1/2}
\#
as long as $n \geq v^2  p^{-1}  t$. Write $\hat{\bv}_\ell = (\hat{v}_{\ell 1} , \ldots \hat{v}_{\ell p })^\T$. It follows that, with probability at least $1- 2p e^{-t}$,
\begin{align}
  \| \hat{\bb}_j -  \widetilde \bb_j \|_2  =   \bigg\{ \sum_{\ell=1}^K   (  \hat{\lambda}_\ell^{1/2} - \overline{\lambda}_\ell^{1/2}  )^2 \, \hat{v}_{\ell j}^2 \bigg\}^{1/2} \lesssim   p^{-1} ( v\sqrt{t} \, n^{-1/2} +1 )     \nn
\end{align}
for all $1\leq j\leq p$. Similarly,
\begin{align}
  \| \widetilde{\bb}_j -  \bb_j \|_2 = \bigg\{  \sum_{\ell =1}^K \overline{\lambda}_\ell (  \hat{v}_{\ell j}  - \overline v_{\ell j}  )^2 \bigg\}^{1/2} \leq  \max_{1\leq \ell \leq K }   \overline{\lambda}_\ell^{1/2} \cdot  \sqrt{K} \, \| \hat{\bv}_\ell - \overline{\bv}_\ell \|_\infty \lesssim  v \sqrt{t} \, (np)^{-1/2} + p^{-1/2}  . \nn
\end{align}
By taking $t=\log(np)$, the previous two displays together imply (22).

Next we consider $\max_{1\leq j\leq p} |\hat  {\sigma}_{ \varepsilon ,jj} - \sigma_{\varepsilon , jj}|$. Note that with probability at least $1- 4pe^{-t}$, $\max_{1\leq j\leq p} |\hat{\theta}_j - \EE (X_j^2) |   \lesssim (t/n)^{1/2}$ as long as $n\gtrsim t$. Therefore, it suffices to focus on $\|  \hat \bb_j \|_2^2 - \|  \bb_j \|_2^2$, which can be written as
$  \sum_{\ell=1}^K  ( \hat{\lambda}_\ell - \overline{\lambda}_\ell )  \hat{v}_{\ell j}^2 + \sum_{\ell=1}^K \overline{\lambda}_\ell  (  \hat{v}_{\ell j}^2 - \overline{v}_{\ell j}^2  )$. Under Assumption~2, it follows from (20) and (21) that on the event $\{ \| \hat{\bSigma}_U - \bSigma \| \leq 4v  (t/n)^{1/2}  \}$,
\begin{align}
 & | \| \hat\bb_j \|_2^2 - \|  \bb_j \|_2^2 | \nn \\
 & \leq \sum_{\ell=1}^K | \hat{\lambda}_\ell - \overline{\lambda}_\ell | \| \hat{\bv}_{\ell } \|_\infty^2 + \sum_{\ell=1}^K \overline{\lambda}_\ell ( \| \hat{\bv}_\ell \|_\infty + \| \overline{\bv}_\ell \|_\infty ) \| \hat{\bv}_\ell - \overline{\bv}_\ell \|_\infty  \nn \\
 & \lesssim  v \sqrt{t} \, (np)^{-1/2} + p^{-1/2}    \nn
\end{align}
as long as $n\geq  v^2  p^{-1} t$, which proves (23) by taking $t=\log(np)$. \qed

\subsection{Proof of Theorem~4}

For $\hat{\mu}_j$'s and $\hat{\theta}_{jk}$'s with $\tau_j = a_j (n/t_1)^{1/2}$ and $\tau_{jk} = a_{jk} (n/t_2)^{1/2}$, it follows from Lemma~\ref{lem3} and the union bound that as long as $n\geq  8\max(t_1, t_2)$,
\begin{align}
	\max_{1\leq j\leq p} | \hat{\mu}_j - \mu_j | \leq 4\max_{1\leq j\leq p} a_j   \sqrt{\frac{t_1}{n}}  ~\mbox{ and }~  \max_{1\leq j\leq k \leq p} | \hat{\theta}_{jk} - \e(X_j X_k) |   \leq  4\max_{1\leq j\leq k \leq p} a_{jk}  \sqrt{\frac{t_2}{n}} \nn
\end{align}
with probability at least $1- 2pe^{-t_1} - (p^2+p)e^{ -t_2 }$. In particular, taking $t_1=\log(n p)$ and $t_2 = \log( n p^2)$ implies that as long as $n \gtrsim \log(np)$, $\| \hat{\bSigma}_{{\rm H}} - \bSigma \|_{\max} \lesssim w_{n,p}^{-1}$ with probability greater than $1- 4n^{-1}$.

The rest of the proof is similar to that of Theorem~3, simply with the following modifications. Under Assumption~2, it follows from \eqref{eigenvalue.perturb} and \eqref{eigenvector.perturb} in Lemma~\ref{infinity.perturbation} that, with probability at least $1- 4n^{-1}$,
\begin{align}
 |  \wt{\lambda}_\ell^{1/2} - \overline{\lambda}_\ell^{1/2} | = |\wt{\lambda}_\ell  - \overline{\lambda}_\ell  | / (\wt{\lambda}_\ell^{1/2}+ \overline{\lambda}_\ell^{1/2}) \lesssim  \sqrt{p} \,( w_{n,p}^{-1} + p^{-1} )  , \nn \\
 \| \overline{\bv}_\ell \|_{\infty}= \| \overline{\bb}_\ell \|_{\infty} / \| \overline{\bb}_\ell \|_2  \leq \| \Bb \|_{\max} / \| \overline{\bb}_\ell \|_2  \lesssim p^{-1/2} , \nn \\
  \| \wt{\bv}_\ell - \overline \bv_\ell \|_\infty \lesssim   p^{-1/2}w_{n,p}^{-1} + p^{-1}  ~\mbox{ and }~  \| \wt{\bv}_\ell \|_\infty \lesssim   p^{-1/2}  . \nn
\end{align}
Plugging the above bounds into the proof of Theorem~3 proves the conclusions. \qed

\subsection{Proof of Theorem~5}
\label{proof.thm2}

The key of the proof is to show that ${T}_j(\Bb)$ provides a good approximation of $T_j^{\circ}$ uniformly over $1\leq  j\leq p$. To begin with, note that the estimator $\hat{\theta}_j$ with $\tau_{jj} = a_{jj} (n/t)^{1/2}$ for $a_{jj} \geq  \var(X_j^2)^{1/2}$ satisfies $\PP \{ | \hat{\theta}_j - \theta_j  | \geq  4a_{jj} (t/n)^{1/2} \} \leq 2e^{-t}$, where $\theta_j = \EE(X_j^2)$. Together with the union bound, this yields that with probability greater than $1- 2p  e^{-t}$,
\#
  \max_{1\leq j\leq p} | \hat{\theta}_j - \theta_j  | \leq 4 \max_{1\leq j \leq p} a_{jj}^{1/2} \sqrt{\frac{t}{n}}  \label{var.uniform.conv}
\#
as long as $n \geq  8 t$. Next, observe that
\#
	\bigg|  {T}_j(\Bb)  -  \sqrt{\frac{n}{\hat{\sigma}_{\varepsilon, jj}}} ( \hat{\mu}_j - \bb_j^\T \bar{\bbf} \,)  \bigg|  =  \sqrt{\frac{n}{\hat{\sigma}_{\varepsilon, jj}}} | \bb_j^\T  \{ \bar{\bbf} - \hat{\bbf}(\Bb) \} |   \leq  \sqrt{\frac{n}{\hat{\sigma}_{\varepsilon, jj}}} \| \bb_j \|_2 \| \hat{\bbf}(\Bb) - \bar{\bbf} \|_2  \label{Tj.approx.1}
\#
and
\#
\bigg|  \sqrt{\frac{n}{\hat{\sigma}_{\varepsilon, jj}}} ( \hat{\mu}_j - \bb_j^\T \bar{\bbf}  \,)  - T^\circ _j  \bigg|
\leq \bigg| \frac{ 1}{ \sqrt{ \hat{\sigma}_{\varepsilon, jj} }  } - \frac{1}{\sqrt{ \sigma_{\varepsilon, jj} } } \bigg| (    |\sqrt{n} \,\hat{\mu}_j|  +  \| \bb_j \|_2 \| \sqrt{n} \bar \bbf \|_2 ) .  \label{Tj.approx.2}
\#
Applying Proposition~3 with $t=\log n$ shows that, with probability at least $1-C_1 n^{-1}$,
\#
 	\| \hat{\bbf} (\Bb) - \bar{\bbf} \|_2 \lesssim   ( K \log n )^{1/2} p^{-1/2} . \label{Tj.approx.3}
\#
Moreover, it follows from Lemma~\ref{lem3}, \eqref{factor.concentration.ineq} and \eqref{var.uniform.conv} that, with probability greater than $1- 4pe^{-t_1} - e^{-t_2}$,
\#
	\max_{1\leq j\leq p}  | \hat{\mu}_j - \mu_j | \lesssim \sqrt{\frac{t_1}{n}},  \ \  \max_{1\leq j\leq p} \bigg| \frac{\hat{\sigma}_{\varepsilon, jj}}{\sigma_{\varepsilon,jj}} -1 \bigg| \lesssim \sqrt{\frac{t_1}{n}} ~\mbox{ and }~ \|\bar{\bbf}  \|_2 \lesssim \sqrt{\frac{ K + t_2 }{n}}. \nn
\#
Taking $t_1 = \log(np)$ and $t_2 =\log n$, we deduce from \eqref{Tj.approx.1}--\eqref{Tj.approx.3} that, with probability at least $1-C_2  n^{-1} $,
\#
	\max_{j \in \mathcal{H}_0 } |  T_j(\Bb) - T_j^\circ  |  \lesssim   \{ K  +  \log(np) \} n^{-1/2} +   ( K n\log n)^{1/2} p^{-1/2} . \label{Tj.unif.approxi}
\#
Based on \eqref{Tj.unif.approxi}, the rest of the proof is almost identical to that of Theorem~1 and therefore is omitted. \qed

\subsection{Proof of Theorem~\ref{thm.all}}

For convenience, we write $\hat{\bb}_j = \hat{\bb}_j(\mathcal{X}_1)$ for $j=1,\ldots,p$, which are the estimated loading vectors using the first half of the data. Let $\hat{\bbf}(\mathcal{X}_2)$ be the estimator of $\bar{\bbf}$ obtained by solving (26) using only the second half of the data and with $\bb_j$'s replaced by $\hat{\bb}_j$'s.

We keep the notation used in Section~3.2.2, but with all the estimators constructed from $\mathcal{X}_1$ instead of the whole data set. Recall that $\hat{\Bb} = (\hat{\bb}_1, \ldots, \hat{\bb}_p)^\T = (\wt \lambda_1^{1/2} \hat \bv_1, \ldots, \wt \lambda_K^{1/2} \hat \bv_K )$. Following the proof of Theorem~4, we see that as long as $n\gtrsim \log(np)$, the event $\mathcal{E}_{\max} :=  \{ \| \hat{\bSigma}_{{\rm H}} - \bSigma\|_{\max} \lesssim w_{n,p}^{-1}  \}$ occurs with probability at least $1-4n^{-1}$. On $\mathcal{E}_{\max}$, we have
$$
	\max_{1\leq \ell \leq K} | \wt \lambda_\ell^{1/2} - \overline{\lambda}^{1/2}_\ell | \lesssim \sqrt{p} \,( w_{n,p}^{-1}+ p^{-1} )   ~\mbox{ and }~ \max_{1\leq \ell \leq K } \| \hat \bv_\ell \|_\infty \lesssim p^{-1/2} ,
$$
which, combined with the pervasiveness assumption $\overline{\lambda}_\ell \asymp p$, implies $ \max_{1\leq \ell \leq K }  \wt \lambda_\ell^{1/2}  \lesssim \sqrt{p}$. Moreover, write $\bdelta_j= \hat \bb_j -   \bb_j$ for $1\leq j\leq p$ and note that
\#
  \hat \Bb^\T \hat \Bb - \Bb^\T \Bb = \sum_{j=1}^p (  \hat \bb_j \hat \bb_j^\T  - \bb_j  \bb_j^\T ) = \sum_{j=1}^p \bdelta_j \bdelta_j^\T + 2 \sum_{j=1}^p \bdelta_j   \bb_j^\T .\nn
\#
It follows that $ \| p^{-1} ( \hat \Bb^\T \hat \Bb -   \Bb^\T \Bb ) \| \leq \max_{1\leq j\leq p} ( \| \bdelta_j \|_2^2 + 2 \| \bb_j \|_2 \| \bdelta_j \|_2  )$. Again, from the proof of Theorem~4 we see that on the event $\mathcal{E}_{\max}$, $ \| p^{-1} ( \hat \Bb^\T \hat \Bb -   \Bb^\T \Bb )\| \lesssim w_{n,p}^{-1} + p^{-1/2}$. Under Assumption~3, putting the above calculations together yields that with probability greater than $1-4n^{-1}$,
$$
	\lambda_{\min}( p^{-1} \hat \Bb^\T \hat \Bb ) \geq  \frac{c_l}{2} ~\mbox{ and }~ \| \hat \Bb \|_{\max} \leq C_1
$$
as long as $n \gtrsim \log(np)$. By the independence between $\hat{\bb}_j$'s and $\mathcal{X}_2$, the conclusion of Proposition~3 holds for $\hat{\bbf}(\mathcal{X}_2)$.

Next, recall that
\begin{align}
	 {T}_j =  \sqrt{  \frac{n}{\hat \sigma_{\varepsilon , jj}} }   \{  \hat{\mu}_j - \hat{\bb}_j^\T \hat{\bbf}(\mathcal{X}_2)  \},   \nn
\end{align}
where $\hat{\mu}_j$'s and $\hat{\sigma}_{\varepsilon, jj}$'s are all constructed from $\mathcal{X}_2$. Note that
\$
 & |  \sqrt{n} \{ \hat{\mu}_j - \hat{\bb}_j^\T \hat{\bbf}(\mathcal{X}_2) \} - \sqrt{n} \{ \hat{\mu}_j - {\bb}_j^\T \bar{\bbf} \} |  \leq \sqrt{n}  \| \hat{\bb}_j \|_2 \|   \hat{\bbf}(\mathcal{X}_2)  - \bar{\bbf} \|_2  +  \sqrt{n} \| \bar{\bbf} \|_2 \|  \hat{\bb}_j - \bb_j \|_2.
\$
This, together with (28), Theorem~4 and \eqref{factor.concentration.ineq}, implies that with probability at least $1- C_2 n^{-1}$,
\$
& \max_{1\leq j\leq p}|  \sqrt{n} \{ \hat{\mu}_j -  \hat{\bb}_j^\T \hat{\bbf}(\mathcal{X}_2) \} - \sqrt{n} \{ \hat{\mu}_j - {\bb}_j^\T \bar{\bbf} \} | \\
&  \quad \quad \quad  \lesssim     ( K n \log n)^{1/2} p^{-1/2} +  ( K + \log n)^{1/2} ( w_{n,p}^{-1} + p^{-1/2} ).
\$
Following the proof of Theorem~5, it can be shown that with probability at least $1- C_3 n^{-1}$,
\#
	\max_{ j \in \mathcal{H}_0} | {T}_j - T^\circ_j | \lesssim    ( K  n \log n)^{1/2} p^{-1/2} +  \{ K + \log(np)\} n^{-1/2} . \nn
\#
The rest of the proof is almost identical to that of Theorem~1 and therefore is omitted. \qed

\section{Additional proofs}
\label{app:B}

In this section, we prove Propositions~2 and 3 in the main text, and Lemmas~\ref{lem2.1}--\ref{infinity.perturbation} in Section~\ref{sec.proof}.

\subsection{Proof of Proposition~2}

By Weyl's inequality and the decomposition that $\hat{\bSigma} = \Bb \Bb^\T + (\hat{\bSigma} - \bSigma) + \bSigma_\varepsilon$, we have
$$
	\max_{1\leq \ell\leq K} | \hat{\lambda}_\ell -  \overline{\lambda}_\ell  | \leq \| \hat{\bSigma} - \bSigma \|_2 + \| \bSigma_\varepsilon \|_2 ~\mbox{ and }~ \max_{ K +1\leq \ell \leq p} | \hat{\lambda}_\ell    | \leq \| \hat{\bSigma} - \bSigma \|_2   + \| \bSigma_\varepsilon \|_2 ,
$$
where $\hat \lambda_1, \ldots, \hat \lambda_p$ are the eigenvalues of $\hat{\bSigma}$ in a non-increasing order. Thus, (20) follows immediately. Next, applying Corollary~1 in \cite{YWS15} to the pair $(\hat{\bSigma}, \Bb \Bb^\T)$ gives that, for every $1\leq \ell \leq K$,
\begin{align}
	  \| \hat{\bv}_\ell - \overline{\bv}_\ell \|_2 \leq  \frac{2^{3/2} \|  (\hat{\bSigma} - \bSigma) + \bSigma_\varepsilon \|_2 }{\min( \overline \lambda_{\ell-1}  -  \overline \lambda_\ell   ,   \overline \lambda_{\ell }   -   \overline \lambda_{\ell+1}  )  } , \nn
\end{align}
where we put $\overline \lambda_0 = \infty$ and $\overline \lambda_{K + 1}=0$. Under Assumption~2, this proves (21). \qed

\subsection{Proof of Proposition~3}

To begin with, we introduce the following notation. Define the loss function $L_\gamma(\bw) = p^{-1} \sum_{j=1}^p \ell_\gamma( \bar{X}_j -  {\bb}^\T_j \bw)$ for $\bw \in \RR^K$, $\bw^* = \bar{\bbf}$ and $\hat{\bw} = \argmin_{\bw \in \RR^K} L_\gamma(\bw)$. Without loss of generality, we assume $\| \Bb \|_{\max} \leq 1$ for simplicity.

Define an intermediate estimator $\hat{\bw}_\eta =\bw^*+ \eta ( \hat{\bw} - \bw^*)$ such that $\|\hat \bw_\eta - \bw^*\|_2\leq r$ for some $r>0$ to be specified below \eqref{L2.score.bound}. We take $\eta =1$ if $\|\hat \bw - \bw^*\|_2\leq r$; otherwise, we choose $\eta \in(0,1)$ so that  $\|\hat \bw_\eta  -\bw^*\|_2=r$. Then, it follows from Lemma~A.1 in \cite{sun2016adaptive} that
\begin{align}
 \langle \nabla L_\gamma(\hat \bw_\eta)-\nabla L_\gamma( \bw^*), \hat \bw_\eta - \bw^* \rangle \leq \eta \langle \nabla L_\gamma(\hat \bw)-\nabla L_\gamma( \bw^*), \hat \bw - \bw^* \rangle ,
\end{align}
where $\nabla L_\gamma(\hat \bw) = \mathbf{0}$ according to the Karush-Kuhn-Tucker condition. By the mean value theorem for vector-valued functions, we have
$$
	\nabla L_\gamma(\hat\bw_\eta)-\nabla L_\gamma(\bw^*) = \int_0^1 \nabla^2 L_\gamma( (1-t) \bw^* + t \hat\bw_\eta) \, dt   \, (\hat\bw_\eta - \bw^*).
$$
If, there exists some constant $a_{\min}>0$ such that
\begin{align}
 \min_{\bw \in \RR^K : \| \bw - \bw^* \|_2 \leq r }\lambda_{\min} ( \nabla^2 L_\gamma( \bw ) ) \geq  a_{\min} , \label{Hessian.lbd}
\end{align}
then it follows $a_{\min}\|\hat\bw_\eta -\bw^*\|_2^2\leq -  \eta \langle \nabla L_\gamma (\bw^*) , \hat\bw - \bw^* \rangle\leq \| \nabla L_\gamma(\bw^*)\|_2\|\hat\bw_\eta -\bw^*\|_2$, or equivalently,
\begin{align}
	a_{\min}\|\hat\bw_\eta -\bw^*\|_2 \leq  \| \nabla L_\gamma(\bw^*)\|_2, \label{est.error.ubd1}
\end{align}
where $\nabla L_\gamma(\bw^*) = - p^{-1} \sum_{j=1}^p \psi_\gamma(   \mu_j + \bar{\varepsilon}_j )  {\bb}_j$.

First we verify \eqref{Hessian.lbd}. Write $\mathbf{S} = p^{-1} \Bb^\T \Bb$ and note that
$$
	\nabla^2 L_\gamma(\bw) = \frac{1}{p} \sum_{j=1}^p \bb_j \bb_j^\T  I(|  \bar{X}_j - \bb^\T_j \bw |\leq \gamma ) ,
$$
where $\bar{X}_j - \bb^\T_j \bw = \bb^\T_j(\bw^* - \bw) + \mu_j + \bar{ \varepsilon }_j$. Then,  for any $\bu \in \mathbb{S}^{K -1}$ and $\bw \in \RR^K$ satisfying $\| \bw - \bw^* \|_2\leq r$,
\begin{align}
	& \bu^\T  \nabla^2 L_\gamma(\bw) \bu   \nn \\
	& \geq    \bu^\T  \mathbf{S} \bu - \frac{1}{p} \sum_{j=1}^p  (\bb_j^\T \bu)^2 I( |\bar{\varepsilon}_j + \mu_j |>\gamma/2  ) - \frac{1}{p} \sum_{j=1}^p    (\bb_j^\T \bu)^2 I\{ | \bb_j^\T (\bw^* - \bw ) |  > \gamma/2 \} \nn \\
	& \geq   \bu^\T  \mathbf{S} \bu  - \max_{1\leq j\leq p} \| \bb_j \|_2^2 \, \bigg\{ \frac{1}{p} \sum_{j=1}^p   I( |\bar{\varepsilon}_j + \mu_j |>\gamma/2  ) +  \frac{4}{\gamma^2} \| \bw - \bw^* \|_2^2 \,  \bu^{{\rm T}}  \mathbf{S} \bu \bigg\} . \nn
\end{align}
By Assumption~3, $ \lambda_{\min}( \mathbf{S} )  \geq  c_l$ for some constant $c_l >0$ and $ \max_{1\leq j\leq p} \| \bb_j \|_2^2 \leq  K$. Therefore, as long as $\gamma > 2 r \sqrt{K}$ we have
\begin{align}
  \min_{\bw \in \RR^K: \| \bw - \bw^* \|_2\leq r } \lambda_{\min}(  \nabla^2 L_\gamma(\bw) )  \geq  (  1 - 4 \gamma^{-2} r^2 K ) c_l - \frac{K}{p}\sum_{j=1}^p   I( |\bar{\varepsilon}_j + \mu_j |>\gamma/2  ), \label{min.eigenvalue-1}
\end{align}
To bound the last term on the right-hand side of \eqref{min.eigenvalue-1}, it follows from Hoeffding's inequality that for any $t >0$,
\#
	\frac{1}{p} \sum_{j=1}^p  I ( |\bar{\varepsilon}_j + \mu_j |>\gamma/2 ) \leq \frac{1}{p} \sum_{j=1}^p \PP ( |\bar{\varepsilon}_j + \mu_j |>\gamma/2 ) + \sqrt{\frac{t}{2p}} \nn
\#
with probability at least $1- e^{-t}$.
This, together with \eqref{min.eigenvalue-1} and the inequality
\begin{align}
	\frac{1}{p} \sum_{j=1}^p \PP ( |\bar{\varepsilon}_j + \mu_j |>\gamma/2 ) \leq \frac{4}{\gamma^2 p } \sum_{j=1}^p  ( \mu_j^2 + \e  \bar{\varepsilon}_j^2 )  =  4 \gamma^{-2}  ( p^{-1} \| \bmu \|_2^2 + n^{-1} \overline{\sigma}_\varepsilon^2 )   \nn
\end{align}
implies that, with probability greater than $1- e^{-t}$,
\#
	  \min_{\bw \in \RR^K: \| \bw - \bw^* \|_2\leq r } \lambda_{\min}(  \nabla^2 L_\gamma(\bw) )  \geq   \frac{3}{4} c_l -  K \sqrt{\frac{t}{2p }} - \frac{4 K }{\gamma^2 }  \bigg( \frac{ \| \bmu \|_2^2}{p} + \frac{ \overline{\sigma}_\varepsilon^2 }{n} \bigg)    \label{min.eigenvalue-2}
\#
as long as $\gamma \geq  4 r \sqrt{K}$.

Next we bound $\| \nabla L_\gamma(\bw^*)\|_2$. For every $1\leq \ell \leq K$, we write $\Psi_\ell = p^{-1}  \sum_{j=1}^p \psi_{j \ell } :=  p^{-1} \sum_{j=1}^p \gamma^{-1} \psi_\gamma(  \mu_j + \bar{ \varepsilon }_j )  b_{j \ell} $, such that $\| \nabla L_\gamma(\bw^*)\|_2 \leq \sqrt{K}  \, \| \nabla L_\gamma(\bw^*)\|_\infty = \gamma \sqrt{K} \, \max_{1\leq \ell \leq d} |\Psi_\ell |$. Recall that, for any $u\in \RR$, $- \log( 1 - u +  u^2 ) \leq   \gamma^{-1} \psi_\gamma(\gamma u ) \leq \log( 1 + u + u^2 )$. After some simple algebra, we obtain that
\begin{align}
	 e^{\psi_{   j \ell }}  & \leq \{ 1+ \gamma^{-1}( \mu_j + \bar{\varepsilon}_j ) + \gamma^{-2}(\mu_j + \bar{\varepsilon}_j)^2 \}^{b_{j\ell} I(b_{j\ell} \geq  0)}  \nn \\
	 & \quad~ + \{ 1- \gamma^{-1} ( \mu_j + \bar{\varepsilon}_j ) + \gamma^{-2} (\mu_j + \bar{\varepsilon}_j)^2 \}^{-b_{j\ell} I(b_{j\ell} < 0)} \nn \\
&	\leq 	1 + \gamma^{-1} ( \mu_j + \bar{\varepsilon}_j )b_{j\ell} + \gamma^{-2 }   ( \mu_j + \bar{\varepsilon}_j )^2  .
	\nn
\end{align}
Taking expectation on both sides gives
\begin{align}
 \e ( e^{\psi_{j \ell}  } )     \leq 1 +  \gamma^{-1}  |\mu_j |   +  \gamma^{-2} ( \mu_j^2 + n^{-1}\sigma_{\varepsilon, jj} )  .  \nn
\end{align}
Moreover, by independence and the inequality $1+t \leq e^t$, we get
\#
 \e ( e^{ p \Psi_\ell } )  & = \prod_{j=1}^p \e ( e^{\psi_{j \ell}} )  \leq \exp \bigg\{  \frac{1}{\gamma}\sum_{j=1}^p  |\mu_j |   + \frac{1}{\gamma^2 }\sum_{j=1}^p   \bigg( \mu_j^2 +  \frac{ \sigma_{\varepsilon, jj} }{n} \bigg)  \bigg\}  \nn \\
 & \leq   \exp\bigg(   \frac{\| \bmu \|_1 }{\gamma} + \frac{\| \bmu \|_2^2 }{\gamma^2} + \frac{ \overline{\sigma}_\varepsilon^2 \, p }{ \gamma^2 n  } \bigg) . \nn
\#
For any $t >0$, it follows from Markov's inequality that
\$
	\PP ( p \Psi_{j  } \geq     2t ) \leq e^{-   2t }  \e ( e^{   p \Psi_\ell } )  \leq  \exp \bigg\{   \frac{\| \bmu \|_1 }{ \gamma} + \frac{\| \bmu \|_2^2 }{ \gamma^2} +   \frac{   \overline{\sigma}_\varepsilon^2 \, p  }{ \gamma^2 n }  - 2t   \bigg\} \leq \exp(1- t)
\$
provided
\#
	\gamma \geq  \max\Bigg\{ \| \bmu \|_1  ,    \overline{\sigma}_\varepsilon \sqrt{ \frac{    \| \bmu \|_2^2 / \overline{\sigma}_\varepsilon^2 +  p/n }{t} }   \Bigg\} .   \label{gamma.lower.bound.1}
\#	
Under the constraint \eqref{gamma.lower.bound.1}, it can be similarly shown that $\PP ( - p \Psi_{j  } \geq    2t ) \leq e^{1 -t}$. Putting the above calculations together, we conclude that
\#
 & \PP\bigg\{  \| \nabla L_\gamma(\bw^*)\|_2 \geq   \sqrt{K} \,  \frac{2\gamma t  }{p} \bigg\} \nn \\
 &  \leq \PP\bigg\{ \| \nabla L_\gamma(\bw^*)\|_\infty \geq     \frac{  2\gamma  t}{p}  \bigg\}  \leq \sum_{\ell = 1}^K \PP (  | p \Psi_\ell | \geq   2 t ) \leq 2e K \exp( - t) . \label{L2.score.bound}
\#

With the above preparations, now we are ready to prove the final conclusion. It follows from \eqref{min.eigenvalue-2} that with probability greater than $1-  e^{-t}$, \eqref{Hessian.lbd} holds with $a_{\min} = c_l /4$, provided that $\gamma \geq  4  \sqrt{K} \max\{ r,  c_l^{-1/2}  (\| \bmu \|_2^2/p + \overline{\sigma}^2_{\varepsilon} /n )^{1/2}\}$
and $p\geq    8 c_l^{-2} K^2 t $. Hence, combining \eqref{est.error.ubd1} and \eqref{L2.score.bound} with $r = \frac{\gamma}{4\sqrt{K}}$ yields that, with probability at least $1- (1+2e K ) e^{-t}$, $\|  \hat{\bw }_\eta - \bw^* \|_2 \leq 8 c_l^{-1} \sqrt{K} \, p^{-1} \gamma t < r$ as long as
$p > 32 c_l^{-1} K t $.
By the definition of $\hat{\bw}_\eta$, we must have $\eta = 1$ and thus $\hat{\bw} = \hat{\bw}_\eta$. \qed

\subsection{Proof of Lemma~\ref{lem2.1}} \label{secC3}

Let $1\leq j\leq p$ be fixed and define the function $h(\theta)=\e \{ \ell_\tau(X_j -\theta)\}$, $\theta\in \RR$. By the optimality of $\mu_{j,\tau}$ and the mean value theorem, we have $h'(\mu_{j,\tau}) = 0$ and
\begin{align}
	  h''(\widetilde{\mu}_{j,\tau}) (\mu_j - \mu_{j,\tau}) =  h' (\mu_j) -  h'(\mu_{j,\tau})  = h'(\mu_j) = - \e \{ \psi_\tau( \xi_j ) \} ,   \label{lem2_p1}
\end{align}
where $\widetilde{\mu}_{j,\tau} =\lambda \mu_j + (1-\lambda)\mu_{j,\tau}$ for some $0\leq \lambda \leq 1$. Since $\e(\xi_j)=0$, we have $- \e \{  \psi_\tau( \xi_j) \} = \e\{ \xi_j I(| \xi_j | >\tau)  - \tau I( \xi_j > \tau) + \tau I( \xi_j < -\tau) \}$, which implies
\begin{align}
	|\e \{ \psi_\tau( \xi_j) \} | \leq  \tau^{1-\kappa}  \upsilon_{\kappa, j} .  \label{lem2_p2}
\end{align}

Next we consider $h''(\widetilde{\mu}_{j,\tau})  = \PP (|X_j -\widetilde \mu_{j,\tau} |\leq  \tau )$. Since $h$ is a convex function that is minimized at $\mu_{j,\tau}$, $h(\widetilde{\mu}_{j,\tau} )\leq \lambda h(\mu_j) + (1-\lambda) h(\mu_{j,\tau} )\leq h(\mu_j ) \leq   \sigma_{jj} /2$. On the other hand, note that $h(\theta) \geq  \e \{  ( \tau |X_j - \theta| -  \tau^2/2 )1(|X_j -\theta|>\tau )\}$ for all $\theta \in \RR$. Combining these upper and lower bounds on $h(\widetilde \mu_{j,\tau} )$ with Markov's inequality gives
\begin{align}
	& \tau \e \{ | X_j - \widetilde \mu_{j,\tau} |  I(| X_j -\widetilde \mu_{j,\tau} |>\tau ) \} \nn \\
	&  \leq \frac{1}{2} \tau^2 \PP (| X_j - \widetilde \mu_{j,\tau} |>\tau ) + \frac{1}{2} \sigma_{jj}  \leq  \frac{1}{2} \tau  \,\e \{ |X_j -\widetilde \mu_{j,\tau} | I(| X_j - \widetilde \mu_{j,\tau} |>\tau) \}  + \frac{1}{2} \sigma_{jj} , \nn
\end{align}
which further implies that for every $0\leq \lambda \leq 1$,
\begin{align}
	\PP ( |X_j - \widetilde \mu_{j,\tau} | >\tau ) \leq \tau^{-1} \e \{ |X_j - \widetilde \mu_{j,\tau} | 1 ( |X_j - \widetilde \mu_{j,\tau} |>\tau ) \} \leq   \sigma_{jj} \tau^{-2}. \nn
\end{align}
Together with \eqref{lem2_p1} and \eqref{lem2_p2}, this proves \eqref{lem2.1}.  \qed

\subsection{Proof of Lemma~\ref{lemBR}} \label{secC4}

Throughout the proof, we let $1\leq j \leq p$, $a \geq   \sigma_{jj}^{1/2}$, $t\geq  1$ be fixed and write $\tau = a(n/t)^{1/2}$ with $n\geq  8t$. The dependence of $\tau$ on $(a , n,t)$ will be assumed without displaying. First we introduce the following notations. Define functions $L(\theta) = - \sn \ell_\tau(X_{ij} - \theta)$, $\zeta(\theta) = L(\theta) - \e L(\theta)$ and $w^2(\theta) = - \frac{d^2}{d \theta^2} \e L(\theta)$, such that $\hat{\mu}_j = \argmax_{\theta\in \RR} L(\theta)$. Moreover, we write
\# \label{w0.def}
	w_0^2 :=   w^2(\mu_j)  =  \alpha_\tau n \ \ \mbox{ with } \ \  \alpha_\tau=\PP (| X_j - \mu_j | \leq \tau ).
\#

For every $r>0$, define the parameter set
\#
	\Theta_0(r) = \{  \theta \in \RR : | w_0 ( \theta - \mu_j ) | \leq r \} .
\#
Then, it follows from Lemma~\ref{lem3} that
\#  \label{concentration.MLE}
	\PP \{ \hat{\mu}_j \in  \Theta_0(r_0) \}  \geq  1-2\exp(-t) ,
\#
where $r_0 = 4 a (\alpha_\tau t)^{1/2}$. Based on this result, we only need to focus on the local neighborhood $\Theta_0(r_0)$ of $\mu_j$. The rest of the proof is based on Proposition~3.1 in \cite{S2013}. To this end, we need to check Conditions~($\mathcal{L}_0$) and ($ED_2$) there.

\noindent
{\sc Condition ($\mathcal{L}_0$):} Note that, for every $\theta\in \Theta_0(r)$,
\begin{align}
	  |  w_0^{-1} w^2(\theta) w_0^{-1}  -1  |   & = |   \alpha_\tau^{-1} -1   -   \alpha_\tau^{-1}  \PP ( |X_j -\theta | > \tau )    | \nn \\
	& \leq   \alpha_\tau^{-1}  \max [  1 - \alpha_\tau  ,    \{ \sigma_{jj} + (\theta - \mu_j)^2 \} \tau^{-2} ] . \nn
\end{align}
By Chebyshev's inequality, we have $1\geq  \alpha_\tau \geq  1-\sigma_{jj} \tau^{-2}  \geq  7/8$. Therefore,
$$
	| w_0^{-1} w^2(\theta) w_0^{-1}  -1 | \leq \alpha_\tau^{-1}  \{  \sigma_{jj} +   ( \alpha_\tau n)^{-1}  r^2 \}  \tau^{-2}.
$$
This verifies Condition~$(\mathcal{L}_0)$ by taking
\#
	\delta(r) =  \alpha_\tau^{-1} \sigma_{jj} \tau^{-2} +  \alpha_\tau^{-2}     \tau^{-2} n^{-1}   r^2  , \ \ r > 0. \nn
\#

\noindent
{\sc Condition ($ED_2$):} Note that $\zeta''(\theta) =- \sn \{ 1(|X_{ij} - \theta |\leq \tau) - \PP(|X_{ij} - \theta |\leq \tau)\}$. For every $\lambda \in \RR $ satisfying $|\lambda |\leq \alpha_\tau \sqrt{n}$, using the inequalities $1+ u \leq e^u$ and $e^u \leq 1 + u + u^2 e^{u \vee 0}/2$ we deduce that
\begin{align}
	  \e \exp \{  \lambda \sqrt{n}   \zeta''(\theta)  / w_0^2    \}
	& = \prod_{i=1}^n \e   \exp [ - \lambda w_0^{-2}\sqrt{n} \{  I(|X_{ij} - \theta |\leq \tau ) - \PP (|X_{ij} - \theta |\leq \tau ) \} ] \nn \\
	& \leq \prod_{i=1}^n  \{ 1 +  (1/2)  \lambda^2 w_0^{-4}  n    \exp ( |\lambda| w_0^{-2} \sqrt{n} )  \} \nn \\
	& \leq  \prod_{i=1}^n \{ 1 +  (e/2) \alpha_\tau^{-2} \lambda^2     n^{-1} \}  \leq \exp \{ (e/2) \alpha_\tau^{-2} \lambda^2 \} . \nn
\end{align}
This verifies Condition~($ED_2$) by taking $\omega = n^{-1/2}$, $\nu_0 = e^{1/2} \alpha_\tau^{-1}$ and $\mathbf{g}(r)  = \alpha_\tau \sqrt{n}$, $r>0$. Now, using Proposition~3.1 in \cite{S2013} we obtain that as long as $ \alpha_\tau^2 n \geq  4+2t$,
\begin{align}
	 & \sup_{\theta \in \Theta_0(r)} |   \alpha_\tau \sqrt{n} (\theta - \mu_j) +  n^{-1/2}\{  L'(\theta) - L'(\mu_j) \} |   \nn \\
	& \qquad \qquad \qquad  \leq  \alpha_\tau^{1/2} \delta(r) r  + 6 \alpha_\tau^{-1/2} e^{1/2} ( 2t + 4  )^{1/2}  n^{-1/2} r \nn
\end{align}
with probability greater than $1-e^{-t}$. Under the conditions that $n\geq  8t$ and $t\geq  1$, it is easy to see that $ \alpha_\tau^2 n  \geq  (7/8)^2 \cdot  8 t  \geq  6 t \geq  4 + 2t$. Moreover, observe that
$$
	  \sup_{\btheta \in \Theta_0(r)} | (\alpha_\tau - 1 )\sqrt{n} (\theta - \mu_j) |  \leq   \alpha_\tau^{-1/2}   \sigma_{jj} \tau^{-2} r  .
$$
The last two displays, together with \eqref{concentration.MLE} and the fact that $L'(\hat{\mu}_j)=0$ prove \eqref{Bahadur.representation} by taking $r=r_0$. The proof of Lemma~\ref{lemBR} is then complete.  \qed

\subsection{Proof of Lemma~\ref{app.meanvar}} \label{secC5}

Under model (1), we have $\xi_j = \bb_j^\T \bbf +  \varepsilon_j$, where $\e ( \varepsilon_j ) = 0$ and $\varepsilon_j$ and $\bbf$ are independent. Therefore,
\begin{align}
	&	\e_{\bbf} \psi_\tau(\xi_j)   - \bb_j^\T \bbf \nn \\
	& = - \e_{\bbf}  (\varepsilon_j + \bb_j^\T \bbf - \tau ) I( \varepsilon_j > \tau - \bb_j^\T \bbf ) + \e_{\bbf} (- \varepsilon_j - \bb_j^\T \bbf - \tau) I( \varepsilon_j < -\tau - \bb_j^\T \bbf  ). \nn
\end{align}
Therefore, as long as $\tau > |\bb_j^\T \bbf|$, we have for any $q \in [2 , \kappa]$ that
\begin{align}
	  | \e_{\bbf} \psi_\tau(\xi_j)   - \bb_j^\T \bbf | \leq \e_{\bbf} \{ | \varepsilon_j | I(| \varepsilon_j | > \tau - |\bb_j^\T \bbf|  )  \} \leq   (\tau - |\bb_j^\T \bbf|  )^{1-q} \,\e ( |\varepsilon_j|^q )   \nn
\end{align}
almost surely. This proves \eqref{approxi.mean} by taking $q$ to be 2 and $\kappa$.

For the conditional variance, observe that
\# \label{var.dec}
	\e_{\bbf}  \{ \psi_{\tau}(\xi_j) - \bb_j^\T \bbf \}^2 = \var_{\bbf}  \{ \psi_{\tau}(\xi_j) \} + \{ \e_{\bbf} \psi_{\tau}(\xi_j) - \bb_j^\T \bbf  \}^2
\#
and that $\psi_{\tau}(\xi_j) - \bb_j^\T \bbf$ can be written as
\begin{align}
   \varepsilon_j    I( | \bb_j^\T \bbf + \varepsilon_j | \leq \tau  ) + (\tau - \bb_j^\T \bbf) I(\bb_j^\T \bbf +  \varepsilon_j > \tau  ) - ( \tau + \bb^\T_j \bbf) I(\bb_j^\T \bbf +  \varepsilon_j<-\tau ), \nn
\end{align}
which further implies
\begin{align}
	&  \{  \psi_{\tau}(\xi_j) - \bb_j^\T \bbf  \}^2 \nn \\
	& =  \varepsilon_j^2  I (| \bb_j^\T \bbf +  \varepsilon_j | \leq \tau ) + (\tau - \bb_j^\T \bbf)^2 I(\bb_j^\T \bbf +  \varepsilon_j > \tau  ) + ( \tau + \bb^\T_j \bbf)^2 I(\bb_j^\T \bbf +  \varepsilon_j<-\tau ) . \nn
\end{align}
Taking conditional expectation on both sides yields
\begin{align}
	 & \e_{\bbf} \{ \psi_{\tau}(\xi_j) - \bb_j^\T \bbf \}^2  \nn \\
	 & = \e(  \varepsilon_j^2 ) - \e_{\bbf}  \{ \varepsilon_j^2 I(|\bb_j^\T \bbf +  \varepsilon_j| >\tau  ) \}  \nn \\
	 & \quad   + (\tau - \bb_j^\T \bbf)^2 \PP_{\bbf} (   \varepsilon_j > \tau - \bb_j^\T \bbf  ) + (\tau + \bb_j^\T \bbf)^2 \PP_{\bbf} (  \varepsilon_j < - \tau - \bb_j^\T \bbf ).  \nn
\end{align}
Using the equality $u^2 = 2 \int_0^u t\, dt$ for $u>0$ we deduce that as long as $\tau > |\bb_j^\T \bbf|$,
\begin{align}
	 &  \e_{\bbf}  \varepsilon_j^2  I( \bb_j^{{\rm T}} \bbf + \varepsilon_j >\tau ) \nn \\
	&  = 2 \e_{\bbf} \int_0^\infty I(  \varepsilon_j > t) I(  \varepsilon_j > \tau - \bb_j^\T \bbf  ) t \, dt \nn \\
	& = 2 \e_{\bbf} \int_0^{\tau - \bb_j^{{\rm T}} \bbf}  I(  \varepsilon_j > \tau - \bb_j^\T \bbf  )t \, dt + 2 \e_{\bbf} \int_{\tau - \bb_j^{{\rm T}} \bbf}^\infty I(  \varepsilon_j > t)  t \, dt \nn \\
	& =  (\tau - \bb_j^\T \bbf)^2  \PP_{\bbf}(  \varepsilon_j > \tau - \bb_j^\T \bbf )    + 2 \int_{\tau - \bb_j^{{\rm T}} \bbf}^\infty \PP(  \varepsilon_j > t)  t \, dt  . \nn
\end{align}
Analogously, it can be shown that
\begin{align}
	 \e_{\bbf}  \{ \varepsilon_j^2 I( \bb_j^\T \bbf + \varepsilon_j <-\tau )  \}  = (\tau + \bb_j^\T \bbf)^2 \PP_{\bbf} ( \varepsilon_j < - \tau - \bb_j^{{\rm T}} \bbf ) + 2 \int_{\tau + \bb_j^{{\rm T}} \bbf}^\infty \PP(- \varepsilon_j >t )  t \, dt  . \nn
\end{align}
Together, the last three displays imply
\begin{align}
	  0  & \geq   \e_{\bbf} \{  \psi_{\tau}(\xi_j) - \bb_j^\T \bbf  \}^2 - \e ( \varepsilon_j^2) \nn \\
	& \geq  - 2 \int_{\tau - |\bb_j^{{\rm T}} \bbf|}^{\infty} \PP ( |\varepsilon_j | > t  ) t \, dt   \geq  - 2 \e ( | \varepsilon_j|^\kappa ) \int_{\tau - |\bb_j^{{\rm T}} \bbf|}^{\infty}  t^{1-\kappa} \, dt   = - \frac{2}{\kappa-2} \frac{ \e ( | \varepsilon_j|^\kappa ) }{( \tau - |\bb_j^{{\rm T}} \bbf | )^{ \kappa-2}}   . \nn
\end{align}
Combining this with \eqref{var.dec} and \eqref{approxi.mean} proves \eqref{approxi.var}.

Finally, we study the covariance $\cov_{\bbf} (\psi_\tau(\xi_j), \psi_\tau(\xi_k) ) $ for $j\neq k$. By definition,
\begin{align}
	&  \cov_{\bbf}  (  \psi_\tau(\xi_j),\psi_\tau(\xi_k) )  \nn \\
	& = \e_{\bbf} \{   \psi_\tau(\xi_j) -  \bb_j^\T \bbf  + \bb_j^\T \bbf - \e_{\bbf}  \psi_\tau(\xi_j)  \}    \{  \psi_\tau( \xi_k) -  \bb_k^\T \bbf  + \bb_k^\T \bbf - \e_{\bbf}  \psi_\tau( \xi_k) \} \nn \\
	& =   \underbrace{ \e_{\bbf}  \{ \psi_\tau( \xi_j) -  \bb_j^{{\rm T}} \bbf  \} \{ \psi_\tau( \xi_k) -  \bb_k^\T \bbf \} }_{\Pi_1}    -   \underbrace{ \{  \e_{\bbf}  \psi_\tau( \xi_j) -  \bb_j^{{\rm T}} \bbf  \}   \{ \e_{\bbf}  \psi_\tau( \xi_k)  -  \bb_k^\T \bbf     \} }_{\Pi_2} . \nn
\end{align}
Recall that $\psi_\tau( \xi_j) - \bb_j^\T \bbf=  \varepsilon_j I(| \xi_j| \leq \tau ) + (\tau - \bb_j^\T \bbf) I(\xi_j >\tau) - (\tau+ \bb_j^\T \bbf) I( \xi_j < -\tau)$. Hence,
\begin{align}
	\Pi_1  & = \e_{\bbf}  \varepsilon_j \varepsilon_k  I( |\xi_j| \leq \tau, |\xi_k |\leq \tau ) +(\tau - \bb_k^\T \bbf)  \e_{\bbf} \varepsilon_j  I( |\xi_j | \leq \tau , \xi_k >\tau ) \nn \\
	& \quad - (\tau + \bb_k^\T \bbf) \e_{\bbf}  \varepsilon_j I( |\xi_j | \leq \tau, \xi_k < -\tau ) + (\tau - \bb_j^\T \bbf) \e_{\bbf} \varepsilon_k I( \xi_j > \tau, |\xi_k | \leq \tau ) \nn \\
	& \quad + (\tau - \bb_j^\T \bbf)(\tau - \bb_k^\T \bbf) \e_{\bbf}	I(  \xi_j   > \tau, \xi_k > \tau ) - ( \tau - \bb_j^\T \bbf )(\tau  + \bb_k^\T \bbf)  \e_{\bbf} I(  \xi_j   > \tau, \xi_k < - \tau  ) \nn \\
	& \quad - (\tau  + \bb_j^\T \bbf )\e_{\bbf} \varepsilon_k  I(  \xi_j < -\tau,  |\xi_k | \leq \tau ) - (\tau+ \bb_j^\T \bbf)( \tau - \bb_k^\T \bbf ) \e_{\bbf}  I( \xi_j < -\tau, \xi_k >\tau ) \nn \\
	& \quad + (\tau + \bb_j^\T \bbf)(\tau+ \bb_k^\T \bbf) \e_{\bbf} I(  \xi_j <-\tau, \xi_k < -\tau ).  \label{cov.I.dec}
\end{align}
Note that the first term on the right-hand side of \eqref{cov.I.dec} can be written as
\begin{align}
	 & \e_{\bbf} \varepsilon_j \varepsilon_k  I( |\xi_j| \leq \tau, |\xi_k |\leq \tau )    \nn \\
	 & = \cov( \varepsilon_j, \varepsilon_k ) -  \e_{\bbf} \varepsilon_j \varepsilon_k I( |\xi_j| > \tau  ) - \e_{\bbf} \varepsilon_j \varepsilon_k  I(  |\xi_k | > \tau ) + \e_{\bbf} \varepsilon_j \varepsilon_k  I( |\xi_j| > \tau, |\xi_k | > \tau ) , \nn
\end{align}
where
\begin{align}
	|  \e_{\bbf} \varepsilon_j \varepsilon_k I( |\xi_j| > \tau )  |  \leq  | \tau - \bb_j^\T \bbf |^{2-\kappa} \e (  | \varepsilon_j |^{\kappa-1} | \varepsilon_k| ) \leq  2^{\kappa-2} \tau^{2-\kappa}  ( \e | \varepsilon_j|^\kappa  )^{(\kappa-1)/\kappa}  ( \e | \varepsilon_k|^\kappa  )^{1/\kappa}  \nn
\end{align}
and
\begin{align}	
	& | \e_{\bbf} \varepsilon_j \varepsilon_k  I( |\xi_j| > \tau, |\xi_k | > \tau ) | \nn \\
	&  \leq | \tau - \bb_j^\T \bbf |^{2-\kappa} \e  ( |\varepsilon_j|^{\kappa/2} |\varepsilon_k|^{\kappa/2}   ) \leq  2^{\kappa-2} \tau^{2-\kappa}   ( \e | \varepsilon_j |^\kappa  )^{1/2}  ( \e | \varepsilon_k |^\kappa  )^{1/2} \nn
\end{align}
almost surely on $\mathcal{G}_{jk}$. The previous three displays together imply
\begin{align}
 | \e_{\bbf} \varepsilon_j \varepsilon_k I( |\xi_j| \leq \tau, |\xi_k |\leq \tau  )  - \cov(\varepsilon_j, \varepsilon_k )  | \lesssim \tau^{2-\kappa}  \nn
\end{align}
almost surely on $ \mathcal{G}_{jk}$. For the remaining terms on the right-hand side of \eqref{cov.I.dec}, it can be similarly obtained that, almost surely on $\mathcal{G}_{jk}$,
\begin{align}
	  |  \e_{\bbf} \varepsilon_j  I( |\xi_j | \leq \tau , \xi_k >\tau  )   | \leq |  \tau - \bb_k^\T \bbf |^{1-\kappa } \e ( | \varepsilon_j|  | \varepsilon_k|^{\kappa-1}  ) , \nn \\
 |  \e_{\bbf} \varepsilon_j I( |\xi_j | \leq \tau , \xi_k < - \tau ) |  \leq  |  \tau +  \bb_k^\T \bbf |^{1-\kappa} \e ( | \varepsilon_j| | \varepsilon_k |^{\kappa-1}  ),  \nn \\
	\mbox{ and }~\e_{\bbf}  I( \xi_j > \tau, \xi_k < -\tau  ) \leq  | \tau - \bb_j^\T \bbf  |^{-\kappa/2}   | \tau  + \bb_k^\T \bbf  |^{-\kappa/2}   \e ( | \varepsilon_j \varepsilon_k|^{\kappa/2}  ) . \nn
\end{align}

Putting together the pieces, we get $| \Pi_1 -\cov(\varepsilon_j, \varepsilon_k )  | \lesssim \upsilon_{jk}\tau^{2-\kappa}$ almost surely on $\mathcal{G}_{jk}$. For $\Pi_2$, it follows directly from \eqref{approxi.mean} that $|\Pi_2| \lesssim \upsilon_{jk}^2\tau^{ 2 -2 \kappa } $ almost surely on $\mathcal{G}_{jk}$. These bounds, combined with the fact that $ \cov_{\bbf} ( \psi_\tau(\xi_j), \psi_\tau(\xi_k) ) =\Pi_1 - \Pi_2$, yield \eqref{approxi.cov}. \qed

\subsection{Proof of Lemma~\ref{factor.concentration}} \label{secC6}

For any $\bu \in \RR^K$, by independence we have
\begin{align}
	  \e \exp(\bu^\T \bbf_i ) \leq \exp(C_1 \| \bbf\|_{\psi_2}^2 \| \bu \|_2^2 ) ~\mbox{ for all } i=1,\ldots, n, \\ \nn
	 \mbox{and }~ \e \exp( \sqrt{n}\,\bu^\T  \bar{\bbf} ) = \prod_{i=1}^n \e \exp  (\bu^\T \bbf_i / \sqrt{n}  )  \leq  \exp ( C_1 \|\bbf\|_{\psi_2}^2 \| \bu \|_2^2  ), \nn
\end{align}
where $C_1>0$ is an absolute constant. From Theorem~2.1 in \cite{HKZ2012} we see that, for any $t>0$,
\begin{align}
	 \PP \{ \|  \sqrt{n} \bar{\bbf} \|_2^2 >  2C_1\| \bbf \|_{\psi_2}^2 (  K + 2\sqrt{ K  t} + 2t  )  \} \leq e^{-t} \nn \\
	 \mbox{and }~ \PP\{ \| \bbf_i \|_2^2 > 2C_1 \| \bbf \|_{\psi_2}^2(K+2\sqrt{Kt} + 2t) \} \leq e^{-t}, \ \ i=1,\ldots, n.\nn
\end{align}
This proves \eqref{factor.concentration.ineq} and \eqref{factor.l2.ineq} by the union bound.

For $\hat{\bSigma}_f$, applying Theorem~5.39 in \cite{V2012} yields that, with probability at least $1-2 e^{-t}$, $\| \hat{\bSigma}_f - \Ib_K \| \leq \max(\delta, \delta^2)$, where $\delta = C_2 \| \bbf\|^2_{\psi_2} n^{-1/2}(K+t)^{1/2}$ and $C_2 > 0$ is an absolute constant. Conclusion \eqref{factor.cov.concentration} then follows immediately. \qed

\subsection{Proof of Lemma~\ref{infinity.perturbation}} \label{secC7}

For each $1\leq \ell \leq K$, as $\overline{\lambda}_\ell >0$ and by Weyl's inequality, we have $| \wt \lambda_\ell - \overline \lambda_\ell | \leq | \lambda_\ell( \hat{\bSigma}_{{\rm H}} ) - \overline{\lambda}_\ell | \leq \|  \hat{\bSigma}_{{\rm H}} - \bSigma \| + \| \bSigma_\varepsilon \|$. Moreover, note that for any matrix $\Eb \in \RR^{d_1 \times d_2}$,
$$
	\| \Eb \|^2  \leq \| \Eb \|_{1}  \vee \| \Eb \|_{ \infty}  \leq  (d_1 \vee d_2 )\| \Eb \|_{\max}.
$$
Putting the above calculations together proves \eqref{eigenvalue.perturb}.

Next, note that
$$
 \hat{\bSigma}_{{\rm H}}  =  \hat{\bSigma}_{{\rm H}}  - \bSigma +  \Bb \Bb^\T +  \bSigma_{\varepsilon} =  \sum_{\ell=1}^K \overline{\lambda}_\ell \overline{\bv}_\ell\overline{\bv}_\ell^\T +  \hat{\bSigma}_{{\rm H}}  - \bSigma +  \bSigma_{\varepsilon}  .
$$
Under Assumption~2, it follows from Theorem~3 and Proposition~3 in \cite{FWZ2016} that
$$
	\max_{1\leq \ell \leq K} \| \hat{\bv}_\ell - \overline{\bv}_\ell \|_\infty \leq  \frac{C}{ p^{3/2}}   ( \|   \hat{\bSigma}_{{\rm H}}  - \bSigma\|_\infty + \| \bSigma_\varepsilon \|_\infty  ) \leq  C (  p^{-1/2}  \| \hat{\bSigma}_{{\rm H}} - \bSigma\|_{\max}  + p^{-1}  \| \bSigma_\varepsilon \|   ),
$$
where we use the inequalities $\|   \hat{\bSigma}_{{\rm H}}  - \bSigma\|_\infty \leq p  \|   \hat{\bSigma}_{{\rm H}}  - \bSigma\|_{\max}$ and $\| \bSigma_\varepsilon \|_\infty \leq p^{1/2} \| \bSigma_\varepsilon \|$ in the last step and $C >0$ is a constant independent of $(n,p)$. This proves \eqref{eigenvector.perturb} . \qed

\section{Additional numerical results on FDP/FDR control}
\label{sec:B}

In the end, we present some additional simulation results that complement Section 4.5.
Under Models 2 and 3 defined in Section 4.2,  we compare the numerical performance of the five tests regarding FDP/FDR control.
We take $\alpha=0.05$, $p=500$ and let $n$ gradually increase from 100 to 200.
The empirical FDP is defined as the average false discovery proportion based on 200 simulations. The simulation results are presented in Figures \ref{Sim_fig_FDP_2} and \ref{Sim_fig_FDP_3}, respectively.

\begin{figure}[htbp]
 \centering
 \includegraphics[scale=0.4]{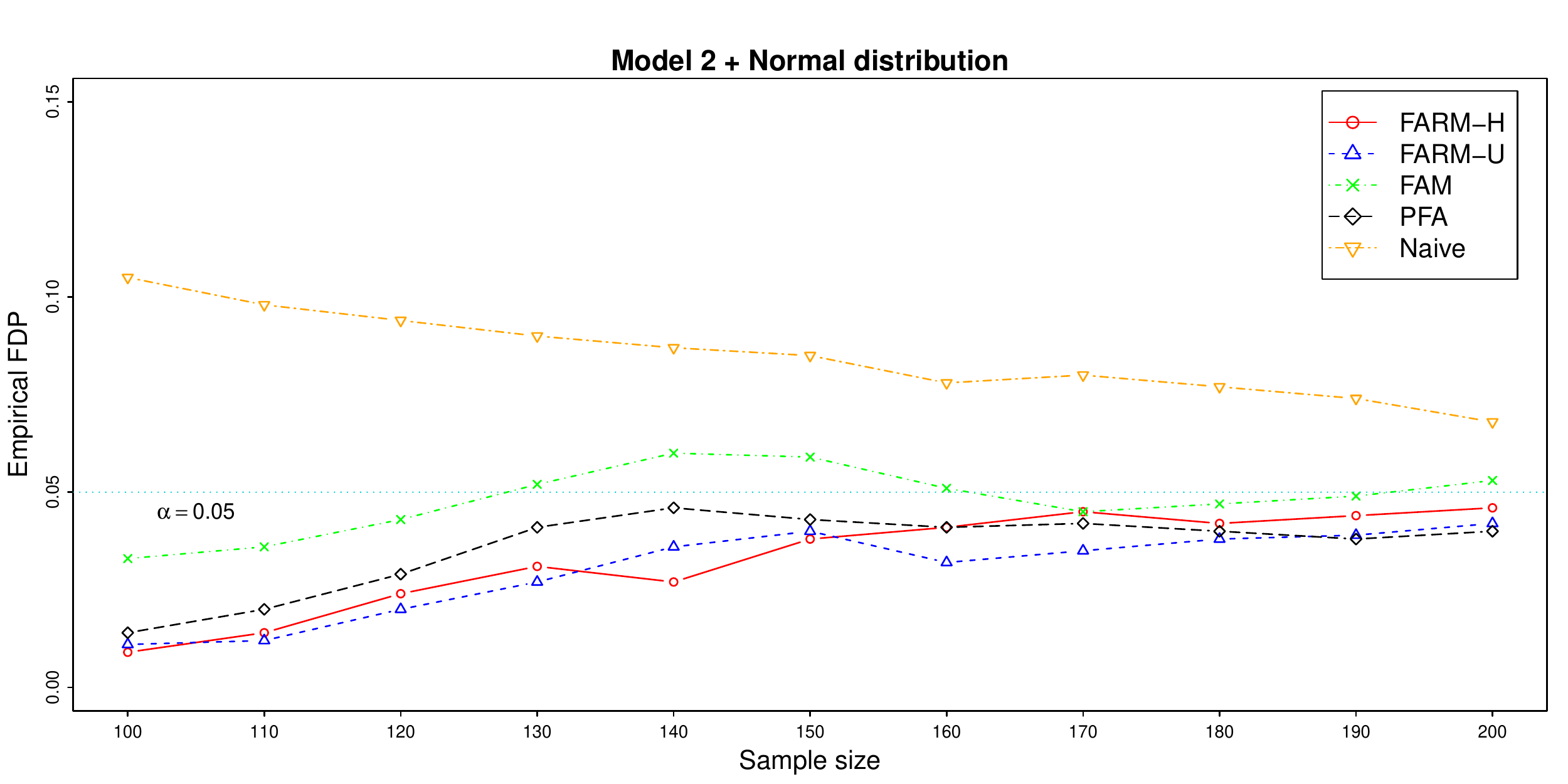}
 \includegraphics[scale=0.4]{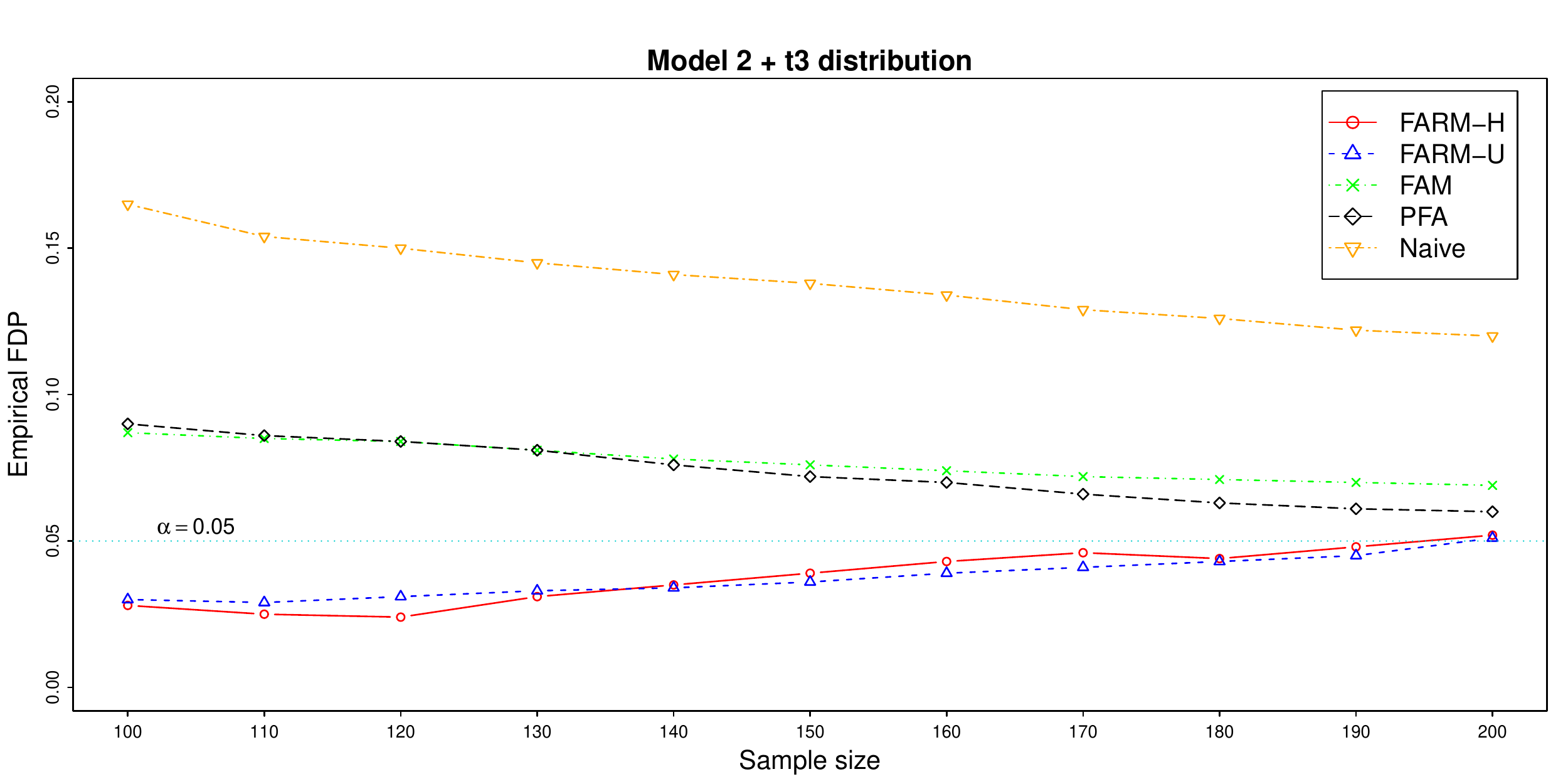}
 \includegraphics[scale=0.4]{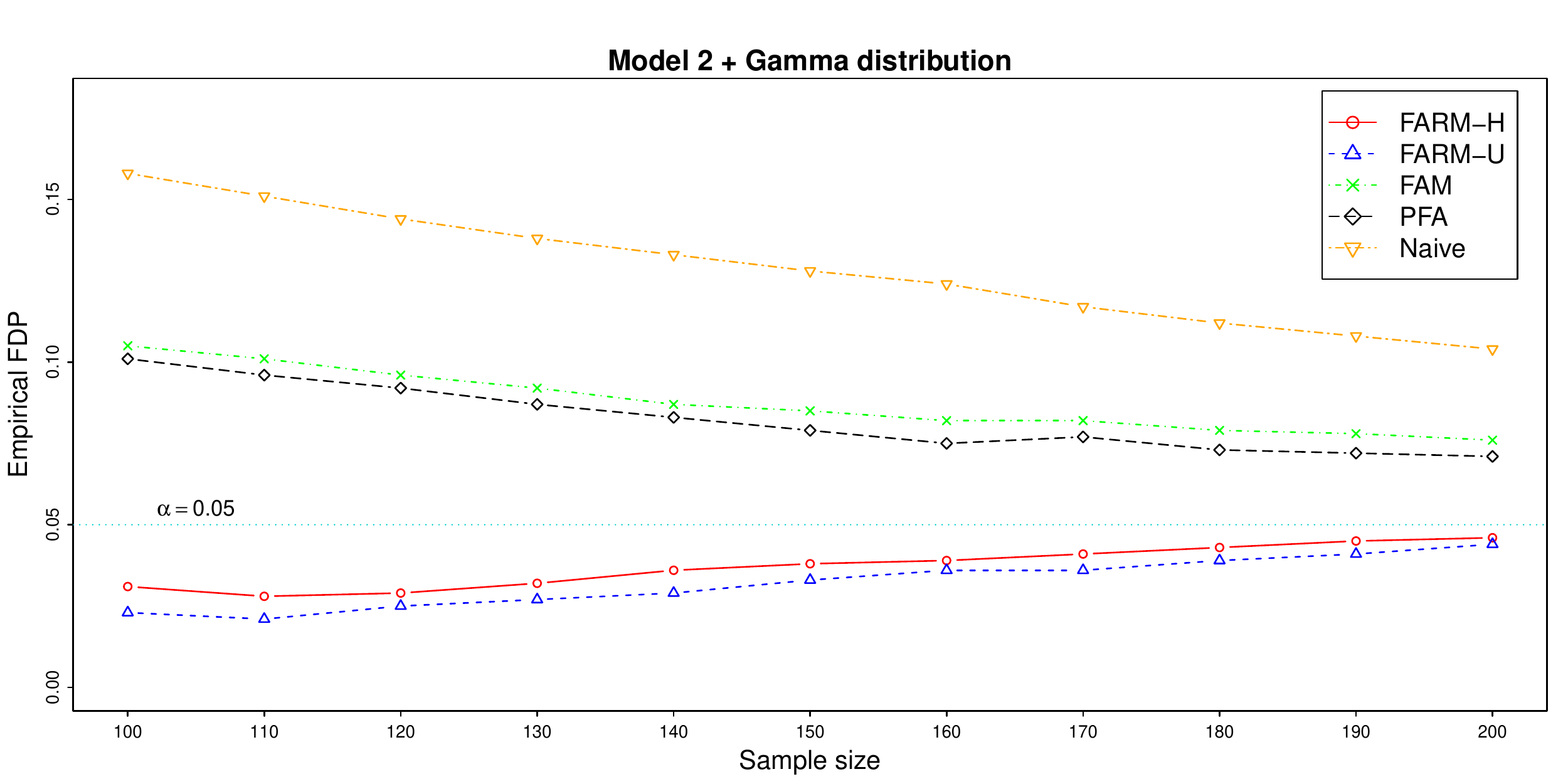}
 \includegraphics[scale=0.4]{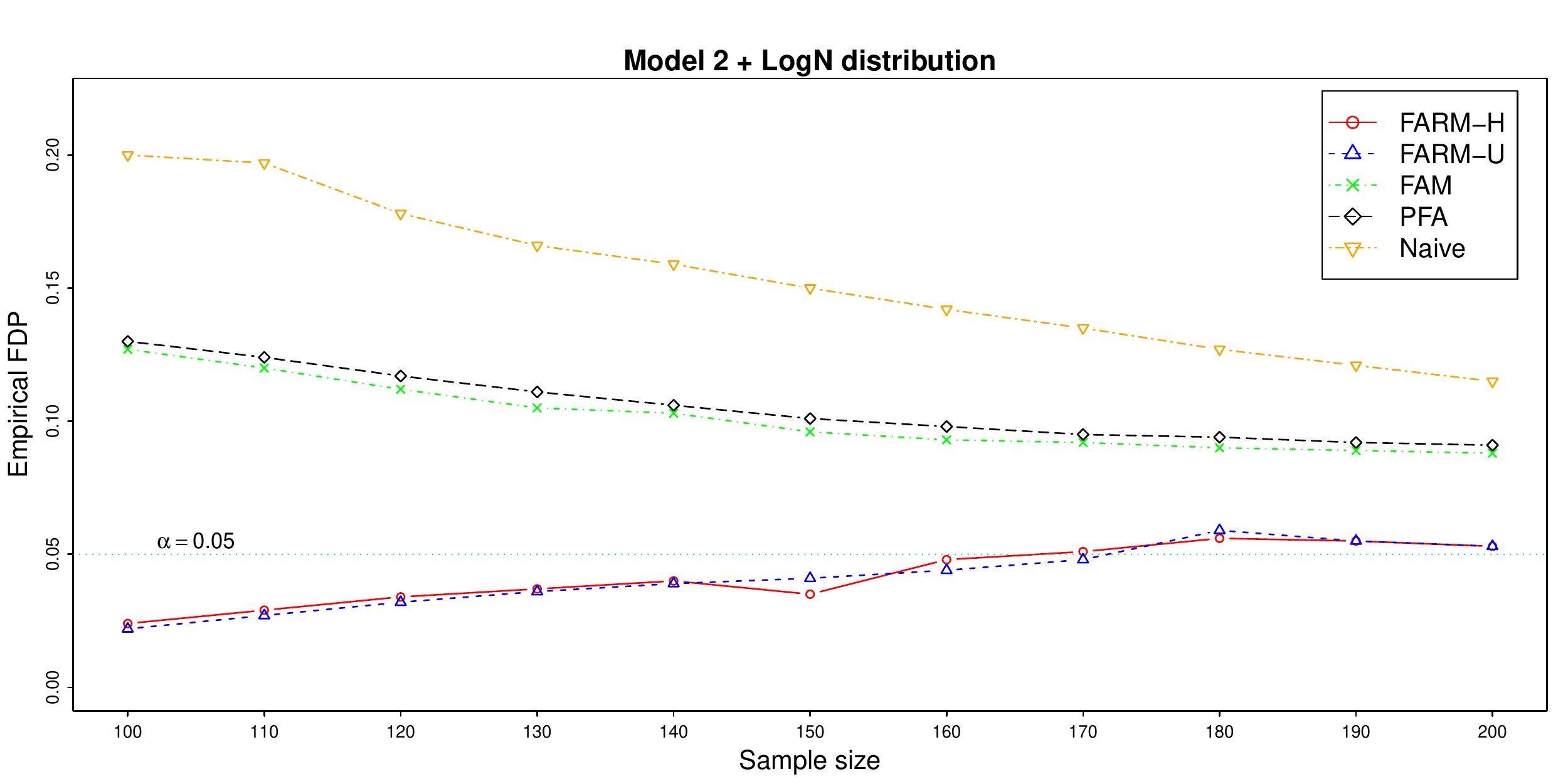}
  \begin{singlespace}
  \caption{Empirical FDP versus sample size for the five tests at level $\alpha=0.05$. The data are generated from Model 2 with $p=500$ and sample size $n$ ranging from 100 to 200 with a step size of 10. The panels from top to bottom correspond to the four error distributions in Section 4.2. } \label{Sim_fig_FDP_2}
  \end{singlespace}
\end{figure}

\begin{figure}[htbp]
 \centering
 \includegraphics[scale=0.4]{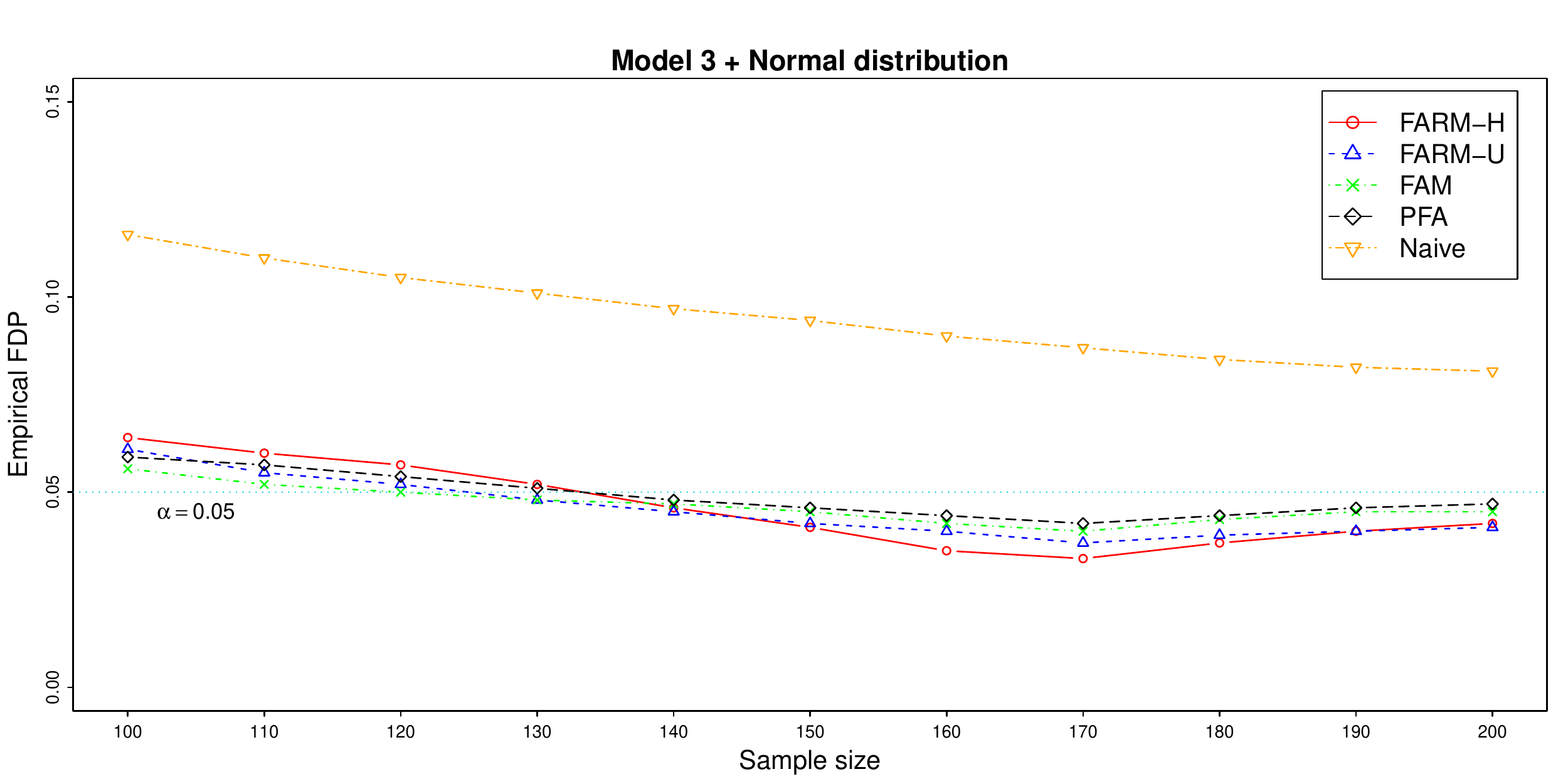}
 \includegraphics[scale=0.4]{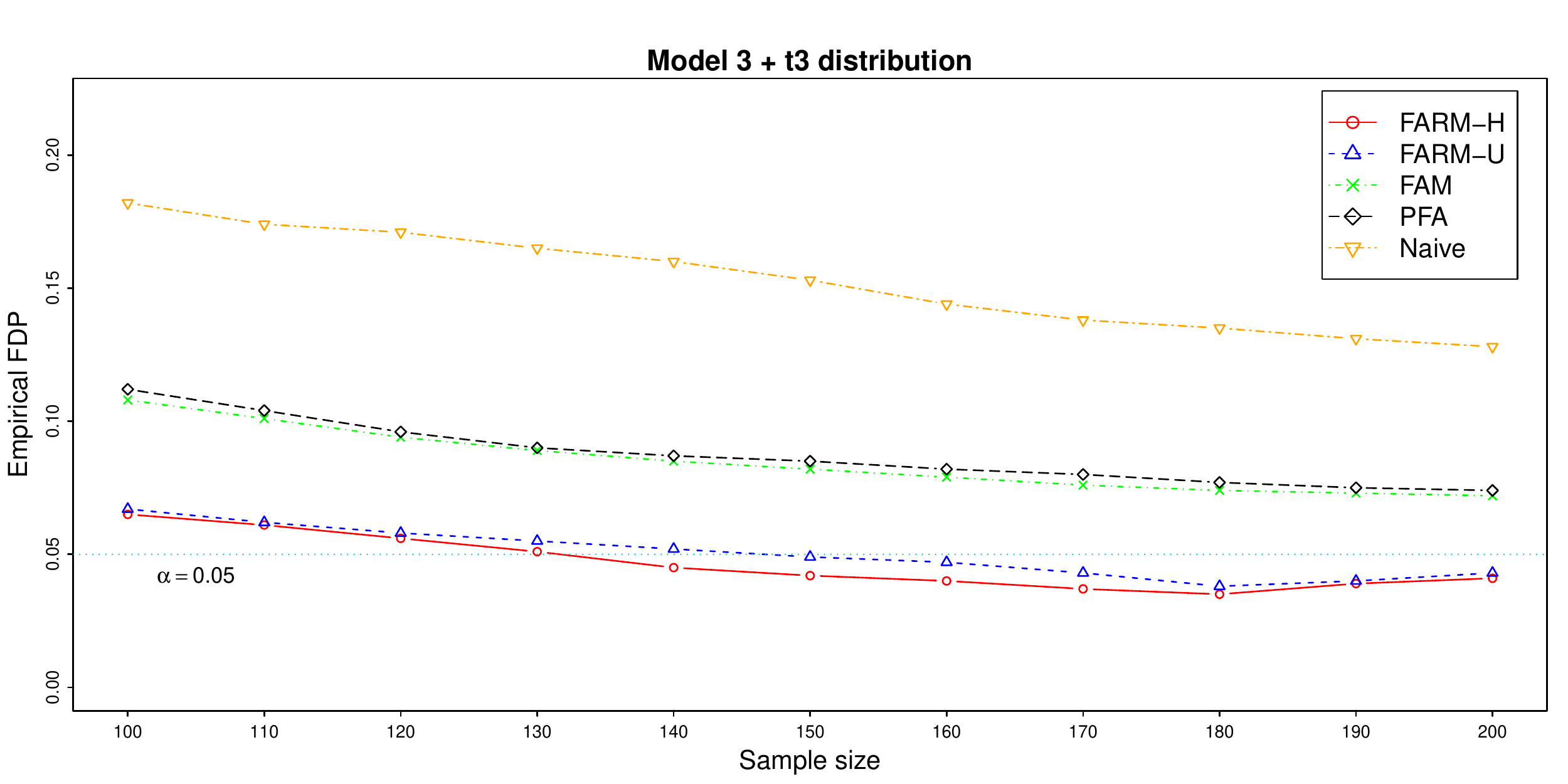}
 \includegraphics[scale=0.4]{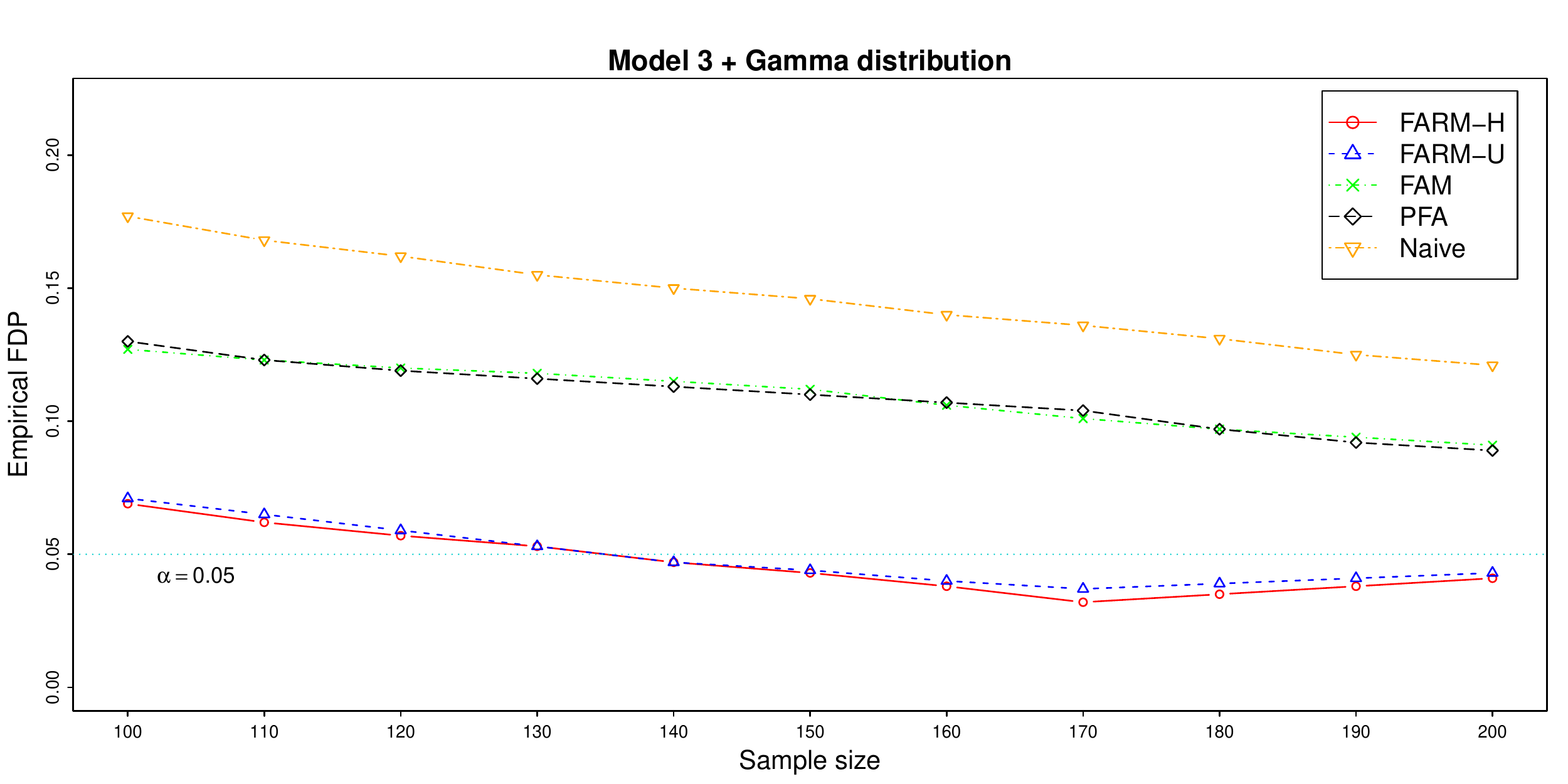}
 \includegraphics[scale=0.4]{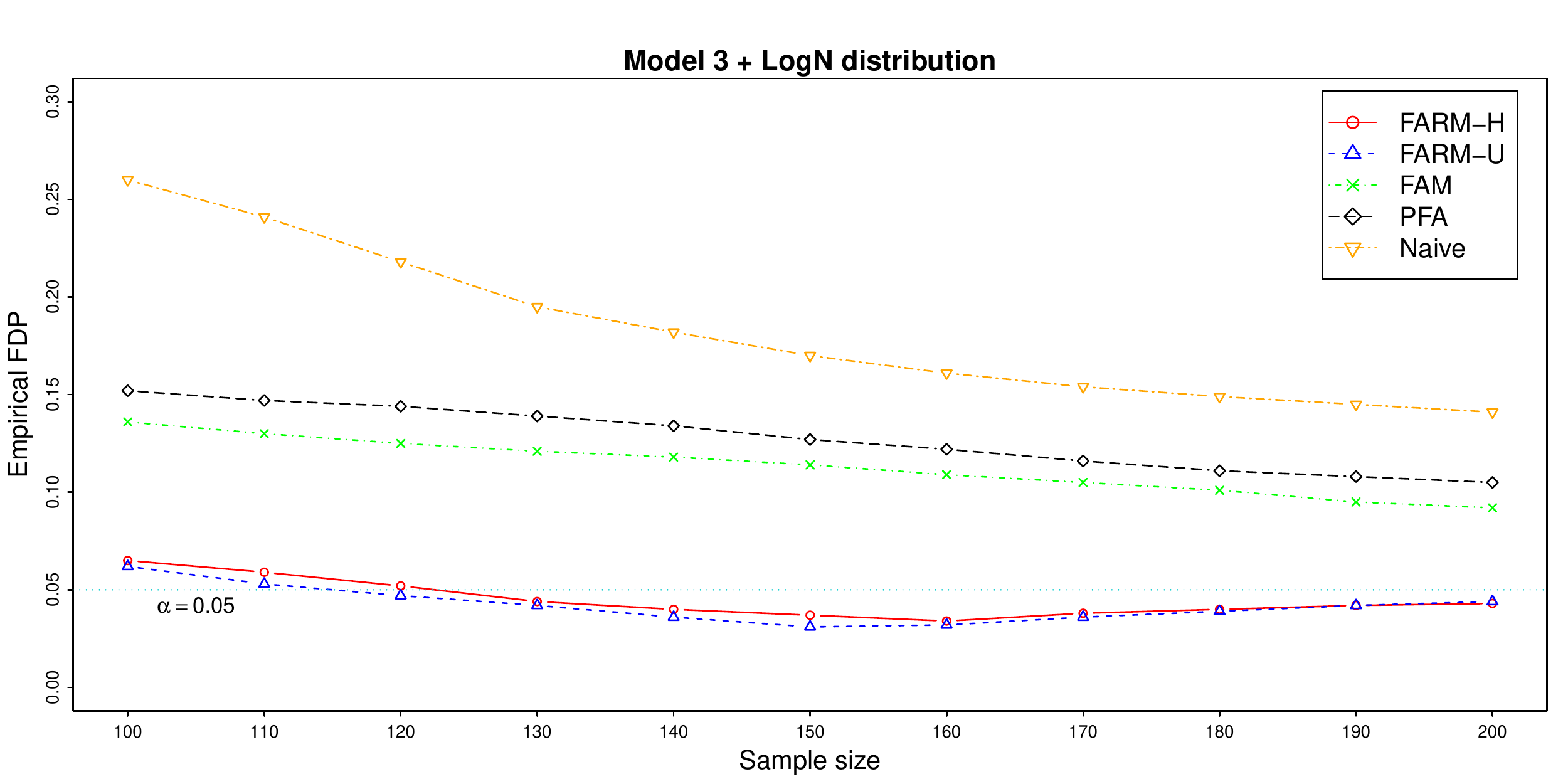}
  \begin{singlespace}
  \caption{Empirical FDP versus sample size for the five tests at level $\alpha=0.05$. The data are generated from Model 3 with $p=500$ and sample size $n$ ranging from 100 to 200 with a step size of 10. The panels from top to bottom correspond to the four error distributions in Section 4.2. } \label{Sim_fig_FDP_3}
  \end{singlespace}
\end{figure}

\end{document}